\documentclass[11pt]{article}

\usepackage[letterpaper, margin=1in]{geometry}

\usepackage[utf8]{inputenc}
\usepackage[T1]{fontenc}
\usepackage[colorlinks=true, linkcolor=red, urlcolor=blue, citecolor=gray]{hyperref}
\usepackage{url}
\usepackage{booktabs}
\usepackage{amsfonts}
\usepackage{nicefrac}
\usepackage{microtype}
\usepackage{xcolor}

\usepackage{amsmath,amsthm,amssymb,url}
\usepackage{thmtools, thm-restate}
\usepackage{algorithm,algpseudocode}
\usepackage{mathtools}
\usepackage{graphicx}
\usepackage{caption}
\usepackage{subcaption}
\usepackage{tikz}
\usetikzlibrary{calc,arrows.meta,positioning}
\usetikzlibrary{decorations.pathreplacing}
\usetikzlibrary{arrows}

\newtheorem{theorem}{Theorem}
\newtheorem{property}{Property}

\newtheorem{lemma}{Lemma}
\newtheorem{implemma}{Imported Lemma}
\newtheorem{imptheo}{Imported Theorem}
\newtheorem{corollary}{Corollary}
\newtheorem{definition}{Definition}[section]
\newtheorem{rthm}{Theorem}
\newtheorem{rcrlry}{Corollary}

\newenvironment{reptheorem}[1]
 {\renewcommand\therthm{\ref*{#1}}\rthm}
 {\endrthm}
\newenvironment{repcorollary}[1]
 {\renewcommand\thercrlry{\ref*{#1}}\rcrlry}
 {\endrcrlry}

\hypersetup{
 colorlinks=true,
 linkcolor=blue,
 citecolor=red,
 urlcolor=blue,
 linktocpage
}
\definecolor{ForestGreen}{RGB}{34,139,34}

\let\origref\ref
\renewcommand{\ref}[1]{(\origref{#1})}

\title{Sublinear Time Eigenvector Approximation via Column Sampling}

\author{
Rajarshi Bhattacharjee\footnote{Manning College of Information and Computer Sciences, University of Massachusetts, Amherst, \texttt{\{rbhattacharj, cmusco\}@cs.umass.edu}}
\and
Cameron Musco\footnotemark[1]
\and
Dominic Rutkowski\footnote{Databricks, \texttt{rutkowski.dominic@gmail.com}}
}
\date{}

  \usepackage{nth}
  \usepackage{intcalc}

\definecolor{darkgreen}{HTML}{1b7837}

\newcommand{\bv}[1]{\mathbf{#1}}

\DeclareUnicodeCharacter{2212}{-}

\newcommand{\R}{\mathbb{R}}

\newcommand{\E}{\mathbb{E}}
\newcommand{\poly}{\mathop\mathrm{poly}}

\newcommand{\Sbb}{\bv{\bar S}}

\newcommand{\Sb}{\mathbf{S}}
\newcommand{\Ab}{\mathbf{A}}

\newcommand{\xb}{\mathbf{x}}

\DeclareMathOperator{\tr}{tr}
\newcommand{\norm}[1]{\|#1\|}

\begin{document}

\begin{titlepage}
\maketitle
 \thispagestyle{empty}

\begin{abstract}
We study sublinear time sampling methods for approximating the outlying eigenvectors of large matrices. Our main result is an algorithm that uniformly samples just  $\Tilde{O}\left(\frac{\log n}{\epsilon^4} \right)$ columns of a symmetric matrix $\bv A \in \R^{n \times n}$ scaled to have entries bounded in magnitude by $1$, and, for any eigenvalue $\lambda$ of $\bv A$ with $|\lambda| \ge \epsilon n$, outputs an approximate eigenvector $\bv v$ satisfying $\|\bv A \bv v-\lambda \bv v \|_2 
  \leq \epsilon n$. For approximating only the eigenvector corresponding to the largest magnitude eigenvalue, our algorithm samples just $\Tilde{O}\left(\frac{\log n}{\epsilon^2} \right)$ columns. Given the ability to sample rows and columns of $\bv A$ according to their squared norms, we give a similar result, but with an improved error bound of $\epsilon \norm{\bv A}_F$. For top eigenvector approximation, we show that our bound is tight up to logarithmic terms, even for potentially adaptive algorithms. All our algorithms are non-adaptive.

  A key feature of our algorithms is that they output approximate eigenvectors that are spanned by just a small number of $\bv A$'s columns, and whose individual entries can be computed rapidly, in just $\poly(\log n,1/\epsilon)$ time per entry. This makes them applicable in the ``quantum-inspired algorithms'' framework of Tang \cite{tang2019quantum}, where we give the first sublinear time classical algorithms for eigenvector approximation with additive error $\epsilon \norm{\bv A}_F$. 
  
  Finally, we present an alternative approach, based on a truncated Nystr\"{o}m method, that, while not allowing $\poly(\log n,1/\epsilon)$ time entrywise computation of the approximate eigenvectors, achieves near optimal sample complexity for general symmetric matrices, and improved bounds for positive semidefinite matrices.
  %
  %, while our lower bounds hold against even adaptive algorithms. %Finally, assuming the ability to sample columns with probabilities proportional to their squared Euclidean norms, we give an improved error bound of $\epsilon \norm{\bv A}_F$.

  Technically, our bounds build on recent work on approximating the outlying \emph{eigenvalues} of symmetric matrices via random sampling~\cite{Bhattacharjee:2021wl,swartworth2025tight}. We demonstrate for the first time that these general approaches extend to the important problem of eigenvector estimation.
\end{abstract}
\end{titlepage}

\tableofcontents
\thispagestyle{empty}
\clearpage

\section{Introduction}

Computing the eigenvectors of large matrices is a key primitive in numerical linear algebra,
underpinning e.g., principal component analysis, spectral clustering,
graph-based learning, recommendation systems, and vibration analysis in engineering. For an $n \times n$ matrix,
all eigenvalues and eigenvectors can be computed to
high accuracy using direct eigendecomposition in $O(n^{\omega})$ time,
where $\omega \approx 2.37$ is the exponent of
matrix multiplication
\cite{demmel2007fast,alman2021refined}.
For computing just a few eigenvectors corresponding to the
largest magnitude eigenvalues to good accuracy,
 iterative methods such as the power method and
other Krylov subspace methods can be
used~\cite{saad2011numerical, golub2013matrix}.
These methods repeatedly multiply the input matrix by
query vectors, requiring $O(n^2)$ time per multiplication
for dense matrices.
Even this linear runtime can be prohibitive for large
$n$~\cite{martinsson2020randomized}.

This motivates the study of \emph{sublinear time algorithms},
which read only $o(n^2)$ entries of the input matrix. Significant work has studied sublinear time algorithms for various
matrix problems under structural assumptions, such as the input matrix being Toeplitz, Hankel, or positive semidefinite~\cite{clarkson2017low,
 musco2017sublinear, bakshi2020robust,musco2024sublinear,kapralov2023toeplitz,bakshi2018sublinear, indyk2019sample,backurs2021faster}.
A few works have studied sublinear time algorithms for
general matrices without strong structural assumptions.
For example,~\cite{BakshiChepurkoJayaram:2020,Bhattacharjee:2021wl,swartworth2025tight}
study the closely related problems of eigenvalue estimation and positive semidefiniteness testing for
symmetric matrices.
\cite{swartworth2025tight} also present a sublinear time algorithm
for approximating the top eigenvector of a symmetric positive semidefinite (PSD)
 matrix.

In general however, there have been no results on sublinear time eigenvector approximation algorithms for general symmetric matrices.

\subsection{Our contributions}

In this paper, we address this gap, presenting sublinear time algorithms that approximate all the eigenvectors of a symmetric input matrix that correspond to sufficiently large magnitude outlying eigenvalues. In particular, we tackle the following problem:
\begin{center}
\emph{Given a symmetric matrix $\bv A \in \R^{n \times n}$,
for every large magnitude outlying eigenvalue $\lambda$,
output an approximate eigenvector $\bv v$
such that $\|\bv A \bv v - \lambda \bv v\|_2$ is small.}
\end{center}
Here, $\| \cdot\|_2$ is the Euclidean norm.
Note that if $(\bv v,\lambda)$ is an actual eigenpair of $\bv A$, then $\bv A \bv v = \lambda \bv v$, so
the ``backward residual error'' $\|\bv A \bv v - \lambda \bv v\|_2$ is exactly zero.
As is standard in work on eigenvector approximation,
we focus on the residual error instead of a ``forward'' error
metric like $\|\bv u-\bv v \|_2$ where $\bv u$ is a true
eigenvector. Such a forward error bound becomes impossible to guarantee when there are small eigenvalue gaps. In contrast, small residual error can typically be achieved
independent of the eigengaps. Further, when there are eigengaps, a residual error bound can be easily converted
to a forward error bound by Davis-Kahan type
perturbation bounds~\cite{davis1970rotation, wedin1972perturbation}.

\paragraph{Algorithmic approach.}
Our algorithms all follow the same high-level approach:
let $\Sbb \in \mathbb{R}^{n \times s}$ be
a sampling matrix with a single non-zero entry per column, so that $\bv A \Sbb$ samples and potentially rescales $s$
columns from $\bv A$.
 Our algorithms output approximate eigenvectors spanned by just the $s$ sampled columns -- i.e., vectors of the form $\bv v = \bv A \Sbb \bv x$ for
$\bv x \in \mathbb{R}^s$.
As we will see, $s= \textrm{poly}(\log n, 1/\epsilon)$
where $\epsilon \in (0,1)$ is the desired accuracy will suffice. Thus, computing $\bv v$ will require just $\tilde O(n/\poly(\epsilon))$ time.
The key question is how to compute the coefficients of the linear combination, $\bv x \in \R^s$.

Building on prior work that approximates the \emph{eigenvalues} of $\bv A$ using those of the $s \times s$ random principal submatrix $\Sbb^T \bv A \Sbb$ \cite{Bhattacharjee:2021wl,swartworth2025tight}, our main approach will take $\bv x $ to be an eigenvector of $\Sbb^T \bv A \Sbb$. We prove that `lifting' this eigenvector to $\bv A$ via multiplication with $\bv A \Sbb$ yields an approximate eigenvector for $\bv A$. We can compute $\bv x$ in just $O(s^\omega) = \poly(\log n, 1/\epsilon)$ time via eigendecomposition of $\Sbb^T \bv A \Sbb$. Additionally, with $\bv x \in \R^s$ in hand, individual entries of $\bv v$ can be computed in just $O(s) = \poly(\log n, 1/\epsilon)$ time, by taking the inner product of $\bv x$ with the appropriate row of $\bv A \Sbb$. I.e., we can query individual entries of our approximate eigenvectors in just $\poly(\log n, 1/\epsilon)$ time.

We also present a related class of algorithms, where we let $\bv A \Sbb \bv x$ be an exact eigenvector of a \emph{rank-truncated Nystr\"{o}m} approximation to $\bv A$~\cite{clarkson2009numerical,nakatsukasa2020fast}. In this case, $\bv x$ can be obtained in $\tilde O(n/\poly(\epsilon))$ time by solving an appropriate generalized eigenvector problem involving the $s \times s$ matrices $\Sbb^T \bv A \Sbb$ and $\Sbb^T \bv A^2 \Sbb $. We show that this approach can achieve improved $\epsilon$ dependencies over our first approach, at the cost of sacrificing $ \poly(\log n, 1/\epsilon)$ runtime for computing individual entries of the approximate eigenvector. We note that a similar idea was used by \cite{swartworth2025tight} to approximate the top eigenvector of a PSD matrix -- we significantly extend the approach to non-PSD matrices and to approximating all large magnitude eigenvectors.

\paragraph{Summary of Contributions.}

We detail our main contributions below.

\begin{enumerate}
 \item \textbf{Random Submatrix-Based Algorithms.} For symmetric $\bv{A} \in \R^{n \times n}$ with
 entries bounded in magnitude by 1, we show that letting $\bv A \Sbb$ contain $\widetilde{O}(\log n/\epsilon^4)$ uniformly sampled columns of $\bv A$, and letting $\bv x$ be an eigenvector of $\Sbb^T \bv A \Sbb$,
 suffices to produce
 approximate eigenvectors $\bv{v}$ with
 $\|\bv{A}\bv{v} - \lambda\bv{v}\|_2 \leq \epsilon n$
 for all eigenvalues $\lambda$ with $|\lambda| \geq \epsilon n$,
 of which there are at most $1/\epsilon^2$. For approximating
 just the top eigenvector, we give a stronger bound of
 $\widetilde{O}(1/\epsilon^2)$ columns, which we show is tight in terms of sample complexity, up to logarithmic factors, even for general, possibly adaptive algorithms (our algorithm is non-adaptive).

 Note that the assumption that the entries of $\bv{A}$
 are bounded in magnitude by 1 is just for convenience in stating our bounds -- for general $\bv A$ our error bound is
 $\|\bv{A}\bv{v} - \lambda\bv{v}\|_2 \leq \epsilon n \|\bv{A}\|_{\infty}$,
 where $\| \bv{A} \|_{\infty} = \max_{i,j} |\bv A_{ij}|$. Our algorithms do not need knowledge of $\norm{\bv A}_\infty$. Also note that some additive error of this form is necessary. Consider $\bv A$ which is all zero except with $\bv A_{ij} = \bv A_{ji}$ very large for a random pair $(i,j)$. In sublinear time, we cannot hope to identify this pair with good probability, and must incur error depending on $|\bv A_{ij}|$.

 Similar to prior work on eigenvalue approximation \cite{Bhattacharjee:2021wl,swartworth2025tight}, we give stronger error bounds under the assumption that we can sample rows/columns of $\bv A$
 according to their squared Euclidean norms. In this setting, we show that we can approximate all outlying eigenvectors
 with $|\lambda| \geq \epsilon \|\bv{A}\|_F$
 to error $\epsilon\|\bv{A}\|_F$
 by sampling
 $\widetilde{O}(\log^4 n/\epsilon^4)$ columns.
 For approximating just the top eigenvector,
 we show that
 $\widetilde{O}(\log^4 n/\epsilon^2)$ columns suffice.
 Note that the error bound of $\epsilon\|\bv{A}\|_F$
 is always at least as strong as the $\epsilon n \|\bv{A}\|_{\infty}$ bound given
 by uniform sampling, and is potentially much stronger e.g., when $\bv A$ is sparse or numerically sparse.

 \item \textbf{Quantum-inspired eigenvector approximation.} As discussed, a benefit of our
 submatrix-based algorithms is that we do not need to compute the full approximate
 eigenvector $\bv v = \bv A \Sbb \bv x$ in a single shot using $\tilde O(n/\poly(\epsilon))$ time.
 In just $\poly(s) = \poly(\log n,1/\epsilon)$ time, we can compute $\bv x \in \R^s$ and
 any desired entry of $\bv v$.
 This feature also allows us to extend our squared column-norm sampling
 algorithms to the recently popular \emph{quantum-inspired} sampling-and-query (SQ)
 access model of~\cite{tang2019quantum,chia2020sampling}. In this model,
 one can sample the columns of the input matrix according to their squared
 Euclidean norms and query its entries at unit cost.
 The goal is to design an algorithm that provides similar SQ access to the desired output.
 We do this, showing
 that we can either output specified entries of our approximate eigenvectors or sample
 entries according to their squared magnitudes using just
 $\Tilde{O}\left(\frac{\log^8 n}{\epsilon^8} \right)$ queries to the SQ oracle on $\bv A$.
 To the best of our knowledge,
 this gives the first sublinear time eigenvector approximation algorithm
 in the quantum-inspired SQ model.
 A quantum algorithm which outputs a classical description
 of a good approximation of the top eigenvector
 under a constant eigenvalue gap
 in $n^{1.5+o(1)}$ time
 was proposed in~\cite{chen2025quantum}.

 \item \textbf{Improved sample complexities via rank-truncated Nystr\"{o}m.} Finally, we present sublinear time algorithms based on the
 rank-truncated Nystr\"{o}m method~\cite{nakatsukasa2023randomized, cai2022fast}, which achieve
 optimal sample complexities
 for approximating all outlying eigenvectors, up to logarithmic factors.
 The cost of this optimality is that these algorithms cannot compute individual entries
 of the approximate eigenvectors
 in $\textrm{poly}(\log n, 1/\epsilon)$ time.

 For general symmetric matrices with bounded entries, we give an algorithm that samples $\widetilde{O}(\log n/\epsilon^2)$ columns and,
 for every outlying eigenvalue $\lambda$ with $|\lambda| \geq \epsilon n$,
 outputs $\bv{v} \in \R^n$ with
 $\|\bv{A}\bv{v} - \lambda\bv{v}\|_2 \leq \epsilon n$.
 This is an $O(1/\epsilon^2)$ improvement in sample complexity over our submatrix-sampling-based result.
 On the technical side, we obtain our guarantees by bounding the spectral norm of the error of the rank-truncated Nystr\"{o}m approximation. To the best of our knowledge, such bounds were not previously known for $\emph{indefinite}$ matrices, and may be of independent interest.
 When columns can be sampled proportionally to squared norms,
 the error bound improves to $\epsilon\|\bv{A}\|_F$
 for all eigenvalues with $|\lambda| \geq \epsilon\|\bv{A}\|_F$. We give even stronger bounds when the input matrix is positive semidefinite. Here we can leverage the standard Nystr\"{o}m method~\cite{williams2000using,gittens2013revisiting}, giving an algorithm
 that samples only $\widetilde{O}(1/\epsilon)$ columns and achieves $\epsilon n$ error in the bounded entry setting .
 When columns are sampled proportionally to their diagonal entries,
 we achieve the tighter error bound of $\epsilon\,\mathrm{tr}(\bv{A})$.
 A simple lower bound shows that it is necessary to access $\Omega(1/\epsilon)$ columns, even for adaptive algorithms.

\end{enumerate}

\noindent
\textbf{Comparision with Matrix Sparsification.}
We note that one
can achieve similar $\tilde O(n/\poly(\epsilon))$ runtimes for approximating the outlying eigenvectors as outlined above
using entrywise matrix sparsification. For example,
it is known~\cite{drineas2011note, braverman2021near,bhattacharjee2024universal} that
when the entries of $\bv A$ are bounded in magnitude by 1,
one can uniformly sample
and scale
$\tilde O(n/\epsilon^2)$ entries to form a sparse matrix
$\hat{\bv A}$
such that $\|\bv A-\hat{\bv A} \|_2 \leq \epsilon n$.
Using standard eigenvector/eigenvalue perturbation bounds,
it can then be shown that the eigenvectors of $\hat{\bv A}$ are
approximate eigenvectors of $\bv A$ and achieve the same
backward error bounds as our algorithms.
\cite{shin2023adaptive} study an adaptive entrywise-sampling based approximate power method and prove that their algorithm
achieves the same sample complexity (up to polylogarithmic factors)
for approximating the top eigenvector
as our methods.

However, the approximate eigenvectors produced by entrywise-sampling methods (i.e., the eigenvectors of $\hat{\bv A}$)
cannot be represented as a linear combination of just a few columns of $\bv A$. Thus, we cannot output
an entry of these eigenvectors in $o(n)$ time, and the algorithms do not extend to the quantum-inspired query model.
Moreover, our column-sampling-based algorithms for PSD matrices achieve strictly better
query complexities.

\paragraph{Concurrent Work on Local Eigenvector Approximation.} Note that the problem of approximating the entries of the top eigenvector of a
bounded-entry matrix \emph{locally} in sublinear time was studied concurrently
in~\cite{menand2026locally}. The algorithm of~\cite{menand2026locally} is essentially the
same as our random submatrix-based algorithm (Algorithm~\ref{alg:eigvec}): the lifted
vector $\bv A \Sbb \bv x$ is normalized by the approximate eigenvalue rather than by its
Euclidean norm. We bound the residual
error $\|\bv A \bv v - \lambda \bv v\|_2 \leq \epsilon n$, while~\cite{menand2026locally}
bound the Rayleigh quotient, $|\lambda_{\max}(\bv A) - \langle \bv v, \bv A \bv v\rangle /
\|\bv v\|_2^2| \leq \epsilon n$. Note that the former implies the latter. Their analysis shows that the algorithm reads $\widetilde O(1/\epsilon^2)$ entries per queried coordinate, for a
total of $O(\frac{n}{\epsilon^{2}}\textrm{polylog}\frac{1}{\epsilon})$ entries for bounding the Rayleigh quotient, whenever the
eigenvalue being approximated has magnitude $\Omega(\|\bv A\|_2)$. This matches the per coordinate sample complexity of $O(\frac{\log n}{\epsilon^{2}}
\textrm{polylog}\frac{1}{\epsilon})$ and the total sample complexity of $O(\frac{n\log n}{\epsilon^{2}}
\textrm{polylog}\frac{1}{\epsilon})$ given by our analysis for bounding the residual error (Theorem~\ref{thm:eigtop}) up to polylogarithmic factor. We note that the $\log n$ factor from our analysis can likely be removed following the analysis of~\cite{menand2026locally} for bounding the spectral norm of the small magnitude eigenvalues after sampling. The sample complexity for estimating the top eigenvector is also shown to be optimal in our paper, as well as in~\cite{menand2026locally}, up to poly-logarithmic factors.

When the largest positive and the largest
negative eigenvalue have very different magnitudes, the query complexity given by~\cite{menand2026locally} for approximating the
eigenvector of the smaller of the two eigenvalues is
$O(\frac{1}{\epsilon^{3.3}}\textrm{polylog}\frac{1}{\epsilon})$ entries per coordinate in the worst case for the Rayleight-quotient guarantee, improving on the
$O(\frac{\log n}{\epsilon^{4}}\textrm{polylog}\frac{1}{\epsilon})$ entries that
Theorem~\ref{thm:main_bounded} requires for the residual error guarantee in this case. On the other hand,
Theorem~\ref{thm:main_bounded} applies uniformly to \emph{every} eigenvector whose
eigenvalue has magnitude at least $\epsilon n$ (of which there are at most $1/\epsilon^2$
for a bounded-entry matrix), whereas~\cite{menand2026locally} state no guarantee for
eigenvectors beyond the extreme ones. Finally, we note
that~\cite{menand2026locally} use their algorithm to provide local computation algorithms
for the sparsest-cut and max-cut problems in the dense graph model
of~\cite{goldreich1998property}. We expect our algorithm can likewise be used for these
applications.

\subsection{Related Work}

There has been significant work on sublinear time algorithms
for various linear algebraic problems, including testing matrix rank~\cite{balcan2019testing},
testing positive semidefiniteness~\cite{BakshiChepurkoJayaram:2020},
and approximating eigenvalues~\cite{swartworth2025tight,Bhattacharjee:2021wl}.
Streaming algorithms for approximating the top eigenvector of matrix
in sublinear space for matrices with only a few heavy rows
have also been proposed~\cite{kacham2024approximating}.

There has also been significant work on sublinear time algorithms for finding low-rank approximations,
particularly for positive semidefinite (PSD) matrices, using techniques like random Fourier features~\cite{rahimi2007random}, the Nystr{\"o}m method~\cite{williams2000using}, importance sampling \cite{frieze2004fast,drineas2006fast,mahoney2009cur}, leverage score sampling ~\cite{musco2017recursive,musco2017sublinear,bakshi2020robust}, adaptive sampling \cite{chen2025randomly}, and beyond \cite{boutsidis2014optimal}.
Sublinear time low-rank approximation algorithms have also been proposed for
structured matrices like
Toeplitz~\cite{musco2024sublinear,
lawrence2020low,kapralov2023toeplitz},
Hankel~\cite{kapralov2026sublinear} and
distance~\cite{bakshi2018sublinear,indyk2019sample}
matrices. This work is relevant since one could potentially use the eigenvalues and eigenvectors of
the low-rank approximations to approximate those of the original matrix. However, without assuming eigenvalue decay of the input matrix, it is not clear how to directly achieve algorithms with small residual error. Nevertheless, we do indeed leverage techniques from this line of work, especially in our rank-truncated Nystr\"{o}m-based algorithms. Our submatrix-based algorithms and their analysis, especially those that assume a column norm sampling oracle are closely related to classic work on norm and leverage score based sampling algorithms for low-rank approximation \cite{frieze2004fast,drineas2006fast}.

\section{Technical Overview}\label{sec:tech-overview}

We now give a detailed technical overview of our results, including a proof sketch for our main submatrix sampling based algorithm.

\subsection{Random Submatrix Sampling}\label{subsec:uniform}
Consider symmetric $\bv A \in \R^{n \times n}$ scaled such that
$\|\bv A \|_{\infty} \leq 1$. As discussed,
our first algorithm (Algorithm~\ref{alg:eigvec}), is based on the framework studied in~\cite{Bhattacharjee:2021wl,swartworth2025tight}
for eigenvalue approximation.
Let $\Sbb \in \R^{n \times s}$ be a matrix that samples $s$
columns uniformly at random from $\bv A$ and scales them by $\sqrt{n/s}$.
Algorithm~\ref{alg:eigvec} computes an eigenvector $\bv x$ of the
$s \times s$ scaled principal submatrix $\Sbb^T \bv A \Sbb$ corresponding to eigenvalue $\lambda$
with magnitude $\geq \frac{\epsilon n}{2}$. The algorithm then
\emph{lifts} $\bv x$ back to $\bv A$ by applying $\bv A \Sbb$ and outputs the unit-norm approximate eigenvector
$\bv v \;=\; \frac{\bv A \Sbb \bv x}{\|\bv A \Sbb \bv x\|_2}.$
We state our main result on the correctness of Algorithm~\ref{alg:eigvec} is below -- with a full statement given in Section \ref{app:uniform samp}.

\begin{algorithm}[t]
\caption{Eigenvector Approximation Using Uniform Sampling}
\label{alg:eigvec}
\begin{algorithmic}[1]
\State {\bfseries Input:} Symmetric $\bv A \in \mathbb{R}^{n\times n}$, accuracy $\epsilon \in (0,1)$, expected number of columns to sample $s$.
\State
Let $\bv{\bar S} \in \R^{n \times |S|}$ be a sampling matrix
that samples each column of $\bv A$ independently with probability $\frac{s}{n}$
and rescales sampled columns by $\sqrt{\frac{n}{s}}$.
Let $S \subseteq [n]$ be the set of sampled indices.

\State Compute the eigenvectors of $\bv x_1, \bv x_2, \ldots \bv x_{|S|}$ of
$\bv{\bar S}^T \bv{A} \bv{\bar S}$
corresponding to eigenvalues $\lambda_1 \ge \ldots \ge \lambda_{|S|}$.

\State For any $\lambda_i$ with $|\lambda_i| \geq \frac{\epsilon n}{2}$,
compute approximate eigenvector $\bv v_i=\frac{\bv A \bv{\bar S} \bv x_i}{\|\bv A \bv{\bar S} \bv x_i \|_2}$.
\State {\bfseries Return:} All eigenpairs $(\bv v_i, \lambda_i)$
computed in the previous step.
\end{algorithmic}
\end{algorithm}

\begin{theorem}[Informal]\label{thm:main_bounded} Let $\bv A \in \R^{n \times n}$ be symmetric with
 $\|\Ab \|_{\infty} \leq 1$,
 and let $\epsilon, \delta \in (0,1)$. With probability at least $1-\delta$, Algorithm~\ref{alg:eigvec} run on $\bv A$ with sample size $s=\frac{c\log n}{\epsilon^4 \delta}\log^3 \frac{1}{\epsilon \delta}$ for a sufficiently large constant $c$ outputs, for every eigenvalue $\lambda$ of $\bv A$
 with $|\lambda| \geq \epsilon n$, a pair $\tilde \lambda \in \R$
 and $\bv v \in \R^n$ with $\norm{\bv v}_2 = 1$
 such that
 \begin{align*}
 |\lambda-\tilde \lambda| \leq \epsilon n \text{, and } \quad
 \|\bv A \bv v - \lambda \bv v \|_2 \leq \epsilon n.
 \end{align*}
Further, the algorithm samples just
$\tilde{O}\left(\frac{\log n}{\epsilon^4 \delta}\right)$ columns from $\bv A$ and thus $\tilde{O}\left(\frac{n}{\epsilon^4 \delta}\right)$ entries in expectation.
\end{theorem}
Since $\|\bv A \|_{\infty} \leq 1$, we have $\| \bv A\|_F \leq n$
and there can be at most $ \frac{1}{\epsilon^2} $
eigenvalues with magnitude $\geq \epsilon n$. These are
the outlying eigenvectors that Algorithm \ref{alg:eigvec} approximates well. It is known \cite{Bhattacharjee:2021wl,swartworth2025tight} that the eigenvalues of of
$\Sbb^T \bv A \Sbb$ approximate those of $\bv A$ to additive error $\Theta(\epsilon n)$
(see Imported Theorem~\ref{thm:main-eigval}). Given this error bound, one cannot distinguish an eigenvalue of magnitude $ \ge \epsilon n$ from one of
magnitude $c\epsilon n$ for a large enough constant $c \in (0,1)$; thus, Algorithm \ref{alg:eigvec} may output pairs corresponding to eigenvalues slightly
below $\epsilon n$. Nevertheless, as captured by the full statement of Theorem \ref{thm:main_bounded} in Section \ref{app:uniform samp}, any output eigenpair will satisfy the stated approximation guarantee for some true eigenvalue $\lambda$ of $\bv A$.

We also show that Algorithm~\ref{alg:eigvec} achieves a
better sample complexity for approximating just the top eigenvector.

\begin{theorem}[Informal]\label{thm:eigtop}
Let $\bv A \in \R^{n \times n}$ be symmetric with $\|\Ab \|_{\infty} \leq 1$,
and let $\lambda^*$
 be the largest magnitude eigenvalue of $\bv A$. Algorithm~\ref{alg:eigvec} run on $\bv A$ with sample size $s=\frac{c\log n}{\epsilon^2 \delta}\log^5 \frac{1}{\epsilon \delta}$ for a sufficiently large constant $c$ outputs $\tilde \lambda \in \R$ and
 $\bv v \in \R^n$ with $\norm{\bv v}_2 = 1$ such that, with probability at least $1-\delta$,
 \begin{align*}
 |\lambda^*-\tilde \lambda| \leq \epsilon n \text{, and } \quad
 \|\bv A \bv v - \lambda^* \bv v \|_2 \leq \epsilon n.
 \end{align*}
 Further, the algorithm samples just
$\tilde{O}\left(\frac{\log n}{\epsilon^2 \delta}\right)$ columns from $\bv A$ and thus $\tilde{O}\left(\frac{n}{\epsilon^2 \delta}\right)$ entries in expectation.
\end{theorem}
Note that when the largest and smallest eigenvalues $\lambda_1(\bv A)$ and $\lambda_n(\bv A)$ have different signs but
are within $O(\epsilon n)$ of each other in magnitude, we cannot distinguish which is the largest in
magnitude. Hence, the full statement of Theorem \ref{thm:eigtop} in Section~\ref{app:uniform samp} considers both
the largest magnitude positive and negative eigenvalues of $\bv A$ and
ensures that, as long as they are at least $\geq \frac{\| \bv A\|_2}{2}$ in magnitude,
the approximate eigenvector corresponding to the
them has residual error $\leq \epsilon n$.
Since $\|\bv A\|_2 = \max(|\lambda_1(\bv A)|, |\lambda_n(\bv A)|)$,
at least one of the
them is always $\geq \frac{\| \bv A\|_2}{2}$. The factor $1/2$
is
chosen for simplicity and can be replaced by any constant $c \in (0,1)$.

\paragraph{Proof Sketch.}
We now sketch the proofs of
Theorems~\ref{thm:main_bounded}
and~\ref{thm:eigtop}. Following~\cite{Bhattacharjee:2021wl,swartworth2025tight},
we write $\bv A = \bv A_o+\bv A_m$ where $\bv A_o$ is its projection onto the eigenvectors corresponding to outlying eigenvalues of magnitude at least
$\epsilon n$ and $\bv A_m$ is the projection onto the remaining
`middle' eigenvectors.

\smallskip

\noindent\textbf{Basic Matrix Approximation Bounds.}
Our analysis rests on
two key facts about the scaled sampling matrix $\Sbb$, first established in~\cite{Bhattacharjee:2021wl} and then
refined in~\cite{swartworth2025tight}.
Both hold because the outlying eigenvectors of a
bounded-entry matrix are \emph{incoherent}: their
mass is spread out across coordinates.
More specifically, if $\bv U_{\geq \lambda}$ is the set of eigenvectors
of $\bv A$ with eigenvalues with magnitude at least $\lambda$,
the leverage scores of $\bv U_{\ge \lambda}$ can be bounded by
$\leq \frac{n}{\lambda^2}$.
Using this fact, \cite{swartworth2025tight} prove the following:

\smallskip
\noindent \textbf{Property 1:} For $\lambda \geq \epsilon n$, and for $s = \widetilde O(1/\epsilon^2)$, $\Sbb$ is an $\frac{\epsilon n}{\lambda}$-distortion \emph{subspace
embedding} for $\bv U_{\geq \lambda}$ i.e.
\begin{align*}
 \big(1 - \tfrac{\epsilon n}{\lambda}\big)\, \bv I \;\preceq\; \bv U_{\ge \lambda}^T \Sbb \Sbb^T \bv U_{\ge \lambda}
 \;\preceq\; \big(1 + \tfrac{\epsilon n}{\lambda}\big)\, \bv I.
\end{align*}

We can also use the incoherence
of $\bv U_o$, i.e. the eigenvectors of $\bv A_o$, to bound $\|\bv A_o \|_{\infty}$. In turn, using triangle inequality and the assumption that $\norm{\bv A}_\infty \le 1$, we can bound $\|\bv A_m \|_{\infty} \leq 1+\|\bv A_o \|_{\infty}$.
Using existing spectral norm bounds for random principal submatrices of bounded entry matrices \cite{tropp2008norms,Bhattacharjee:2021wl}, we can show that sampling doesn't `amplify' the middle eigenvalues. Specifically:

\smallskip
\noindent \textbf{Property 2:} For $s = O(\log n/\epsilon^2)$, $\| \Sbb^T \bv A_m \Sbb\|_2 \leq
\epsilon n \| \bv A\|_{\infty} \leq \epsilon n$.
\smallskip

\smallskip

\noindent\textbf{Approximate Eigenpair Residual Bounds.}
We now use the above properties to
bound the residual error of our approximate eigenvectors.
Since $\|\bv A_m \|_2 \leq \epsilon n$,
we have
$\|\bv A_m \bv v\|_2 \leq \epsilon n$ for any unit vector $\bv v$.
So, for an approximate eigenvector $\bv v$, it is enough to show
$\| \bv A_o \bv v-\lambda \bv v\|_2 \leq \epsilon n$.
Recall that $\bv v$ is of the form $\bv v=\frac{\bv A \Sbb \bv x}{\|\bv A \Sbb \bv x \|_2}$
where $\bv x$ is an eigenvector of
$\Sbb^T \bv A \Sbb$ corresponding to some eigenvalue $\lambda$ with
$|\lambda| \geq \epsilon n$.
So we have $\bv A_o \bv v= \frac{\bv A_o^2 \Sbb \bv x}{\|\bv A \Sbb \bv x\|_2}$. Thus, we must establish that $\norm{\frac{\bv A_o^2 \Sbb \bv x}{\|\bv A \Sbb \bv x\|_2}-\lambda \bv v}_2 = O(\epsilon n).$

From \textbf{Property 1} applied with error parameter $\epsilon^2$, by sampling $\tilde{O}(1/\epsilon^4)$ columns,
$\Sbb$ is a $\epsilon$-subspace embedding for the outlying
eigenvectors $\bv U_o$ i.e. $\bv U_o^T\Sbb \Sbb^T \bv U_o =(1 \pm \epsilon)\bv I$.
Using this, we can show that $\bv A_o^2 \Sbb \bv x
\approx \bv A_o \Sbb \Sbb^T \bv A_o \Sbb \bv x$ up to error $\epsilon n \|\bv A \Sbb \bv x\|_2$.
Note that on dividing both sides by the normalization factor $\|\bv A \Sbb \bv x\|_2$,
this is exactly the desired residual error $\epsilon n$.
Thus, it remains to show
\begin{align}\label{eq:interCam}
\left \|\frac{\bv A_o \Sbb \Sbb^T \bv A_o \Sbb \bv x}{\norm{\bv A \Sbb \bv x}_2} - \lambda \bv v \right \|_2 = O(\epsilon n).
\end{align}
Since $\bv x$ is an eigenvector of $\Sbb^T \bv A \Sbb$, we have:
\begin{align}\label{eq:num}
 \bv A_o \Sbb \Sbb^T \bv A_o \Sbb \bv x=\bv A_o \Sbb \Sbb^T \bv A \Sbb \bv x-\bv A_o \Sbb \Sbb^T \bv A_m \Sbb \bv x
 &=\lambda \bv A_o \Sbb \bv x -\bv A_o \Sbb \Sbb^T \bv A_m \Sbb \bv x \nonumber \\
 &= \underbrace{\lambda\, \bv A \Sbb \bv x}_{\bv t_1}
 \;-\; \underbrace{\lambda\, \bv A_m \Sbb \bv x}_{\bv t_2}
 \;-\; \underbrace{\bv A_o \Sbb \Sbb^T \bv A_m \Sbb \bv x}_{\bv t_3}.
\end{align}
Dividing by $\|\bv A \Sbb \bv x\|_2$, the first term $\bv t_1$
becomes $\lambda \bv v$ by the definition of $\bv v$. Thus, to prove \eqref{eq:interCam}, it remains to bound
the terms
$\bv t_2$ and $\bv t_3$ by $\epsilon n \|\bv A \Sbb \bv x\|_2$.

Using that $\|\bv A_m \|_{\infty}$ is bounded, we get
$\|\bv A_m \Sbb\|_2 \leq \epsilon n$. Thus, $\norm{\bv t_2}_2 \le \epsilon n |\lambda|$. By sampling $\tilde{O}(\log n/\epsilon^4)$ columns,
from \textbf{Property 2} applied with error parameter $\epsilon^2$, we have $\|\Sbb^T \bv A_m \Sbb\|_2 \leq \epsilon^2 n$.
By spectral
submultiplicativity, we can thus bound $\bv t_3$ as:
\[
 \norm{\bv t_3}_2 = \|\bv A_o \Sbb \Sbb^T \bv A_m \Sbb \bv x\|_2
 \;\le\; \|\bv A_o \Sbb\|_2 \cdot \|\Sbb^T \bv A_m \Sbb\|_2 \;=\; O(\epsilon^2 n^2) .
\]
Here, the last step follows from \textbf{Property 1} i.e. $\Sbb$ is a constant factor
subspace embedding for the eigenvectors of $\bv A_o$,
so that $\|\bv A_o \Sbb\|_2=O(\|\bv A_o \|_2) = O(\norm{\bv A}_2) =O(n)$ since $\norm{\bv A}_\infty \le 1$.

\smallskip

\noindent\textbf{Bounding $\norm{\bv A \Sbb \bv x}_2$. }
Our final step is to bound the normalization factor as
$\|\bv A \Sbb \bv x\|_2 \approx |\lambda| \geq \epsilon n$. This ensures that our upper bounds of $\norm{\bv t_2}_2 \le \epsilon n |\lambda|$ and $\norm{\bv t_3}_2 \le \epsilon^2 n^2$ are at most $\epsilon n \norm{\bv A \Sbb \bv x}_2$, as required.
To show $\|\bv A \Sbb \bv x\|_2 \approx |\lambda|$, we expand to get:
\begin{align}\label{eq:den}
 \|\bv A \Sbb \bv x\|_2^2 \;=\; \bv x^T \Sbb^T \bv A^2 \Sbb \bv x
 \;=\; \bv x^T \Sbb^T \bv A_o^2 \Sbb \bv x + \bv x^T \Sbb^T \bv A_m^2 \Sbb \bv x.
\end{align}
To bound $\bv x^T \Sbb^T \bv A_m^2 \Sbb \bv x$
we use \textbf{Property 2}, which says that by sampling $O(\log n/\epsilon^4)$ columns, we have
$\bv x^T \Sbb^T \bv A_m^2 \Sbb \bv x
\leq \|\Sbb^T \bv A_m^2 \Sbb \|_2 \leq \epsilon^2 n \|\bv A^2 \|_{\infty} \leq \epsilon^2 n^2$.
The last step follows from the fact that $\|\bv A^2\|_\infty \le n$ as $\|\bv A\|_\infty \le 1$.
For the outlying contribution $\bv x^T \Sbb^T \bv A_o^2 \Sbb \bv x$, we again use
\textbf{Property 1} to
replace $\bv I$ with $\Sbb \Sbb^T$, giving
$\bv x^T \Sbb^T \bv A_o^2 \Sbb \bv x \approx \bv x^T (\Sbb^T \bv A_o \Sbb)^2 \bv x \approx |\lambda|^2$.
Together, we get
\[
 \|\bv A \Sbb \bv x\|_2^2 \;\approx\; |\lambda|^2 \pm \epsilon^2 n^2 \;\approx\; O(\lambda^2),
\]
where the final step uses the assumption from $|\lambda| = \Omega(\epsilon n)$ (after adjusting $\epsilon$ by constants).
Overall, we have $\|\bv A \Sbb \bv x\|_2 \approx |\lambda| $, completing
the proof of Theorem~\ref{thm:main_bounded}.

\smallskip

\noindent\textbf{Improved Bound for the Top Eigenvector.} The improved $\tilde O(1/\epsilon^2)$ sample complexity for the top eigenvector given in Theorem~\ref{thm:eigtop} follows the same proof approach. However, the sample complexity improves in two places: when bounding $\norm{\bv A \Sbb \bv x}_2$ and $\norm{\bv t_3}_2$. For $\norm{\bv A \Sbb \bv x}_2$, rather than embedding the entire
outlying subspace $\bv U_o$ at distortion $\epsilon$,
following~\cite{swartworth2025tight},
we split $\bv A_o = \sum_{i=1}^r \bv A_{o,i}$ into
level sets, where $\bv A_{o,i}$ contains eigenvalues of magnitude in
$[n 2^{-i}, n2^{-(i-1)}]$ and $r=O(\log (1/\epsilon))$.
Using \textbf{Property 1}, $\Sbb$ is a $\frac{\epsilon}{ 2^{-i}}$-distortion
subspace embedding for level set $i$ using $\tilde{O}(1/\epsilon^2)$ columns.
This is enough to show $\bv x^T \Sbb^T \bv A^2 \Sbb \bv x \geq \bv x^T \Sbb^T \bv A_o^2 \Sbb \bv x \approx \|\bv A \|_2^2$ as required.
So, from~\eqref{eq:den}, we have $\|\bv A \Sbb \bv x\|_2 \approx |\lambda|=\|\bv A \|_2$. Secondly, the term $\bv t_3$ in~\eqref{eq:num}
can now be bounded by $\epsilon n \|\bv A \|_2$ instead of $\epsilon^2 n \|\bv A \|_2$
as the $\|\bv A \|_2$ term cancels out with the denominator. This again allows us to use sample complexity $\tilde O(1/\epsilon^2)$ instead of $\tilde O(1/\epsilon^4).$

As we show in Theorem~\ref{thm:gen-nystrom-eigvec},
$\widetilde O(1/\epsilon^2)$ sampled columns do
suffice to approximate all the outlying eigenvectors,
via a different technique
leveraging a rank-truncated Nystr\"om approximation,
at the cost of losing the ability to compute individual entries of the approximate eigenvector in $\textrm{poly}(\log n, 1/\epsilon)$ time. Closing this gap for Algorithm~\ref{alg:eigvec},
or proving that this algorithm requires $\omega(1/\epsilon^2)$ columns to approximate all outlying eigenvectors
is an interesting open question.

\paragraph{Squared Column-Norm Sampling.}

In Algorithm~\ref{alg:eigvec-rowsamp}, we sample columns according to their
squared norms for approximating the outlying and the top eigenvectors (see Theorems~\ref{thm:sqnorm1} and~\ref{thm:sqnorm2} respectively) to $\epsilon \| \bv A\|_F$ error, with similar sample complexities as in the uniform sampling case.
While $\| \bv A\|_F \leq n$
for bounded-entry matrices, $\| \bv A\|_F$ can be much smaller in practice.
So the error bound of $\epsilon \| \bv A\|_F$ can be
potentially much smaller than the error of $\epsilon n$
obtained from Algorithm~\ref{alg:eigvec}. Algorithm~\ref{alg:eigvec-rowsamp}
is very similar to Algorithm~\ref{alg:eigvec}, but
following~\cite{Bhattacharjee:2021wl, swartworth2025tight},
we need to judiciously zero out certain entries of the matrix before
computing the eigenvectors of the principal submatrix and scaling them up.
The zeroing out procedure is described in Steps~\ref{step:zero} and~\ref{step:zero1} of Algorithm~\ref{alg:eigvec-rowsamp}.
To see why this is necessary for eigenvalue computation, consider the case where $\bv A$ is
the identity matrix. If one samples a $1/\epsilon^2$-sized principal submatrix and
rescales it by $n\epsilon^2$, then the resulting matrix will have eigenvalues of size
$\epsilon^2 n$. This is fine if one wants $\epsilon n$ additive error.
But it is much larger than the $\epsilon \|\bv A\|_F = \epsilon\sqrt{n}$ additive error
that we actually want. So, zeroing out the diagonal entries in this case
ensures that the approximated eigenvalues are 0, which is well within
the acceptable error bound of the true eigenvalue, which is just $1$.

\paragraph{Lower Bound for Top Eigenvector.}
In Section~\ref{app:lower}, we show that the sample complexity of Theorem \ref{thm:eigtop} for top eigenvector approximation is tight up to polylogarithmic factors. Formally,
\begin{theorem}\label{thm:lower_main}
Let $\mathcal{A}$ be any randomized algorithm that
(possibly adaptively)
reads entries of
a binary matrix $\bv M \in \{0,1 \}^{n \times n}$ and
outputs unit-norm $\bv v$ that,
for $\epsilon \in (0,1)$ satisfies with prob. at least $2/3$:
\begin{align*}
 \|\bv M \bv v-\lambda_1(\bv M) \bv v \|_2 \leq \epsilon n,
\end{align*}
Then
$\mathcal{A}$ must read at
least $\Omega(\frac{n}{\epsilon^2})$ entries of $\bv M$
provided
$\epsilon=\Omega(\frac{1}{\sqrt{n}})$.
\end{theorem}
The proof is via a reduction from the
$(\epsilon, n)$-Distributed
Detection problem~\cite{bhattacharjee2024universal}.
In this problem, the goal is to recover a
hidden vector $\bv q \in \{0,1\}^n$, distributed uniformly on the Hamming cube,
given query access to an $n \times n$
matrix $\bv A$ whose entries are independent, with the $(i,j)$-th entry drawn from
$\mathrm{Bernoulli}(1/2+\epsilon)$ if $\bv q_i = 1$ and $\mathrm{Bernoulli}(1/2)$ otherwise.
It was shown in~\cite{bhattacharjee2024universal} that recovering a constant fraction of the
entries of $\bv q$ requires reading $\Omega(n/\epsilon^2)$ entries of $\bv A$.
Observe that we have $\E[\bv A] =
\tfrac12 \bv 1 \bv 1^T + \epsilon\, \bv q \bv 1^T
= \big(\tfrac12 \bv 1 + \epsilon \bv q\big)\bv 1^T$. The
entries of $\bv A - \E[\bv A]$ are random in $[-1,1]$ with mean $0$. So, by standard matrix concentration bounds,
$\|\bv A - \E[\bv A]\|_2 = O(\sqrt n)$~\cite{vershynin2018high}. Since $\E[\bv A]$ is rank-one matrix with a top singular value of $\Theta(n)$, using a standard matrix perturbation bound, we get that the top singular value of $\sigma_1(\bv A)=\Theta(n)$, and $\sigma_i(\bv A) =O(\sqrt{n})$ for all $i >1$.

To reduce the distributed detection problem to top eigenvector approximation of a symmetric matrix,
consider the matrix
$\bv M = \begin{bmatrix} \bv 0 & \bv A \\ \bv A^T & \bv 0 \end{bmatrix}$, whose eigenvalues are
$\pm\sigma_i(\bv A)$ with eigenvectors $\tfrac{1}{\sqrt2}\begin{bmatrix} \bv u_i \\ \pm\bv w_i\end{bmatrix}$,
where $\bv u_i, \bv w_i$ are the left and right singular vectors of $\bv A$.
Also, note that $\bv u_1$ is approximately $\tfrac12 \bv 1 + \epsilon \bv q$.
Since $\sigma_1(\bv A)=\Theta(n)$, and $\sigma_i(\bv A) =O(\sqrt{n})$ for all $i >1$, using a Davis-Kahan type perturbation bound,
we show that approximating the
top eigenvector of $\bv M$ to residual error $\epsilon n$ is
enough to approximate a
top left singular vector $\bv u_1$ of
$\bv A$ to $l_2$ error $O(\epsilon)$.
This is enough to infer a large $\Omega(n)$ fraction of the entries of the
hidden vector $\bv q$, thus solving the $(\epsilon, n)$-Distributed
Detection problem. This proves the lower bound.

\subsection{Fast Computation of Approximate Eigenvector Entries}

Computing an approximate eigenvector of a general symmetric $\bv A \in \R^{n \times n}$
requires at least $\Omega(n)$ time since the length of the vector is $n$.
However, if we don't want to output the full approximate eigenvectors,
but just want to compute an entry of the approximate eigenvectors
given by Algorithms~\ref{alg:eigvec} and~\ref{alg:eigvec-rowsamp},
we show that we can do it in $\mathrm{poly}(\log n, 1/\epsilon)$ time.
A key feature of the approximate eigenvectors
$\bv v=\frac{\bv A \Sbb \bv x}{\|\bv A \Sbb \bv x \|_2}$ given by
Algorithms~\ref{alg:eigvec} and~\ref{alg:eigvec-rowsamp} is that they
are spanned by $s=\mathrm{poly}(\log n, 1/\epsilon)$ columns of $\bv A$.
Moreover, $\bv x$ can be computed in $\mathrm{poly}(s)$
(where $s=\mathrm{poly}(\log n, 1/\epsilon)$)
time exactly
by an eigendecomposition of $\bv{\bar S}^T \bv{A} \bv{\bar S}$, which consists of
$s^2$ entries.
Then,
for any index $i \in [n]$, the numerator $\bv A_{i,:} \Sbb \bv x$ of the
entry
$\bv v_i=\frac{\bv A_{i,:} \Sbb \bv x}{\| \bv A \Sbb \bv x \|_2}$
can be computed
exactly in $O(s)$ time
since it is a weighted sum of $s$ entries of the $i$\textsuperscript{th}
row of $\bv A$.
As explained in the proof sketch of Theorems~\ref{thm:main_bounded}
and~\ref{thm:eigtop}
in Section~\ref{subsec:uniform},
if $\bv x$
is an eigenvector of $\Sbb^T \bv A \Sbb$ corresponding to its eigenvalue $\lambda$
where $|\lambda| \geq \epsilon n$, then
$\|\bv A \Sbb \bv x \|_2 \in [\frac{|\lambda|}{C}, C|\lambda|]$
for some constant $C$ (see Lemma~\ref{Lem:asxl}).
 Thus, $\bv v'=\frac{\bv A \Sbb \bv x}{|\lambda|}$ is a constant factor scaling of our approximate eigenvector $\bv v$ i.e. $\bv v'=c_1 \bv v$ for a constant $c_1 \in [\frac{1}{C}, C]$. Note that $\|\bv A \bv v' - \lambda \bv v'\|_2 =
c_1\|\bv A \bv v - \lambda \bv v \|_2 \leq c_1\epsilon n$. Thus, $\bv v'$ is also an approximate eigenvector of $\bv A$ corresponding to $\lambda$
after adjusting $\epsilon$ by a constant $c_1$. Thus, we can compute the $i$\textsuperscript{th} entry of
the approximate eigenvector to constant factor scaling as
$\frac{\bv A_i \Sbb \bv x}{|\lambda|}$. The time to compute the $i$\textsuperscript{th} entry is dominated by the time taken to compute the eigenvector $\bv x$ of the matrix $\Sbb^T \bv A \Sbb$. This takes at most $O(s^{\omega})$ time for an $s \times s$ submatrix where $\omega$ is the matrix multiplication exponent. Formally, we have the following result:

\begin{corollary}[Entrywise Computation of Approximate Eigenvectors]\label{cor:local}
Let $\bv A \in \R^{n \times n}$ be a symmetric matrix such that
$\|\Ab \|_{\infty} \leq 1$, and let $\epsilon, \delta \in (0,1)$.
Then, there is an algorithm that,
given $\bv A$ and any index
$j \in [n]$, reads
$O(s^2)$
entries in expectation from $\bv A$ and takes $O(s^{\omega})$ time, where $s=O\left(\frac{\log n}{\epsilon^4 \delta}
\log^3 \frac{1}{\epsilon \delta} \right)$ and $\omega$ is the matrix multiplication exponent, and outputs the pairs
$\{(\tilde \lambda_i, \bv v_{ij})\}_{i=1}^{m}$, where $\bv v_{ij}$ is the
$j$\textsuperscript{th} entry of a vector $\bv v_i \in \R^{n}$ with $\|\bv v_i \|_2 \in [\frac{1}{C}, C]$, for an absolute
constant $C > 1$. The pairs
$\{(\tilde \lambda_i, \bv v_{i})\}_{i=1}^{m}$
satisfy the error guarantees of Theorem~\ref{thm:main_bounded}
with probability at least $1-\delta$.

\end{corollary}
We also have a similar result for squared column-norm sampling
(see Corollary~\ref{cor:sqnorm}).

\subsection{Sublinear Time Quantum-Inspired Eigenvector Approximation}
Our algorithm also gives the first sublinear time classical algorithm
for eigenvector approximation in the \emph{sampling and query}
(SQ) access model introduced by Tang~\cite{tang2019quantum} and
formalized by Chia et al.~\cite{chia2020sampling}.
In this model, one is given SQ access to a matrix $\bv A$:
in $O(1)$ time, we can
sample an index $j \in [n]$ with probability
$\frac{\|\bv A_{j,:}\|_2^2}{\|\bv A\|_F^2}$, samply any
entry $\bv A_{ij}$ with probability $\frac{\bv A_{ij}^2}{\|\bv A\|_F^2}$,
query any entry $\bv A_{ij}$, query
any $\|\bv A_{j,:}\|_2$ or $\|\bv A\|_F$.

By adapting the squared column-norm sampling based
Algorithm~\ref{alg:eigvec-rowsamp},
we give the first sublinear time eigenvector approximation algorithm
in this model
which,
using $ \textrm{poly}(\log n, 1/\epsilon) $ queries to $\mathrm{SQ}(\bv A)$,
can output an entry $\bv v_i$ ($i \in [n]$) of an approximate eigenvector $\bv v$,
and sample an index $i \in [n]$
with probability proportional to $\bv v_i^2$.

We state the formal guarantee below.
 A detailed description of the quantum-inspired sampling and query access model is given in
 Section~\ref{sec:qram}.
\begin{theorem}\label{thm:sq-eigvec}
Let $\bv A \in \mathbb{R}^{n \times n}$ be a symmetric matrix
with eigenvalues $\lambda_1(\bv A) \geq \ldots \geq \lambda_n(\bv A)$
to which we have sampling and query access $\mathrm{SQ}(\bv A)$, and
let $\epsilon, \delta \in (0,1)$.
There is an algorithm that makes
$O\left(\frac{\log^8 n}{\epsilon^8 \delta^2} \log^6 \frac{1}{\epsilon \delta} \right)$
queries to $\mathrm{SQ}(\bv A)$ and outputs eigenvalue estimates
$\tilde\lambda_1, \ldots, \tilde\lambda_m \in R$ and provides
sampling and query access to vectors
$\bv v_1, \ldots, \bv v_m \in \R^n$ with $\|\bv v_i \|_2 \in [\frac{1}{C}, C]$ for $i \in [m]$ for an absolute constant $C>1$ such that, with probability at least $1-\delta$, the pairs $\{(\tilde\lambda_i, \bv v_i)\}_{i=1}^{m}$
satisfy the guarantees of Theorem~\ref{thm:sqnorm1}.
Specifically for every $\bv v_i$, $i \in [m]$, we have the following:
\begin{enumerate}
 \item \textup{(Query access)} Given $j \in [n]$, output $ \bv v_{ij}$ using $O\left(\frac{\log^4 n}{\epsilon^4 \delta}\log^3\frac{1}{\epsilon\delta}\right)$
 additional queries to $\mathrm{SQ}(\bv A)$.
 \item \textup{(Sampling access)} Sample an index $j \in [n]$ from the distribution
 $\mathcal{D}_{v_i}(j) = \bv v_{ij}^2$, with $O\left(\frac{\log^4 n}{\epsilon^6 \delta}\log^3\frac{1}{\epsilon\delta}\right)$
 queries to $\mathrm{SQ}(\bv A)$.
\end{enumerate}

\end{theorem}

\subsection{Improved Algorithms For PSD Matrices}
 For PSD matrices with bounded entries, in Section~\ref{sec:psd},
we give an algorithm that samples only $\tilde{O}\left(\frac{1}{\epsilon} \right)$
columns in expectation. The sample complexity is independent of $\log n$ and
a factor of $\epsilon^{-3}$ improvement
over the results of Theorem~\ref{thm:main_bounded} for
the general bounded-entry case
, while achieving the same $\epsilon n$
residual error.
The key insight is a connection to the \emph{Nystr\"{o}m method}~\cite{gittens2013revisiting}:
we form the rank-$s$ approximation $\hat{\bv A} = (\bv A \bv S)(\bv S^T \bv A \bv S)^{\dag}(\bv A \bv S)^T$
and observe that solving the generalized eigenvalue problem $\bv S^T \bv A^2 \bv S\, \bv x = \lambda\, \bv S^T \bv A \bv S\, \bv x$
directly yields an eigenvector $\bv v = \frac{\bv A \bv S \bv x}{\|\bv A \bv S \bv x\|_2}$ of $\hat{\bv A}$.
The residual error $\|\bv A \bv v - \lambda \bv v\|_2$ is then controlled
by the spectral approximation gap $\|\bv A - \hat{\bv A}\|_2$.
To bound this spectral error, we use Ridge Leverage Score (RLS)
sampling~\cite{musco2017recursive}.
In fact the connection to
RLS sampling was first observed in~\cite{swartworth2025tight}.
But they claim that the resulting algorithm will be non-adaptive
algorithm since
the algorithm to compute the RLS scores needs to repeatedly
sample from the matrix.
Thus, they do a more complicated analysis and could
only provide error guarantees for the top
eigenvector.
However, we observe that for a PSD matrix, the $i$\textsuperscript{th}
ridge leverage score satisfies
$l_i^{\beta}(\bv A) \leq \bv A_{ii}/\beta$. So setting
$\beta = \epsilon n$ for a bounded entry matrix shows
that uniform sampling with probability
$\tilde{O}(1/\epsilon n)$ per column i.e., $\tilde{O}(1/\epsilon)$
columns in expectation suffices. The final algorithm is given
in Algorithm~\ref{alg:psd_eigvec}.
Note that, unlike our previous algorithms, this requires forming the matrix
$\Sbb^T \bv A^2 \Sbb$ which takes $\tilde O(n/\epsilon^2)$ time.
So, we can't output an entry of the approximate eigenvector
in $\textrm{poly}(\log n, 1/\epsilon)$ time
unlike our previous algorithms. But we obtain an improved $\ell_\infty$
residual bound: the same algorithm guarantees
$\|\bv A \bv v - \lambda \bv v\|_\infty \leq \sqrt{\epsilon n}$ (Corollary~\ref{cor:psd_eiginf}).
 Our sample complexity is
tight up to logarithmic factors for any (adaptive) column sampling algorithms. To see this, observe that to distinguish a matrix
with a block of ones of size $(1/\epsilon) \times (1/\epsilon)$ and
zeros everywhere else from the all zeros matrix, we need
to read to least $O(1/\epsilon)$ columns (Theorem~\ref{thm:lower_psd}).

If we sample column $i$ with probability proportional
to $\bv A_{ii}$, as in the Randomly-Pivoted Cholesky algorithm~\cite{chen2025randomly},
we can remove the bounded-entry assumption entirely and
achieve error $\epsilon\,\mathrm{tr}(\bv A)$, which can be much smaller
than $\epsilon n$ in practice (Theorem~\ref{thm:diag_samp}).

\subsection{Improved Algorithms Using Rank-Truncated Nystr\"{o}m}

In Section~\ref{sec:gennys},
we present a rank-truncated Nystr\"{o}m-based algorithms
that achieve the same error guarantees
for approximating
outlying eigenvectors as Theorems~\ref{thm:main-eigval} and~\ref{thm:sqnorm1}
(for uniform sampling for bounded-entry matrices and squared column norm sampling respectively)
for general
symmetric (possibly indefinite) matrices
by sampling $\tilde{O}(\log n/\epsilon^2)$ columns (see Theorem~\ref{thm:gen-nystrom-eigvec}).
For the standard Nystr\"{o}m approximation
$\hat{\bv A}=\bv A \Sbb(\Sbb^T \bv A \Sbb)^\dagger \Sbb^T \bv A$ for a PSD $\bv A$, by a standard monotonicity argument~\cite{gittens2011spectral,gittens2013revisiting,musco2017recursive}, we have $\hat{\bv A} \preceq \bv A$ and so
the error $\bv A - \hat{\bv A}$ is itself
PSD. The eigenvalues of $\hat{\bv A}$ never exceed those of $\bv A$. For an indefinite matrix, however, $\Sbb^T \bv A \Sbb$ has both positive and negative eigenvalues, and
$\bv A - \hat{\bv A}$ is no longer PSD. A small magnitude eigenvalue of $\Sbb^T \bv A \Sbb$
becomes an eigenvalue of large magnitude in the pseudoinverse $(\Sbb^T \bv A \Sbb)^\dagger$, and can increase the spectral error
$\|\bv A - \hat{\bv A}\|_2$.
Then, the error is governed by the
eigenvalues of $\Sbb^T \bv A \Sbb$ that are close to zero~\cite{nakatsukasa2023randomized}.

To address this, we use a \emph{truncated} pseudoinverse
$[\Sbb^T \bv A \Sbb]_\tau^\dagger$,
restricting to the subspace of eigenvalues with magnitude at least
$\tau \geq \frac{\epsilon n}{2}$,
and form the rank-truncated Nystr\"{o}m approximation
$\hat{\bv A} = \bv C [\Sbb^T \bv A \Sbb]_\tau^\dagger \bv C^T$
where $\bv C = \bv A \bar{\bv S}$.
We show that $\| \hat{\bv A} - \bv A \|_2 \leq \epsilon n$ for bounded-entry matrices
using uniform sampling.
Hence, it is enough to approximate the eigenvectors of $\hat{\bv A}$.
Approximate eigenvectors of $\hat{\bv A}$ are obtained
by solving the generalized eigenvalue problem
$\Sbb^T \bv A^2 \Sbb\, \bv x = \lambda\, [\Sbb^T \bv A \Sbb]_\tau\, \bv x$
and setting $\bv v = \bv A \Sbb \bv x / \|\bv A \Sbb \bv x\|_2$,
as in Algorithm~\ref{alg:genNystr\"{o}m}.
To the best of our knowledge, spectral norm error bounds for
$\|\bv A - \hat{\bv A}\|_2 $
for
rank-truncated Nystr\"{o}m-type approximations
of \emph{indefinite} symmetric matrices were not previously known and may be of independent interest. The two properties of the sampling matrix $\Sbb$ we used in the proof of Theorems~\ref{thm:main_bounded} and~\ref{thm:eigtop}, along with the truncated eigenvalues of the pseudoinverse, are enough to bound the spectral norm of the error.
We also obtain similar error bounds (of $\epsilon \| \bv A\|_F$) for
squared column-norm sampling (see Algorithm~\ref{alg:nys_squared} and Theorem~\ref{thm:gen-nystrom-eigvec-norm}).

\paragraph{Roadmap.} The rest of the paper is organized as follows.
Section~\ref{sec:prelim} contains the necessary background, including
subspace embedding properties, bounds on middle eigenvalues,
and eigenvalue approximation guarantees for random principal submatrices.
Section~\ref{sec:main} contains the our main results for the random submatrix based algorithms: we first establish an approximate submatrix product lemma in section~\ref{sec:appsub},
then prove the key lemmas about bounding the normalization factor $\|\bv A \Sbb \bv x \|_2$ in Section~\ref{sec:asx2}, and finally prove the key eigenvector approximation results for uniform sampling in section~\ref{app:uniform samp} and for squared column-norm sampling in section~\ref{app:sqsamp}.
Section~\ref{app:lower} proves the lower bound showing $\Omega(1/\epsilon^2)$ sampled
columns are necessary even for approximating just the top eigenvector. Section~\ref{sec:qram} extends our results to the quantum-inspired query model. Finally, Section~\ref{sec:nys} presents our Nyström-based algorithms. Section~\ref{sec:psd} the improved guarantees for PSD matrices via ridge leverage score sampling and section~\ref{sec:gennys} contains the results for general symmetric matrices using the rank-truncated Nyström method.

\section{Preliminaries}\label{sec:prelim}

We first recall the key definitions and results
from~\cite{swartworth2025tight,Bhattacharjee:2021wl}
that we will use in our analysis.
Following~\cite{swartworth2025tight,Bhattacharjee:2021wl}, we decompose
$\bv A$ into matrices with outlying and middle eigenvalues
$\bv A_o$ and $\bv A_m$ and for an approximate eigenvector $\bv v$, we analyze
$\bv A_o \bv v$ and $\bv A_m \bv v$ separately.
Let us first recall the
formal definition of the decomposition.
\begin{definition}\label{def:ao-am}
Let $\bv{A}$ be a symmetric matrix with eigendecomposition $\bv{A} = \bv{U}\bv{\Lambda}\bv{U}^T$.
Let $\bv A_o=\bv U_o \bv \Lambda_o \bv U_o^T$
where $\bv \Lambda_o$ is a diagonal matrix
with eigenvalues of $\bv A$ with magnitude
$\geq L$ on its diagonal,
and $\bv U_o$ is matrix with the corresponding eigenvectors
of $\bv A$ as its columns.
Similarly, let $\bv A_m=\bv A-\bv A_o=\bv U_m \bv \Lambda_m \bv U_m^T$
where $\bv \Lambda_m $
contains all eigenvalues with magnitude $< L$
and $\bv U_m$ is the matrix with the corresponding eigenvectors.
Here, $L$ is a parameter that will be
set depending on the error we are trying to achieve e.g.
$\epsilon \sqrt{\delta}n$ or $\epsilon \sqrt{\delta} \| \bv A\|_F$.
\end{definition}

\subsection{Subspace Embedding Property}

We now formally define the key subspace embedding definitions and results
from~\cite{swartworth2025tight}. We first recall the definition of a subspace embedding.
\begin{definition}\label{def:subspace-embedding}
$\Sbb \in \R^{n \times k}$ is a $1 \pm \epsilon$ distortion subspace
embedding
for $\bv{X} \in \R^{n \times m}$ if for all $\bv{v} \in \R^m$,
\[
 (1 - \epsilon)\|\bv{X}\bv{v}\|_2^2 \leq \|\Sbb^T\bv{X}\bv{v}\|_2^2 \leq (1 + \epsilon)\|\bv{X}\bv{v}\|_2^2.
\]
\end{definition}
Following~\cite{swartworth2025tight}, the key subspace embedding property
that we will use is as follows:
\begin{property}[Assumption 4.1 of~\cite{swartworth2025tight}]\label{ass:subspace-embedding}
For all $\lambda \geq R$, $\bv{\bar{S}}$ is a $\min(R/\lambda, 1/10)$
distortion subspace embedding on the span of eigenvectors of
$\bv{A}$ with associated eigenvalue at least $\lambda$ in magnitude.
\end{property}
We now state two results from~\cite{swartworth2025tight} that show when Property~\ref{ass:subspace-embedding}
is realized under uniform sampling for bounded entry matrices and under row-norm sampling for general matrices.
\begin{implemma}[Lemma 4.2 of~\cite{swartworth2025tight}]\label{lem:subspace-embedding-bounded}
Let $\bv{A} \in \R^{n \times n}$ be a symmetric matrix such that $\|\bv A \|_{\infty} \leq 1$.
Let $\Sbb $ be the scaled sampling matrix as
in Algorithm \ref{alg:eigvec}. Then for
$s \geq \frac{c n^2}{R^2 } \log\frac{1}{\epsilon \delta}$ for some constant $c$,
with probability at least $1-\delta$, $\Sbb$
satisfies Property~\ref{ass:subspace-embedding}.
\end{implemma}
We will set the parameter $R$ according to our desired guarantee. For example, by setting $R=\epsilon \sqrt{\delta} n$, we get that $\Sbb$ is a $\min(\epsilon \sqrt{\delta}n/\lambda,0.1)$ subspace embedding for
all eigenvectors with eigenvalues with magnitude $\geq\lambda \geq \epsilon \sqrt{\delta} n$ when $s \geq \Tilde{O}(1/\epsilon^2 \delta)$
with probability at least $1-\delta$. Next, we state when Property~\ref{ass:subspace-embedding} is satisfied for
column-norm sampling.

\begin{implemma}[Lemma 4.3 of~\cite{swartworth2025tight}]\label{lem:subspace-embedding-norm}
Let $\bv{A} \in \R^{n \times n}$ be a symmetric matrix.
Let $p_i=\min \left(1, \frac{s\|\bv A_i \|^2_2}{\|\bv A\|_F^2}\right)$
for $i \in [n]$.
Let $\Sbb$ be the scaled sampling matrix
that samples each column $i \in [n]$ with probability $p_i$
and scales it by $\frac{1}{\sqrt{p_i}}$.
Then for
$s \geq \frac{c \|\bv A \|_F^2}{R^2}\log\frac{1}{\epsilon\delta}$ for some constant $c$,
with probability at least $1-\delta$, $\Sbb$
satisfies Property~\ref{ass:subspace-embedding}.

\end{implemma}
We will set $R = \epsilon^2 \sqrt{\delta} \| \bv A\|_F$ or
$R = \epsilon \sqrt{\delta} \|\bv A \|_F$ depending on the typeof error bounds we want
for this lemma.
Finally, we restate a useful subspace approximation result for
matrix products.
\begin{implemma}[Lemma 3.5 of~\cite{swartworth2025tight}]\label{lem:approx-matrix-product}
Let $\bv{A}$ and $\bv{B}$ be real matrices and
let $[\bv{A} \mid \bv{B}]$ denote
their concatenation. Suppose $\bv{\bar{S}}$ is an $\epsilon$-distortion
subspace embedding for $[\bv{A} \mid \bv{B}]$. Then
\[
 \|\bv{A}^T \bv{\bar{S}} \bv{\bar{S}}^T \bv{B} - \bv{A}^T \bv{B}\|_2 \leq \epsilon \|\bv{A}\|_2 \|\bv{B}\|_2.
\]
\end{implemma}

\subsection{Middle Eigenvalue Bounds}

We now state the results from~\cite{Bhattacharjee:2021wl}
and~\cite{swartworth2025tight} regarding
bounding $\|\Sbb^T \bv A_m \Sbb \|_2$.
\begin{implemma}[Lemma 4 of~\cite{Bhattacharjee:2021wl}]\label{lem:am}
Let $\bv A \in \mathbb{R}^{n\times n}$ be symmetric
such that $\|\bv A \|_{\infty} \leq 1$.
Let $\bv A_m$ be as in Definition \ref{def:ao-am}
with $L=\epsilon \sqrt{\delta}n$.
Let $\Sbb \in \R^{n \times k}$ be the scaled
sampling matrix as
in Algorithm \ref{alg:eigvec}.
If $s\geq \frac{c\log n}{\epsilon^2\delta}$ for some
sufficiently
large constant
$c$, then with probability at least $1-\delta$,
\begin{align*}
 \|\Sbb^T \bv A_m \Sbb\|_2 \leq \epsilon n
\end{align*}
\end{implemma}
Next, for squared column-norm sampling, following~\cite{swartworth2025tight}
we state the middle
eigenvalue bound (and all subsequent results)
assuming that $\frac{s\| \bv A_i\|_2^2}{\|\bv A \|^2_F} \leq 1$
i.e. $p_i =\frac{s\| \bv A_i\|_2^2}{\|\bv A \|^2_F}$ for all $i \in [n]$
for our target value of $s \geq \frac{c\log^4 n}{\epsilon^2\delta}$.
This assumption will be relaxed later
by the row duplication argument introduced in~\cite{swartworth2025tight},
i.e., we consider a matrix where
the rows and columns of $\bv A$ has been duplicated $M \geq s$ times
and scaled down by $\frac{1}{\sqrt{M}}$.
This \emph{inflated} matrix has the same nonzero eigenvalues as $\bv A$ and
the probability of sampling each row of this matrix is at most 1.
For details,
see Theorem~\ref{thm:sqnorm1}.

As discussed in the previous section, we will also need to zero out some entries of $\bv A$ when performing column-norm sampling for bounding the middle eigenvalues. The formal zeroing out procedure is given
in Definition~\ref{def:zeroing}. We state the middle eigenvalue bounds for the zeroed-out matrix , while assuming that we have $\frac{s\| \bv A_i\|_2^2}{\|\bv A \|^2_F} \leq 1$ for $i \in [n]$.
\begin{implemma}[Lemma 7.6 of~\cite{swartworth2025tight}]\label{lem:ao-sampling}
Let $\bv A'_m$ be as in Definition \ref{def:ao-am} with
$L=\epsilon \sqrt{\delta}\|\bv A \|_F$.
Let $\bv{\bar S} \in \R^{n \times |S|}$ be the
scaled sampling matrix
which samples row $i$ with probability $p_i$
and scales it by $\frac{1}{\sqrt{p_i}}$.
Assume that $\frac{s\| \bv A_i\|_2^2}{\|\bv A \|^2_F} \leq 1$ i.e.
$p_i =\frac{s\| \bv A_i\|_2^2}{\|\bv A \|^2_F}$ for all $i \in [n]$
for some $s\geq \frac{c\log^4 n}{\epsilon^2\delta}$
for some sufficiently large constant
$c$. Then
with probability
at least $1-\delta$,
\begin{align*}
 \|\Sbb^T \bv A'_m \Sbb\|_2 \leq \epsilon \|\bv A \|_F
\end{align*}
\end{implemma}

\subsection{Eigenvalue Approximation}
We now state the main results regarding eigenvalue approximation
from~\cite{swartworth2025tight}. We first state the uniform sampling result.
\begin{imptheo}[Theorem 6.4 of~\cite{swartworth2025tight}]\label{thm:main-eigval}
Let $\bv A \in \R^{n \times n}$ be a symmetric matrix and
let $\|\bv A \|_{\infty} \leq 1.$
Let $\bv{\bar S} $ be the scaled sampling matrix of
Algorithm~\ref{alg:eigvec}
for $s \geq \frac{c}{\epsilon^2\delta} \log^2 \frac{1}{\epsilon \delta}$
where $c$ is a sufficiently large constant.
Let the eigenvalues of $\bv{\bar S}^T \bv{A} \bv{\bar S}$
be $\lambda_1 \geq\lambda_2 \geq \ldots \geq \lambda_{|S|}$.
For all $i \in \{1, \ldots, |S|\}$ with $\lambda_i \ge 0$,
let $\tilde \lambda_i = \lambda_i$.
For all $i \in \{1, \ldots, |S|\}$ with $\lambda_i < 0$,
let $\tilde \lambda_{n-(|S|-i)} = \lambda_i$.
For all other $i \in \{1,\ldots,n\}$, let $\tilde \lambda_i=0$.
Then, with probability at least $1-\delta$, for all $i \in \{1,\ldots,n\}$,
\begin{align*}
 \lambda_i(\bv A) -\epsilon n \leq \tilde \lambda_i \leq \lambda_i(\bv A) +\epsilon n.
\end{align*}
\end{imptheo}
Now we state the main results regarding eigenvalue approximation
using squared column norm sampling. We will
again assume that $\frac{s\| \bv A_i\|_2^2}{\|\bv A \|^2_F} \leq 1$
for all $i \in [n]$.

\begin{imptheo}[Theorem 7.8 of~\cite{swartworth2025tight}]\label{thm:main-eigval-norm}
Let $\bv A \in \R^{n \times n}$ be a symmetric matrix.
Assume that $\frac{s\| \bv A_i\|_2^2}{\|\bv A \|^2_F} \leq 1$
for all $i \in [n]$
for some $s \geq \frac{c\log^4 n}{\epsilon^2\delta} \log^2 \frac{1}{\epsilon \delta}$
where $c$ is a sufficiently large constant.
Let $\bv{\bar S} \in \R^{n \times |S|}$ be the
scaled sampling matrix
which samples row $i$ with probability $p_i =\frac{s\| \bv A_i\|_2^2}{\|\bv A \|^2_F}$
and scales it by $\frac{1}{\sqrt{p_i}}$.
Let $\bv A'$ be the matrix after zeroing out entries of $\bv A$
as in Definition~\ref{def:ao-am}.
Let the eigenvalues of $\bv{\bar S}^T \bv{A}' \bv{\bar S}$
be $\lambda_1 \geq\lambda_2 \geq \ldots \geq \lambda_{|S|}$.
For all $i \in \{1, \ldots, |S|\}$ with $\lambda_i \ge 0$,
let $\tilde \lambda_i = \lambda_i$.
For all $i \in \{1, \ldots, |S|\}$ with $\lambda_i < 0$,
let $\tilde \lambda_{n-(|S|-i)} = \lambda_i$.
For all other $i \in \{1,\ldots,n\}$, let $\tilde \lambda_i=0$.
Then, with probability at least $1-\delta$, for all $i \in \{1,\ldots,n\}$,
\begin{align*}
 \lambda_i(\bv A) -\epsilon \|\bv A \|_F \leq \tilde \lambda_i \leq \lambda_i(\bv A) +\epsilon \|\bv A \|_F.
\end{align*}
\end{imptheo}

\section{Random Submatrix-Based Algorithms}\label{sec:main}

In this section, we
will state and prove the main results
regarding our random submatrix-based algorithms.
\subsection{Approximate Submatrix Product}\label{sec:appsub}
We start by stating and proving a key lemma
which shows that
the matrices $(\Sbb^T\Ab_o \Sbb)^2$
and $\Sbb^T\Ab_o^2\Sbb$ are close in spectral norm.
As explained previously,
our approximate eigenvector is
given by $\bv =\frac{\bv A \Sbb \bv x}{\| \bv A \Sbb \bv x\|_2}$
where $\bv x$ is an eigenvector of the scaled submatrix $\Sbb^T \bv A \Sbb$
corresponding to some eigenvalue $\lambda$.
To bound the residual error $\|\bv A \bv v - \lambda \bv v\|_2$,
we need to show that the denominator $\|\bv A \Sbb \bv x\|_2$ is
approximately $\lambda$.
Since
$\| \bv A \Sbb \bv x\|_2^2=\bv x^T\Sbb^T\bv A^2\Sbb \bv x$,
if $\Sbb^T\bv A^2\Sbb \approx (\Sbb^T\bv A \Sbb)^2$,
we have that $\| \bv A \Sbb \bv x\|_2^2 \approx \bv x^T(\Sbb^T\bv A \Sbb)^2 \bv x
= \lambda^2$.
Since $\Sbb^T \bv A_m \Sbb$ (and $\Sbb^T\bv A_m^2\Sbb$) can be bounded using Imported
Lemmas~\ref{lem:am} and Lemma~\ref{lem:ao-sampling},
so it is enough to show that $\Sbb^T \bv A_o^2 \Sbb \approx (\Sbb^T \bv A_o \Sbb)^2$.

\begin{lemma}\label{Lem:aosum}

Let $\bv A_o, L$ be as in Definition~\ref{def:ao-am}.
Let the matrix $\bv{\bar S}$ satisfy Property~\ref{ass:subspace-embedding}
for $R=L$.
Then, for $r=\lceil \log (\|\bv A \|_2/L)\rceil$, and some constant $C>0$:
\begin{align*}
 \|\Sbb^T\Ab_o \Sbb \Sbb^T \Ab_o\Sbb-\Sbb^T\Ab_o^2\Sbb \|_2 \leq CLr \|\bv A \|_2.
\end{align*}
\end{lemma}
\begin{proof}

We have $\Ab_o= \sum_{i=1}^r \Ab_{o,i}$
where $\Ab_{o,i}=\bv U_{o,i} \bv \Lambda_{o,i} \bv U_{o,i}^T$
contains eigenvalues
with magnitudes
in $(\|\bv A \|_2 \sqrt{\delta}2^{-i}, \|\bv A \|_2 \sqrt{\delta} 2^{-i+1}]$. Thus,
using triangle inequality, we have:
\begin{align}\label{eq:largeeig1}
 \|\Sbb^T\Ab_o \Sbb \Sbb^T \Ab_o \Sbb - \Sbb^T\Ab_o^2\Sbb\|_2
 \leq \sum_{i=1}^r\|\Sbb^T\Ab_{o,i} \Sbb \Sbb^T \Ab_{o,i} \Sbb - \Sbb^T\Ab_{o,i}^2\Sbb\|_2
 + \sum_{i \neq j} \|\Sbb^T\Ab_{o,i} \Sbb \Sbb^T \Ab_{o,j} \Sbb \|_2.
\end{align}
We now bound the terms individually.
Note that, by Property~\ref{ass:subspace-embedding}, $\Sbb$
is a $\min (2^i L/(\|\bv A \|_2 \sqrt{\delta}),1/10)$
subspace embedding for all eigenvectors with eigenvalues
at least $\|\bv A \|_2 \sqrt{\delta} 2^{-i} \geq L$ in magnitude.
So, for any $i \in [r]$ we get:
\begin{align*}
 \|\bv \Lambda_{o,i} \bv U_{o,i}^T\Sbb \Sbb^T \bv U_{o,i}\bv \Lambda_{o,i} - \bv \Lambda_{o,i} \bv U_{o,i}^T \bv U_{o,i} \bv \Lambda_{o,i}\|_2 \leq \frac{L}{\|\bv A \|_2 \sqrt{\delta} 2^{-i}} \|\bv U_{o,i} \bv \Lambda_{o,i} \|^2_2 \leq L \|\bv A \|_2\sqrt{\delta} 2^{-i+2} .
\end{align*}
The last step follows from the fact that $\|\bv \Lambda_{o,i} \|_2 \leq \|\bv A \|_2\sqrt{\delta} 2^{-i+1}$.

Also, $\|\Sbb^T \bv U_{o,i} \|_2 \leq (1+\frac{1}{10}) = \frac{11}{10}$ for any $i \in [r]$.
Thus, for any $i \in [r]$, we get:
\begin{align}\label{eq:1stterm1}
 \|\Sbb^T\Ab_{o,i} \Sbb \Sbb^T \Ab_{o,i} \Sbb - \Sbb^T\Ab_{o,i}^2\Sbb\|_2 &= \| \Sbb^T \bv U_{o,i}(\bv \Lambda_{o,i} \bv U_{o,i}^T\Sbb \Sbb^T \bv U_{o,i}\bv \Lambda_{o,i} - \bv \Lambda_{o,i} \bv U_{o,i}^T \bv U_{o,i} \bv \Lambda_{o,i})\bv U^T_{o,i}\Sbb \|_2 \nonumber\\
 &\leq \| \Sbb^T \bv U_{o,i}\|_2 \| \bv \Lambda_{o,i} \bv U_{o,i}^T\Sbb \Sbb^T \bv U_{o,i}\bv \Lambda_{o,i} - \bv \Lambda_{o,i} \bv U_{o,i}^T \bv U_{o,i} \bv \Lambda_{o,i} \|_2 \|\bv U^T_{o,i}\Sbb \|_2 \nonumber\\
 &\leq \frac{121}{100}L \|\bv A \|_2 \sqrt{\delta} 2^{-i+1}.
\end{align}
Next, from Imported Lemma~\ref{lem:approx-matrix-product},
for any $i, j \in [r]$ such that $i > j$
(so that $j = \min(i,j)$), we get:
\begin{align*}
 \|\bv \Lambda_{o,i} \bv U_{o,i}^T\Sbb \Sbb^T \bv U_{o,j}\bv \Lambda_{o,j} \|_2
 &\leq \frac{L}{\|\bv A \|_2 \sqrt{\delta} 2^{-i}}\|\bv \Lambda_{o,i} \|_2\|\bv \Lambda_{o,j} \|_2 \\
 &\leq \frac{L}{\|\bv A \|_2 \sqrt{\delta} 2^{-i}}\cdot \|\bv A \|_2 \sqrt{\delta} 2^{-i+1}\cdot \|\bv A \|_2 \sqrt{\delta} 2^{-j+1}
 = 2L\|\bv A \|_2 \sqrt{\delta} 2^{-\min(i,j)+1}.
\end{align*}

Then, using $\|\Sbb^T \bv U_{o,i} \|_2 \leq \frac{11}{10}$ and $\|\Sbb^T \bv U_{o,j} \|_2 \leq \frac{11}{10}$, we get:
\begin{align}\label{eq:2ndterm1}
 \|\Sbb^T\Ab_{o,i} \Sbb \Sbb^T \Ab_{o,j} \Sbb \|_2
 \leq \| \Sbb^T \bv U_{o,i}\|_2\|\bv \Lambda_{o,i} \bv U_{o,i}^T\Sbb \Sbb^T \bv U_{o,j}\bv \Lambda_{o,j} \|_2 \|\Sbb^T \bv U_{o,j} \|_2
 \leq \frac{121}{50}L\|\bv A \|_2 \sqrt{\delta}2^{-\min(i,j)+1}.
\end{align}
By symmetry, the same bound holds for $i < j$.
Thus, plugging the bounds from~\eqref{eq:1stterm1} and~\eqref{eq:2ndterm1}
into~\eqref{eq:largeeig1} and bounding $\sqrt{\delta}\leq 1$,
we get (for some constant $C$):

\begin{align*}
 \|\Sbb^T\Ab_o \Sbb \Sbb^T \Ab_o \Sbb - \Sbb^T\Ab_o^2\Sbb\|_2
 &\leq CL \|\bv A \|_2 \sum_{i=1}^{r}2^{-i+1}
 + CL \|\bv A \|_2 \sum_{i \neq j} 2^{-\min(i,j)+1} \\
 &\leq 2CL \|\bv A \|_2 + 4CLr \|\bv A \|_2
 \leq CLr \|\bv A \|_2,
\end{align*}
where in the last step all constants are absorbed in $C > 0$.

\end{proof}

\subsection{Bounding $\|\bv A \Sbb \bv x\|_2$ Using Uniform Sampling}\label{sec:asx2}

To prove our main guarantees for Algorithm~\ref{alg:eigvec},
we need to show that if $\bv x$ is an eigenvector of $\Sbb^T \bv A \Sbb$ with eigenvalue
$\lambda$,
then $\|\bv A \Sbb \bv x\|_2$ is approximately $|\lambda|$
when $\Sbb$ is a scaled sampling matrix as in Algorithm~\ref{alg:eigvec}.
Throughout this section, we will set $L=\epsilon \sqrt{\delta}n$
in the definition of $\bv A_o$ and $\bv A_m$ in Definition~\ref{def:ao-am}.
We start by
applying the Lemma~\ref{Lem:aosum} to the uniform sampling setting.
\begin{lemma}\label{Lem:spec_norm1}
Let $\bv{\bar S}$ be the scaled sampling matrix as in Algorithm~\ref{alg:eigvec}.
Let $\bv A_o, L$ be as in Definition~\ref{def:ao-am} with $L=\epsilon \sqrt{\delta}n$.
Then, for $s \geq \frac{c}{\epsilon^2 \delta}
\log^3 \frac{1}{\epsilon \delta}$, for some constant $c$, with probability
at least $1-\delta$, we have:
\begin{align*}
\|\Sbb^T\Ab_o \Sbb \Sbb^T \Ab_o\Sbb-\Sbb^T\Ab_o^2\Sbb \|_2
\leq \epsilon n \|\bv A \|_2
\end{align*}
\end{lemma}
\begin{proof}
Setting $R=L=\epsilon \sqrt{\delta}n$ in Imported Lemma~\ref{lem:subspace-embedding-bounded},
we have that,
for $s \geq \frac{c }{\epsilon^2 \delta }\log\frac{1}{\epsilon \delta}$,
with probability at least $1-\delta$, $\Sbb$ satisfies Property~\ref{ass:subspace-embedding}.
Observe that for $\|\bv A \|_{\infty}\leq 1$, $\| \bv A\|_2\leq n$.
Thus, $r=\lceil \log(\|\bv A \|_2/\epsilon \sqrt{\delta}n) \rceil \leq \lceil \log(1/\epsilon \sqrt{\delta}) \rceil$.
Then, from Lemma~\ref{Lem:aosum},
we get $\|\Sbb^T\Ab_o \Sbb \Sbb^T \Ab_o\Sbb-\Sbb^T\Ab_o^2\Sbb \|_2 \leq C\epsilon n \log (1/\epsilon\sqrt{\delta}) \|\bv A \|_2$.
The final sample complexity follows after adjusting $\epsilon$
by $\log (1/\epsilon \delta)$ and constants.

\end{proof}

We next prove that $\bv x$ is also an approximate eigenvector
for $(\Sbb^T \Ab_o \Sbb^T)^2$ with eigenvalue $\lambda^2$.
As expalined previously, since $\| \bv A \Sbb \bv x\|_2^2=\bv x^T\Sbb^T\bv A^2\Sbb \bv x$,
we need this result to show that $\bv x^T\Sbb^T\bv A^2\Sbb \bv x$ is approximately $\lambda^2$.
\begin{lemma}\label{lem:lambda_bound}
Consider the setting of Lemma~\ref{Lem:spec_norm1}.
Let $\bv x$ be a unit norm eigenvector of $\Sbb^T \Ab \Sbb$ with eigenvalue $\lambda$
such that $|\lambda| \geq \epsilon n$.
Then, for $s \geq \frac{c \log n}{\epsilon^2 \delta}$
for $c$ a sufficiently large constant,
with probability at least $1-\delta$,
 $$|\bv x^T(\Sbb^T \Ab_o \Sbb)^2\bv x -\lambda^2| \leq \epsilon n |\lambda|$$
\end{lemma}
\begin{proof}
We have
$$\bv x^T(\Sbb^T \Ab \Sbb^T)^2\bv x=\bv x^T(\Sbb^T \Ab_o \Sbb^T)^2\bv x+ \bv x^T(\Sbb^T \Ab_m \Sbb^T)^2\bv x+ 2\bv x^T \Sbb^T \Ab_o \Sbb \Sbb^T \Ab_m \Sbb \bv x.$$
Since $\bv x$ is a unit norm eigenvector of $\Sbb^T \Ab \Sbb^T$ with eigenvalue $\lambda$, $\bv x^T(\Sbb^T \Ab \Sbb^T)^2\bv x=\lambda^2$. We can thus rearrange to obtain:
\begin{align}\label{eq:interm}
 \bv x^T(\Sbb^T \Ab_o \Sbb^T)^2\bv x -\lambda^2=\bv x^T(\Sbb^T \Ab_m \Sbb^T)^2\bv x+2 \bv x \Sbb^T \Ab_o \Sbb \Sbb^T \Ab_m \Sbb \bv x.
\end{align}
Note that, from Imported Lemma~\ref{lem:am}, with probability at least $1-\delta$,
$\|\Sbb^T \Ab_m \Sbb^T \|_2 \leq \epsilon n$ and thus $|\bv x^T(\Sbb^T \Ab_m \Sbb^T)^2\bv x | \leq \epsilon^2 n^2\leq \epsilon n |\lambda|$, by the assumption that $|\lambda| \ge \epsilon n$.

Additionally, $\|\Sbb^T \Ab_o \Sbb \bv x\|_2 =
\| \Sbb^T \Ab \Sbb \bv x-\Sbb^T \Ab_m \Sbb \bv x\|_2
\leq \|\lambda \bv x \|_2+\|\Sbb^T \Ab_m \Sbb \bv x\|_2 \leq |\lambda|+\epsilon n \leq 2|\lambda|$.
Thus, via Cauchy–Schwarz, $|\bv x^T \Sbb^T \Ab_o \Sbb \Sbb^T \Ab_m \Sbb \bv x | \leq \|\Sbb^T \Ab_o \Sbb \bv x\|_2 \cdot \|\Sbb^T \Ab_m \Sbb \bv x\|_2 \leq 2\epsilon n |\lambda|$.
Plugging these bounds into~\eqref{eq:interm}
and adjusting $\epsilon$ by a constant, we have:
\begin{align*}
 |\bv x^T(\Sbb^T \Ab_o \Sbb^T)^2\bv x -\lambda^2| \leq \epsilon n |\lambda|
\end{align*}
\end{proof}

In the next lemma, whose proof
uses similar techniques to that of Imported Lemma~\ref{lem:am},
we bound $\| \bv{A}_m \bv{\bar S} \|_2$.
Since $\|\bv A_m \Sbb \bv x \|_2 \leq \|\bv A_m \Sbb \|_2$,
we will use this to bound the contribution of $\bv A_m$ to $\|\bv A \Sbb \bv x \|_2$ .
We provide a short proof for completeness.

\begin{lemma}\label{lem:as}
Consider the setting of Lemma~\ref{Lem:spec_norm1}.
Then, with probability at least $1-\delta$, we have that:
\begin{align*}
 \| \bv{A}_m \bv{\bar S} \|_2 \leq \epsilon n.
\end{align*}
\end{lemma}
\begin{proof}
To bound $\| \bv A_m \Sbb\|_2$, we first bound $\| \bv A_m \bv S\|_2$
where $\bv S$ is the unscaled sampling matrix.
Theorem 4.1 of~\cite{tropp2008norms} states the following:
\begin{align*}
 \mathbb{E}_2[\|\bv A_m \bv S \|_2] \leq 3\sqrt{\log n}\mathbb{E}_2[\|\bv A_m \bv S \|_{1 \rightarrow 2}]+\sqrt{\frac{s}{n}}\|\bv A_m \|_2.
\end{align*}
Here, $\| .\|_{1 \rightarrow 2}$ denotes the maximum Euclidean norm of the columns of the matrix
and $\mathbb{E}_2[X]=(\mathbb{E}[X^2])^{1/2}$.
Observe that we have:
\begin{align*}
 \|\bv A_m \bv S \|_{1 \rightarrow 2} \leq \|\bv A_m \|_{1 \rightarrow 2}=\|\bv U_m \bv U_m^T \bv A \|_{1 \rightarrow 2} \leq \| \bv A \|_{1 \rightarrow 2} \leq \sqrt{n}.
\end{align*}
Thus, using the fact that $\|\bv A_m \|_2 \leq \epsilon\sqrt{\delta }n$,
we get (for some constant C)
$\mathbb{E}_2[\|\bv A_m \bv S \|_2] \leq C[\sqrt{n\log n}+ \epsilon\sqrt{\delta s n}]$. Thus, setting $s \geq \frac{c\log n}{\epsilon^2 \delta}$ for some large constant $c$,
we have
\begin{align*}
 \mathbb{E}_2[\|\bv A_m \bv S \|_2] \leq C \epsilon\sqrt{\delta sn}.
\end{align*}

Using Markov's inequality, we get $\|\bv A_m \bv S \|_2 \leq C\epsilon \sqrt{sn}$
with probability at least $1-\delta$.
Since $\Sbb$ is scaled by $\sqrt{n/s}$, we have $\|\bv A_m \Sbb\|_2 \leq C\epsilon n$
We get the final after adjusting $\epsilon$ by $C$.

\end{proof}

Next, we show the main result of this section ie.e $\|\bv A\bv \Sbb \bv x \|_2$ is close to $|\lambda|$
when $\bv x$ is an eigenvector of $\Sbb^T \Ab \Sbb$
 with eigenvalue $\lambda$. We first tackle the case
 when $\bv x$ is the top eigenvector of $\Sbb^T \Ab \Sbb$.
\begin{lemma}\label{lem:asx}
Consider the setting of Lemma~\ref{Lem:spec_norm1}
and let $\|\bv A \|_2 \geq \frac{\epsilon n}{4}$.
Let $\bv{x}$ be the eigenvector of $\Sbb^T \Ab \Sbb$ corresponding to
its largest magnitude eigenvalue.
For $s \geq \frac{c \log n}{\epsilon^2 \delta} \log^3 \frac{1}{\epsilon \delta} $,
for sufficiently large constants $c$ and $C$, with
probability at least $1-\delta$, we have:
\begin{align*}
 \frac{\|\bv A \|_2}{C} \leq \|\bv A \Sbb \bv x \|_2 \leq C\|\bv A \|_2.
\end{align*}

\end{lemma}
\begin{proof}

Note that $\|\bv A \Sbb \bv x \|^2_2 = \bv x^T \Sbb^T \bv A^2\Sbb\bv x$.
Hence, it suffices to show that $\bv x^T \Sbb^T \bv A^2\Sbb\bv x$ is approximately $\|\bv A \|_2^2$.
Let $\tilde \lambda_1=\bv x^T\Sbb^T \bv A \Sbb\bv x$.
From Imported Theorem~\ref{thm:main-eigval},
setting the error parameter to $\frac{\epsilon}{16}$,
we get $||\Tilde{\lambda}_1| -\|\bv A \|_2|\leq \frac{\epsilon n}{16} \leq \frac{\|\bv A \|_2}{4}$. So, we get:
\begin{align}\label{eq:a1}
 \frac{3\|\bv A \|_2}{4} \leq |\Tilde{\lambda}_1| \leq \frac{5\|\bv A \|_2}{4}
\end{align}
Thus, using the upper bound $|\Tilde{\lambda}_1|$, we have:
\begin{align*}
 |\tilde{\lambda}_1^2 -\|\bv A \|_2^2| =|\tilde \lambda_1-\| \bv A\|_2| \cdot |\tilde \lambda_1+\| \bv A\|_2|\leq \frac{\|\bv A \|_2}{4}\cdot \frac{9\|\bv A \|_2}{4} \leq \frac{9\|\bv A \|^2_2}{16}.
\end{align*}
Next,
from Lemmas~\ref{Lem:spec_norm1} and~\ref{lem:lambda_bound}
(setting the error parameter to $\frac{\epsilon}{32}$),
we have:
\begin{align}\label{eq:l1}
 |\tilde \lambda_1^2-\bv x^T \Sbb^T \bv A_o^2\Sbb\bv x|
 &\leq |\tilde \lambda_1^2-\bv x^T(\Sbb^T\Ab _o\Sbb)^2\bv x|
 +|\bv x^T(\Sbb^T\Ab _o\Sbb)^2\bv x- \bv x^T \Sbb^T \bv A_o^2\Sbb\bv x | \nonumber \\
 &\leq \frac{\epsilon}{16} n \| \bv A\|_2 \nonumber \\
 &\leq \frac{\| \bv A\|^2_2}{4}.
\end{align}
Note that in the second step, we bounded
$|\tilde \lambda_1^2-\bv x^T(\Sbb^T\Ab _o\Sbb)^2\bv x|$ by $\frac{\epsilon}{32} n \| \bv A\|_2$
using Lemma~\ref{lem:lambda_bound}, and we bounded
$|\bv x^T(\Sbb^T\Ab _o\Sbb)^2\bv x- \bv x^T \Sbb^T \bv A_o^2\Sbb\bv x |$ by $\frac{\epsilon}{32} n \| \bv A\|_2$
using Lemma~\ref{Lem:spec_norm1}.
We also have: $$|\bv x^T\Sbb^T \bv A^2 \Sbb\bv x-\bv x^T\Sbb^T \bv A_o^2 \Sbb \bv x |
=|\bv x^T\Sbb^T \bv A_m^2 \Sbb\bv x|=\| \bv A_m \Sbb \bv x\|^2_2
\leq \|\bv A_m \Sbb \|_2^2 \leq \epsilon^2 n^2 \leq 16\|\bv A \|^2_2.$$
Here, the second last step follows from Lemma~\ref{lem:as}, and
last step follows from the fact that $\|\bv A \|_2 \geq \frac{\epsilon n}{4}$.

Combining the bounds above, we get:
\begin{align*}
 |\bv x^T \Sbb^T \bv A^2\Sbb\bv x- \|\bv A \|^2_2|
 &\leq |\bv x^T \Sbb^T \bv A^2\Sbb\bv x-\bv x^T \Sbb^T \bv A_o^2\Sbb\bv x|
 +|\bv x^T \Sbb^T \bv A_o^2\Sbb\bv x-\tilde \lambda_1^2|+|\tilde\lambda_1^2-\| \bv A\|_2^2|\\
 &\leq C\|\bv A \|_2^2,
\end{align*}
for some large constant $C$.
This proves the upper bound of $O(\|\bv A \|_2^2)$ on $\bv x^T \Sbb^T \bv A^2\Sbb\bv x$. Note that we now need to prove the lower bound separately as the constant on upper bound $C$ is $>1$ as $ \|\bv A_m \Sbb \|_2^2$ is bounded by $O(\|\bv A \|_2^2)$ above. So we prove the the lower bound on $\bv x^T \Sbb^T \bv A^2\Sbb\bv x$ separately.
Observe that we have $\bv x^T \Sbb^T \bv A^2\Sbb\bv x \geq \bv x^T \Sbb^T \bv A_o^2\Sbb\bv x$. From~\eqref{eq:l1} and~\eqref{eq:a1}, we get:
\begin{align*}
\bv x^T \Sbb^T \bv A^2\Sbb\bv x \geq \bv x^T \Sbb^T \bv A_o^2\Sbb\bv x \geq \tilde \lambda_1^2-\frac{\|\bv A \|_2^2}{4} \geq \frac{5\|\bv A \|_2^2}{16}
\end{align*}
Finally, note that we can set the confidence parameter to $\frac{\delta}{3}$
in Lemmas~\ref{Lem:spec_norm1},~\ref{lem:lambda_bound} and~\ref{lem:as} and
take a union bound over the three events to get the final result with probability at least $1-\delta$.
This completes the proof after adjusting the constants appropriately.
\end{proof}

Next, we show that, when
$\bv x$ is an eigenvector of $\Sbb^T \Ab \Sbb$ corresponding to any outlying eigenvalue
$\lambda$ such that $|\lambda| \geq \epsilon n$,
$\|\bv A \Sbb \bv x \|_2$ is still close to $|\lambda|$
as long as $s = \tilde{\Omega}(1/\epsilon^4)$.
Note that the number of columns being sampled here is worse than that of
Lemma~\ref{lem:asx} (for the top eigenvector) by a factor of $\Omega(1/\epsilon^2)$.

\begin{lemma}\label{Lem:asxl}
Consider the setting of Lemma~\ref{Lem:spec_norm1}.
Let $\bv{x}$ be an eigenvector of $\bv{\bar S}^T \bv{A} \bv{\bar S}$
corresponding to its eigenvalue $\lambda$ such that $|\lambda| \geq \epsilon n$.
For $s \geq \frac{c \log n}{\epsilon^4 \delta} \log^3 \frac{1}{\epsilon \delta}$,
for sufficiently large constants $c$ and $C$, with
probability at least $1-\delta$,
\begin{align*}
 \frac{|\lambda| }{C}\leq \|\bv A \Sbb \bv x \|_2 \leq C|\lambda|
\end{align*}
\end{lemma}

\begin{proof}
The proof follows the structure of the proof of Lemma~\ref{lem:asx} with some adjustments.
Using trainagle inequality, we have:
\begin{align}\label{eq:lem6main}
 |\lambda^2-\bv x^T\Sbb^T \Ab^2 \Sbb \bv x|
 \leq |\lambda^2-\bv x^T\Sbb^T \Ab_o^2 \Sbb \bv x|+
 |\bv x^T\Sbb^T \Ab_o^2 \Sbb \bv x-\bv x^T\Sbb^T \Ab^2 \Sbb \bv x|
\end{align}
We now bound the two terms individually.
From Lemmas~\ref{Lem:spec_norm1} and~\ref{lem:lambda_bound}
(by setting the error parameter to $\frac{\epsilon}{32}$),

\begin{align*}
 |\lambda^2-\bv x^T \Sbb^T \bv A_o^2\Sbb\bv x| &\leq | \lambda^2-\bv x^T(\Sbb^T\Ab _o\Sbb)^2\bv x|
 +|\bv x^T(\Sbb^T\Ab _o\Sbb)^2\bv x- \bv x^T \Sbb^T \bv A_o^2\Sbb\bv x | \\
 &\leq \frac{\epsilon^2}{16} n \| \bv A\|_2 \\
 &\leq \frac{\lambda^2}{16}.
\end{align*}
Here, the last step follows from the fact that $\lambda^2>\epsilon^2 n \|\bv A \|_2$
since $|\lambda| \geq \epsilon n$ and $\|\bv A \|_2 \leq n$. From Lemma~\ref{lem:as},
we also have:
$$|\bv x^T\Sbb^T \bv A^2 \Sbb\bv x-\bv x^T\Sbb^T \bv A_o^2 \Sbb \bv x |
=|\bv x^T\Sbb^T \bv A_m^2 \Sbb\bv x|=\| \bv A_m \Sbb \bv x\|^2_2
\leq \|\bv A_m \Sbb \|_2^2 \leq \epsilon^2 n^2 \leq \lambda^2.$$
Thus, from~\eqref{eq:lem6main}, we get:
\begin{align*}
 |\lambda^2-\bv x^T\Sbb^T \bv A^2 \Sbb\bv x|\leq \frac{17\lambda^2}{16}.
\end{align*}
We can also also lower bound $\bv x^T\Sbb^T \bv A^2 \Sbb\bv x$ by $\Omega(\lambda^2)$ as in Lemma~\ref{lem:asx}.
Finally, observe that we can set the confidence parameter to $\frac{\delta}{3}$
in Lemmas~\ref{Lem:spec_norm1},~\ref{lem:lambda_bound} and~\ref{lem:as} and take a union bound over
the three events to get the final result with probability at least $1-\delta$.
This completes the proof.

\end{proof}

\subsection{Eigenvector Approximation Using Uniform Sampling}\label{app:uniform samp}

We now state our final guarantees for Algorithm~\ref{alg:eigvec}.
First, we state the key lemma for approximating the top eigenvector using Algorithm~\ref{alg:eigvec}.
\begin{lemma}\label{lem:err_top}
Consider the setting of Lemma~\ref{Lem:spec_norm1}.
Let $\bv{x}_1$ be the eigenvector of $\Sbb^T \Ab \Sbb$ corresponding to
its largest magnitude eigenvalue $\lambda$. Also, let $|\lambda| \geq \frac{\epsilon n}{2}$.
Let $\bv{v}=\frac{\Ab \Sbb \xb_1}{\|\Ab \Sbb \xb_1 \|_2}$.

Then, for $s \geq \frac{c \log n}{\epsilon^2 \delta} \log^5 \frac{1}{\epsilon \delta}$
where $c$ is a sufficiently large constant,
with probability at least $1-\delta$, we have:
\begin{equation*}
 \|\Ab \bv{v}-\lambda \bv{v}\|_2 \leq \epsilon n.
\end{equation*}
\end{lemma}
\begin{proof}
We set $L=\epsilon \sqrt{\delta}n$ in Definition~\ref{def:ao-am}.
We have $\bv{\bar S}^T \bv{A} \bv{\bar S} \xb_1 =\lambda \xb_1$
or $\bv{\bar S}^T \bv{A}_m \bv{\bar S} \xb_1 + \bv{\bar S}^T \bv{A}_o \bv{\bar S} \xb_1=\lambda \xb_1$.
Thus, for $s \geq \frac{\log n}{\epsilon^2 \delta} $,
from Imported Lemma~\ref{lem:am}, with probability at least $1-\delta$, we get:
\begin{equation*}
 \bv{\bar S}^T \bv{A}_o \bv{\bar S} \xb_1=\lambda \xb_1+\bv{e},
\end{equation*}
where $\bv{e}=-\bv{\bar S}^T \bv{A}_m \bv{\bar S} \xb_1$
and $\|\bv{e} \|_2=\|\bv{\bar S}^T \bv{A}_m \bv{\bar S} \xb_1 \|_2
\leq \|\bv{\bar S}^T \bv{A}_m \bv{\bar S} \|_2 \leq \epsilon n$.
Next, observe that $\Ab \bv{v}=\Ab_o\bv{v}+\Ab_m\bv{v}$ which gives us:
\begin{equation}\label{Eq:2}
 \Ab \bv{v}=\Ab_o\bv{v}+\bv{e}'
\end{equation}
where $\|\bv{e}' \|_2=\|\Ab_m\bv{v} \|_2 \leq \|\Ab_m \|_2 \leq \epsilon n$.

Thus, it is enough to prove that $\bv A_o \bv v=\lambda \bv v+\bv e'' $
such that $\|\bv e'' \|_2 \leq \epsilon n$.

Let $\bv{v}'=\Ab \Sbb \xb_1 = \bv{v} \|\Ab \Sbb \xb_1 \|_2$,
and $\bv{v}'_o=\Ab_o\Sbb \xb_1$,
and $\bv{v}'_m=\Ab_m \Sbb \xb_1$.
Recall that in the proof of Lemma~\ref{Lem:aosum},
we had defined $\Ab_{o,i}=\bv U_{o,i} \bv \Lambda_{o,i} \bv U_{o,i}^T$
as the matrix containing eigenvalues of $\Ab_o$ with
magnitudes in $(\|\bv A \|_2 \sqrt{\delta} 2^{-i}, \|\bv A \|_2 \sqrt{\delta} 2^{-i+1}]$. Then, we have:
\begin{align}\label{eq:ao}
 \Ab_o \bv{v}'=\sum_{i=1}^r \Ab_{o,i} \Ab_{o,i} \Bar{\Sb} \bv x_1,
\end{align}
for $r=\lceil \log (\|\bv A \|_2/L)\rceil$.

We also have $\|\bv \Lambda_{o,i} \|_2 \geq \|\bv A \|_2 \sqrt{\delta}2^{-i}$.
So, from Imported Lemma~\ref{lem:subspace-embedding-bounded},
for $s \geq \frac{c}{\epsilon^2 \delta}\log \frac{1}{\epsilon \delta}$,
we get that, with probability at least $1-\delta$, for all $i \in [r]$:

\begin{align*}
 \| \bv U_{o,i}^T\Sbb \Sbb^T \bv U_{o,i}-\bv U_{o,i}^T \bv U_{o,i}\|_2
 \leq \frac{L}{\|\bv A \|_2 \sqrt{\delta} 2^{-i}}.
\end{align*}
So we have:
\begin{align*}
 \| \bv \Lambda_{o,i} \bv U_{o,i}^T\Sbb \Sbb^T \bv U_{o,i}\bv \Lambda_{o,i} - \bv \Lambda_{o,i} \bv U_{o,i}^T \bv U_{o,i} \bv \Lambda_{o,i}\|_2 \leq
 \frac{L}{\|\bv A \|_2 \sqrt{\delta} 2^{-i}} \|\bv \Lambda_{o,i} \|^2_2 \leq L\|\bv A \|_2 \sqrt{\delta} 2^{-i+2}.
\end{align*}
Thus, we have that for all $i \in [r]$:
\begin{align*}
\Ab_{o,i} \Ab_{o,i}=
\bv U_{o,i} \bv \Lambda_{o,i} \bv U_{o,i}^T \bv U_{o,i} \bv \Lambda_{o,i} \bv U_{o,i}^T
=\bv U_{o,i} \bv \Lambda_{o,i} \bv U_{o,i}^T \Sbb \Sbb^T \bv U_{o,i} \bv \Lambda_{o,i} \bv U_{o,i}^T
+ \bv U_{o,i} \bv E_i \bv U_{o,i}^T
\end{align*}
where
\begin{align}\label{eq:ei}
 \| \bv E_i\|_2 \leq L\|\bv A \|_2 \sqrt{\delta} 2^{-i+2}.
\end{align}
Multiplying both sides above by $\Sbb \bv x_1$, for every $i \in [r]$, we have:
\begin{align*}
\Ab_{o,i} \Ab_{o,i} \Bar{\Sb} \bv x_1 =\Ab_{o,i}\Sbb \Sbb^T \Ab_{o,i} \Sbb \bv x_1+ \bv U_{o,i} \bv E_i \bv U_{o,i}^T \Sbb \bv x_1.
\end{align*}
Summing over both sides above all $i \in [r]$, from~\eqref{eq:ao} we get that:
\begin{align}\label{eq:aoi}
 \Ab_o \bv{v}'&=\sum_{i=1}^r \Ab_{o,i} \Ab_{o,i} \Bar{\Sb} \bv x_1=\sum_{i=1}^r \Ab_{o,i}\Sbb \Sbb^T \Ab_{o,i} \Sbb \bv x_1+ \sum_{i=1}^r \bv U_{o,i} \bv E_i \bv U_{o,i}^T \Sbb \bv x_1.
\end{align}
Now, $\Sbb^T \bv A \Sbb \bv x_1 =\lambda \bv x_1$. Then, we have $ \Sbb^T \Ab_{o,i} \Sbb \bv x_1 = \lambda \bv x_1 -\sum_{j \neq i}\Sbb^T \bv A_{o,j} \Sbb \bv x_1-\Sbb^T \bv A_{m} \Sbb \bv x_1$. So, we get:
\begin{align*}
 \Ab_{o,i}\Sbb \Sbb^T \Ab_{o,i} \Sbb \bv x_1 &= \Ab_{o,i}\Sbb (\lambda \bv x_1 -\sum_{j \neq i}\Sbb^T \bv A_{o,j} \Sbb \bv x_1-\Sbb^T \bv A_{m} \Sbb \bv x_1) \\
 &= \lambda \Ab_{o,i}\Sbb \bv x_1-\sum_{j \neq i} \Ab_{o,i}\Sbb \Sbb^T \bv A_{o,j} \Sbb \bv x_1-\Ab_{o,i}\Sbb \Sbb^T \bv A_{m} \Sbb \bv x_1.
\end{align*}
So, summing over all $i \in [r]$, we get that:
\begin{align*}
 \sum_{i=1}^r \Ab_{o,i}\Sbb \Sbb^T \Ab_{o,i} \Sbb \bv x_1 &= \lambda \sum_{i=1}^r \Ab_{o,i}\Sbb \bv x_1- \sum_{i=1}^r \sum_{j \neq i} \Ab_{o,i}\Sbb \Sbb^T \bv A_{o,j} \Sbb \bv x_1-\sum_{i=1}^r \Ab_{o,i}\Sbb \Sbb^T \bv A_{m} \Sbb \bv x_1 \\
 &= \lambda (\bv A -\bv A_m) \Sbb \bv x_1- \sum_{i=1}^r \sum_{j \neq i} \Ab_{o,i}\Sbb \Sbb^T \bv A_{o,j} \Sbb \bv x_1-\sum_{i=1}^r \Ab_{o,i}\Sbb \Sbb^T \bv A_{m} \Sbb \bv x_1.
\end{align*}

Hence, combining the result above with~\eqref{eq:aoi}, we get:
\begin{align*}
 \bv A_o \bv v' &=\sum_{i=1}^r \Ab_{o,i}\Sbb \Sbb^T \Ab_{o,i} \Bar{\Sb} \bv x_1 + \sum_{i=1}^r \bv U_{o,i} \bv E_i \bv U_{o,i}^T \Bar{\Sb} \bv x_1\\
 &= \lambda \bv A \Sbb \bv x_1- \lambda \bv A_m \Sbb \bv x_1- \sum_{i=1}^r \sum_{j \neq i} \Ab_{o,i}\Sbb \Sbb^T \bv A_{o,j} \Sbb \bv x_1-\sum_{i=1}^r \Ab_{o,i}\Sbb \Sbb^T \bv A_{m} \Sbb \bv x_1\\
 &+ \sum_{i=1}^r \bv U_{o,i} \bv E_i \bv U_{o,i}^T \Bar{\Sb} \bv x_1.
\end{align*}
Dividing both sides by $\|\Ab \Sbb \bv x_1\|_2$ and using the fact that $\bv v=\frac{\Ab \Sbb \bv x_1}{\|\Ab \Sbb \bv x_1\|_2}$, we get,
\begin{align}\label{eq:spliteq}
 \bv A_o \bv v &= \lambda \bv v-\frac{\lambda \bv A_m \Sbb \bv x_1}{\|\Ab \Sbb \bv x_1\|_2} -\frac{\sum_{i=1}^r \sum_{j \neq i} \Ab_{o,i}\Sbb \Sbb^T \bv A_{o,j} \Sbb \bv x_1}{\|\Ab \Sbb \bv x_1\|_2} -\frac{\sum_{i=1}^r \Ab_{o,i}\Sbb \Sbb^T \bv A_{m} \Sbb \bv x_1}{\|\Ab \Sbb \bv x_1\|_2} \nonumber\\
 &+ \frac{\sum_{i=1}^r \bv U_{o,i} \bv E_i \bv U_{o,i}^T \Bar{\Sb} \bv x_1}{\|\Ab \Sbb \bv x_1\|_2}.
\end{align}
We will now bound each of the last four terms on the RHS individually.
We will upper bound the terms in the numerator and lower bound the denominator $\| \bv A \Sbb \bv x_1\|_2$.

\textbf{First term} $\frac{\lambda \bv A_m \Sbb \bv x_1}{\|\Ab \Sbb \bv x_1\|_2}$:
The numerator can be upper bounded by $\|\bv A \|_2\| \bv A_m \Sbb\|_2$.
From Lemma~\ref{lem:as}, we have $\| \bv A_m \Sbb\|_2 \leq \epsilon n$.
From Imported Theorem~\ref{thm:main-eigval}, by choosing $\epsilon$ to be sufficiently small,
we have $\|\bv A \|_2 \geq |\lambda| -\frac{\epsilon n}{4} \geq \frac{\epsilon n}{4}$.
Thus, from Lemma~\ref{lem:asx},
for $s \geq \frac{c \log n}{\epsilon^2 \delta}\log^3 \frac{1}{\epsilon \delta}$,
 we have $$\|\Ab \Sbb \bv x_1 \|_2 \geq c_1\|\bv A \|_2$$
for some constant $c_1$.
So, $\frac{\lambda \bv A_m \Sbb \bv x_1}{\|\Ab \Sbb \bv x_1\|_2}$ is bounded by
$c'\epsilon n$ for some constant $c'$.

\textbf{Second term} $\frac{\sum_{i=1}^r \sum_{j \neq i} \Ab_{o,i}\Sbb \Sbb^T \bv A_{o,j} \Sbb \bv x_1}{\|\Ab \Sbb \bv x_1\|_2}$:
Consider a single term $\Ab_{o,i}\Sbb \Sbb^T \bv A_{o,j} \Sbb \bv x_1$ for some $i,j \in [r]$ with $i \neq j$.
We have:
\begin{align*}
 \Ab_{o,i}\Sbb \Sbb^T \bv A_{o,j} \Sbb \bv x_1 &= \bv U_{o,i}\bv \Lambda_{o,i} ( \bv U_{o,i}^T \Sbb \Sbb^T \bv U_{o,j}) \bv \Lambda_{o,j} \bv U_{o,j}^T \Sbb\bv x_1
\end{align*}
Assume w.l.o.g that $i > j$.
From Imported Lemmas~\ref{lem:subspace-embedding-bounded} and~\ref{lem:approx-matrix-product}, $\Sbb$ is a
$\frac{2^i L}{\|\bv A \|_2}$ distortion subspace embedding for $[\bv U_{o,i} | \bv U_{o,j}]$.
Thus, using the fact that $\bv U_{o,i}^T \bv U_{o,j}=0$ we have:
\begin{align*}
 \| \bv \Lambda_{o,i} \bv U_{o,i}^T \Sbb \Sbb^T \bv U_{o,j}\bv \Lambda_{o,j} \|_2
 &\leq \| \bv U_{o,i}^T \Sbb \Sbb^T \bv U_{o,j}\|_2 \|\bv \Lambda_{o,i} \|_2 \|\bv \Lambda_{o,j} \|_2 \\
 &\leq \frac{L}{\|\bv A \|_2\sqrt{\delta} 2^{-i}} \|\bv \Lambda_{o,i} \|_2 \|\bv \Lambda_{o,j} \|_2 \\
 &\leq \frac{L}{\|\bv A \|_2\sqrt{\delta} 2^{-i}} \|\bv A \|^2_2 \delta 2^{-(i+j)+2} \\
 &\leq 4L \sqrt{\delta} \|\bv A \|_2 2^{-j}.
\end{align*}
Since $\Sbb$ is at
least a constant factor subspace embedding for any $\bv U_{o,j}$
according to Imported Lemma~\ref{lem:subspace-embedding-bounded}, we also have:
\begin{align*}
 \| \bv U_{o,j}^T \Sbb \bv x_1\|_2 \leq 2.
\end{align*}
Using the above two bounds, we get:
\begin{align*}
 \|\Ab_{o,i}\Sbb \Sbb^T \bv A_{o,j} \Sbb \bv x_1 \|_2 &\leq \| \bv U_{o,i} \|_2 \| \bv \Lambda_{o,i} \bv U_{o,i}^T \Sbb \Sbb^T \bv U_{o,j}\bv \Lambda_{o,j} \|_2 \| \bv U_{o,j}^T \Sbb \bv x_1 \|_2 \\
 &\leq C L \sqrt{\delta}\|\bv A \|_2 2^{-j} ,
\end{align*}
for some constant $C$. Since we also have $\|\Ab \Sbb \bv x_1 \|_2 \geq c_1\|\bv A \|_2$,
we get that $\frac{\|\Ab_{o,i}\Sbb \Sbb^T \bv A_{o,j} \Sbb \bv x_1 \|_2}{\|\bv A \Sbb \bv x_1\|_2} \leq C' L 2^{-j+1}\leq C'L$ for some constant $C'$.
So, we get:
\begin{align*}
\frac{\sum_{i=1}^r \sum_{j \neq i} \|\Ab_{o,i}\Sbb \Sbb^T \bv A_{o,j} \Sbb \bv x_1 \|_2}{\|\Ab \Sbb \bv x_1\|_2}
\leq C'Lr \leq C'\epsilon n \log (\|\bv A \|_2/L)
\end{align*}

\textbf{Third term} $\frac{\sum_{i=1}^r \Ab_{o,i}\Sbb \Sbb^T \bv A_{m} \Sbb \bv x_1}{\|\Ab \Sbb \bv x_1\|_2}$:
We have $\sum_{i=1}^r \Ab_{o,i}\Sbb \Sbb^T \bv A_{m} \Sbb \bv x_1=\Ab_{o}\Sbb \Sbb^T \bv A_{m} \Sbb \bv x_1$.
Using spectral submultiplicativity, we have:
$\|\Ab_o\Sbb\Sbb^T\Ab_m\Sbb\xb_1 \|_2 \leq \|\Ab_o\Sbb \|_2 \| \Sbb^T\Ab_m\Sbb\|\|\xb_1\|_2$.
Now, $\| \Sbb^T\Ab_m\Sbb\| \leq \epsilon n$ (from Imported Lemma~\ref{lem:am})
and $\|\Ab_o\Sbb \|_2=\|\bv \Sigma_o \bv U_o^T \Sbb\|_2 \leq \|\bv A \|_2\| \bv U_o^T \Sbb\|_2 \leq 2\|\bv A \|_2$,
where the last inequality follows from the fact that $\Sbb$ is
at least a constant factor subspace embedding
for $\bv U_o$ (Imported Lemma~\ref{lem:subspace-embedding-bounded}).
Thus, we have:
\begin{align}\label{eq:aoam1}
 \|\Ab_o\Sbb\Sbb^T\Ab_m\Sbb\xb_1 \|_2 \leq 2\epsilon n\| \bv A\|_2.
\end{align}
Finally, since $\|\Ab \Sbb \bv x_1\|_2 \geq c_1\|\bv A \|_2$ for some constant $c_1$,
we get that $\frac{\sum_{i=1}^r \Ab_{o,i}\Sbb \Sbb^T \bv A_{m} \Sbb \bv x_1}{\|\Ab \Sbb \bv x_1\|_2}$
is bounded by $O(\epsilon n)$.

\textbf{Fourth term} $\frac{\sum_{i=1}^r \bv U_{o,i} \bv E_i \bv U_{o,i}^T \Bar{\Sb} \bv x_1}{\|\Ab \Sbb \bv x_1\|_2}$:
We have:
\begin{align*}
 \| \bv U_{o,i} \bv E_i \bv U_{o,i}^T \Bar{\Sb} \bv x_1\|_2
 \leq \|\bv E_i \|_2 \| \bv U_{o,i}^T \Bar{\Sb} \bv x_1\|_2
 \leq C L \sqrt{\delta}\|\bv A \|_2 2^{-i+1}
\end{align*}
where we used the fact that $\|\bv E_i \|_2 \leq 2L\sqrt{\delta} \|\bv A \|_2 2^{-i+1}$ (from~\eqref{eq:ei}). Since $\|\bv A \Sbb \bv x_1 \|_2 \geq c_1\|\bv A \|_2 $, we get:
$\frac{\| \bv U_{o,i} \bv e_i \bv U_{o,i}^T \Bar{\Sb} \bv x_1\|_2 }{\|\bv A \Sbb \bv x_1 \|_2} \leq C L \sqrt{\delta} 2^{-i+1}$.
Then, we have:
\begin{align*}
 \sum_{i=1}^r \frac{\| \bv U_{o,i} \bv E_i \bv U_{o,i}^T \Bar{\Sb} \bv x_1\|_2 }{\|\bv A \Sbb \bv x_1 \|_2} \leq 2C L .
\end{align*}

Combining the bounds for all the four terms in~\eqref{eq:spliteq},
we get that $\|\bv A_o \bv v-\lambda \bv v\|_2$
is bounded by $O(Lr)= O(\epsilon n \log(1/\epsilon \delta))$.
By adjusting $\epsilon$ by $\log (1/\epsilon \delta)$
(and constant factors),
and taking a union bound over all the events from
Imported Lemma~\ref{lem:am}, Lemmas~\ref{Lem:aosum}
and~\ref{lem:asx}, we get the final error bound of $\epsilon n$
with probability at least $1-\delta$.
\end{proof}

We will now extend the previous lemma to approximate eigenvectors
corresponding to all the large outlying eigenvalues.
\begin{lemma}\label{lem:err1}
Consider the setting of Lemma~\ref{Lem:spec_norm1}.
Let $\bv{x}$ be an eigenvector of $\bv{\bar S}^T \bv{A} \bv{\bar S}$
corresponding to its eigenvalue $\lambda$ such that
$|\lambda| \geq \epsilon n$.
Let $\bv{v}=\frac{\Ab \Sbb \xb}{\|\Ab \Sbb \xb\|_2}$.
Then, for $s \geq \frac{c \log n}{\epsilon^4 \delta}
\log^3 \frac{1}{\epsilon \delta}$ (where $c$ is a sufficiently large constant),
with probability at least $1-\delta$, we have:
\begin{equation*}
 \|\Ab \bv{v}-\lambda \bv{v}\|_2 \leq \epsilon n
\end{equation*}
\end{lemma}
\begin{proof}
We reuse the notations from the proof of Lemma~\ref{lem:err_top}.
Let $L=\epsilon \sqrt{\delta}n $ in Definition~\ref{def:ao-am}.
Let $\bv{v}'=\Ab \Sbb \xb = \bv{v} \|\Ab \Sbb \xb \|_2$ and $\bv{v}'_o=\Ab_o\Sbb \xb$
and $\bv{v}'_m=\Ab_m \Sbb \xb$. Then, $\Ab_o \bv{v}'=\Ab_o \Ab_o\Sbb \xb
+\Ab_o \Ab_m \Sbb \xb =\Ab_o \Ab_o\Sbb \xb =\bv U_o\bv{\Lambda}_o\bv U_o^T\bv U_o\bv{\Lambda}_o\bv U_o^T \Sbb \xb$.
Recall that $\bv U_o$ contains all
eigenvectors of $\bv A$ corresponding to eigenvalues with magnitude at least $L=\epsilon \sqrt{\delta}n$.
Now, from Imported Lemma~\ref{lem:subspace-embedding-bounded} with $R=\epsilon^2 \sqrt{\delta}n$,
$\Sbb$ is a $\min(1/10,\epsilon)$-subspace embedding for $\bv U_o$.
Thus, for $s \geq \frac{c}{\epsilon^4 \delta}\log\frac{1}{\epsilon \delta}$,
with probability at least $1-\delta$,
we have: $$\bv U_o^T\bv U_o=\bv{I}=\bv U_o^T \bar{\bv S} \bar{\bv S}^T \bv U_o+\bv{E}_3$$
where $\|\bv{E}_3\|_2 \leq \epsilon $.
So, we get
\begin{equation}\label{Eq:3}
 \Ab_o \bv{v}'=\bv U_o\bv{\Lambda}_o\bv U_o^T\Sbb \Sbb^T \bv U_o\bv{\Lambda}_o\bv U_o^T\Sbb \xb+\bv U_o\bv{\Lambda}_o \bv{E}_3 \bv{\Lambda}_o \bv U_o^T \Sbb \xb
\end{equation}
Now, observe that the second term in~\eqref{Eq:3} can be written as
$\bv U_o\bv{\Lambda}_o \bv{E}_3 \bv{\Lambda}_o \bv U_o^T \Sbb \xb
= \bv U_o\bv{\Lambda}_o \bv{E}_3 \bv U_o^T\bv U_o\bv{\Lambda}_o \bv U_o^T \Sbb \xb=
\bv U_o\bv{\Lambda}_o \bv{E}_3 \bv U_o^T \Ab_o \Sbb \xb=\bv U_o\bv{\Lambda}_o \bv{E}_3 \bv U_o^T(\Ab-\Ab_m)\Sbb \xb
=\bv U_o\bv{\Lambda}_o \bv{E}_3 \bv U_o^T \Ab \Sbb \xb $
where the last step follows from the fact that $\bv U_o^T\Ab_m=0$ and so the second term cancels out. Next, observe that the first term in~\eqref{Eq:3} can be written as
\begin{align*}
 \bv U_o\bv{\Lambda}_o\bv U_o^T\Sbb \Sbb^T \bv U_o\bv{\Lambda}_o\bv U_o^T\Sbb \xb &=\bv U_o\bv{\Lambda}_o\bv U_o^T\Sbb (\Sbb^T\Ab_o\Sbb) \xb \\
 &=\bv U_o\bv{\Lambda}_o\bv U_o^T\Sbb(\lambda \xb +\bv{e}_1) \\
 &= \lambda\bv U_o\bv{\Lambda}_o\bv U_o^T\Sbb \xb + \bv U_o\bv{\Lambda}_o\bv U_o^T\Sbb \bv{e}_1 \\
 &= \lambda\Ab_o \Sbb \xb+ \bv U_o\bv{\Lambda}_o\bv U_o^T\Sbb \bv{e}_1 \\
 &= \lambda \bv{v}'_o + \bv U_o\bv{\Lambda}_o\bv U_o^T\Sbb \bv{e}_1 \\
 &= \lambda \bv{v}' + \bv U_o\bv{\Lambda}_o\bv U_o^T\Sbb \bv{e}_1 -\lambda \bv{v}'_m
\end{align*}
where the second step follows Imported Lemma~\ref{lem:am} and the last second to last
and last step follow from the definition of $\bv{v}'_o$ and $\bv{v}'$ respectively.
Plugging the bounds back into~\eqref{Eq:3}, and dividing both sides
of~\eqref{Eq:3} by $\|\Ab \Sbb \xb \|_2$, we get:
\begin{align}\label{Eq:final}
 \Ab_o \bv{v} = \lambda \bv{v}
 + \frac{\bv U_o\bv{\Lambda}_o\bv U_o^T\Sbb \bv{e}_1 }{\|\Ab \Sbb \xb \|_2}
 -\frac{\lambda \bv{v}'_m }{\|\Ab \Sbb \xb \|_2}
 + \frac{\bv U_o\bv{\Lambda}_o \bv{E}_3 \bv U_o^T \Ab \Sbb \xb }{\|\Ab \Sbb \xb \|_2}.
\end{align}
We will now bound each of the last three terms on the right-hand side above individually.

\textbf{First term:} For the first term,
we have $\bigg\|\frac{\bv U_o\bv{\Lambda}_o\bv U_o^T\Sbb \bv{e}_1 }{\|\Ab \Sbb \xb \|_2}\bigg\|_2
=\bigg\|\frac{\Ab_o\Sbb\Sbb^T\Ab_m\Sbb\xb}{\|\Ab \Sbb \xb \|_2} \bigg\|_2$.
From Imported Lemma~\ref{lem:am}, we have $\|\Sbb^T \bv A \Sbb \|_2 \leq \epsilon^2 n$
since $s \geq \frac{c \log n}{\epsilon^4 \delta}$.
Thus, the numerator can be bounded using~\eqref{eq:aoam1} as:
\begin{align}\label{eq:aoam}
 \|\Ab_o\Sbb\Sbb^T\Ab_m\Sbb\xb \|_2 \leq 2\epsilon^2 n\| \bv A\|_2 \leq \epsilon^2 n^2.
\end{align}

From Lemma~\ref{Lem:asxl}, the denominator can be lower bounded as:
\begin{align}\label{eq:denom1}
 \|\Ab \Sbb \xb \|_2 \geq c|\lambda| \geq c\epsilon n.
\end{align}
Thus, we get that
$\bigg\|\frac{\bv U_o\bv{\Lambda}_o\bv U_o^T\Sbb \bv{e}_1 }{\|\Ab \Sbb \xb \|_2}\bigg\|_2
\leq c \epsilon n$
for some constant $c$.

\textbf{Second term:} Next, for bounding the second term,
observe that
$\bigg \|\frac{\lambda \bv{v}'_m }{\|\Ab \Sbb \xb \|_2} \bigg \|_2=\frac{\|\lambda\Ab_m \Sbb \xb \|_2}{\|\Ab \Sbb \xb \|_2}$.
We again bound the numerator and denominator separately.
From Lemma~\ref{Lem:asxl}, we have $\|\Ab \Sbb \xb \|_2 \geq c |\lambda|$ for some constant $c$.
For the numerator,
we have $\|\lambda\Ab_m \Sbb \xb \|_2=|\lambda| \|\Ab_m \Sbb \xb \|_2 $.
From Lemma~\ref{lem:asx}, we have $\|\Ab_m \Sbb \|_2 \leq \epsilon^2 n$.
So, we get $\|\lambda\Ab_m \Sbb \xb \|_2 \leq |\lambda|\epsilon n$.
Finally, using the bounds on the numerator and denominator, we get $\bigg \|\frac{\lambda \bv{v}'_m }{\|\Ab \Sbb \xb \|_2} \bigg \|_2 \leq c'\epsilon n$ for some cpnstant $c'$.

\textbf{Third term:} Finally, for the third term in~\eqref{Eq:final},
observe that
$\bigg \| \frac{\bv U_o\bv{\Lambda}_o \bv{E}_3
\bv U_o^T \Ab \Sbb \xb }{\|\Ab \Sbb \xb \|_2}\bigg \|_2
\leq \frac{\|\bv U_o \|_2 \|\bv{\Lambda}_o \|_2 \|\bv{E}_3 \|_2 \|\bv U_o \|_2 \|\Ab \Sbb \xb \|_2}{\|\Ab \Sbb \xb \|_2}$
using spectral submultiplicativity.
Now, $\|\bv U_o \|_2=1$, $\|\bv{\Lambda}_o \|_2 \leq n$ and $\|\bv{e}_3 \|_2\leq \epsilon $
as stated in~\eqref{Eq:3}. So, we have $\bigg \| \frac{\bv U_o\bv{\Lambda}_o \bv{E}_3 \bv U_o^T \Ab \Sbb \xb }{\|\Ab \Sbb \xb \|_2}\bigg \|_2 \leq \epsilon n$.

Using the bounds on the three terms in~\eqref{Eq:final},
we get $\|\Ab_o \bv{v} - \lambda \bv{v}\|_2 \leq C_1\epsilon n$
for some constant $C_1$. Thus, we have $\bv A_o \bv{v}=\lambda \bv{v}+\bv{e}_4$
where $\|\bv{e}_4 \|_2 \leq C_1\epsilon n$.
Finally, using this in ~\eqref{Eq:2},
we get $\Ab\bv{v}=\Ab_o\bv{v}_o +\bv{e}_2+\bv{e}_4$
where $\| \bv{e}_2+\bv{e}_4\| \leq \|\bv{e}_2 \|_2+\|\bv{e}_4 \|_2 \leq C_2\epsilon n$.
This gives us the final bound after adjusting $\epsilon$ by constant factors.
\end{proof}

We now state the final theorem for the approximating
all the eigenvectors corresponding to the top eigenvalues using Algorithm~\ref{alg:eigvec}.

\begin{reptheorem}{thm:main_bounded}
Let $\bv A \in \R^{n \times n}$ be a symmetric matrix such that
$\|\Ab \|_{\infty} \leq 1$, and let $\epsilon, \delta \in (0,1)$.
Let Algorithm~\ref{alg:eigvec} output the pairs
$\{(\tilde \lambda_j, \bv v_j)\}_{j=1}^{m}$, where $\tilde \lambda_j \in \R$, $\bv v_j \in \R^n$,
$\|\bv v_j \|_2 = 1$, and let $m_+$ (resp.\ $m_-$) be the number of output pairs with
$\tilde\lambda_j > 0$ (resp.\ $\tilde\lambda_j < 0$), so that $m = m_+ + m_-$. Let
$\lambda_1(\bv A) \geq \ldots \geq \lambda_n(\bv A)$ denote the eigenvalues of $\bv A$. Then,
with probability at least $1-\delta$:
\begin{enumerate}
\item There is a bijection $\pi$ between the output pairs and the $m_+$ largest together with the
$m_-$ smallest eigenvalues of $\bv A$, i.e.\ the multiset
$\{\lambda_1(\bv A), \ldots, \lambda_{m_+}(\bv A)\} \cup \{\lambda_{n-m_-+1}(\bv A), \ldots, \lambda_n(\bv A)\}$,
such that every output pair $j \in [m]$ satisfies
\begin{align*}
 |\lambda_{\pi(j)}(\bv A)-\tilde \lambda_{j}| \leq \epsilon n \text{, and } \quad
 \|\bv A \bv v_{j}-\lambda_{\pi(j)}(\bv A) \bv v_{j} \|_2 \leq \epsilon n.
\end{align*}
\item Every eigenvalue of $\bv A$ outside this matched multiset has magnitude less than $\epsilon n$.
\end{enumerate}
Moreover, Algorithm~\ref{alg:eigvec} samples
$\frac{c \log n}{\epsilon^4 \delta}
\log^3 \frac{1}{\epsilon \delta}$
columns from $\bv A$ in expectation,
for some sufficiently large constant $c$.
\end{reptheorem}
\begin{proof}
We condition on the events from Imported Theorem~\ref{thm:main-eigval} (applied
with error parameter $\frac{\epsilon}{4}$)
and Lemma~\ref{lem:err1} (applied with error parameter
$\frac{\epsilon}{2}$), each of which holds with probability at least
$1-\delta/2$, and take a union bound.
First,
if $\tilde\lambda_1 \geq \ldots \geq \tilde\lambda_{|S|}$ are the
eigenvalues of $\bv{\bar S}^T\bv A\bv{\bar S}$, the signed sorted matching of Imported
Theorem~\ref{thm:main-eigval} assigns the $i$-th largest \emph{positive} eigenvalue of
$\bv{\bar S}^T\bv A\bv{\bar S}$ to $\lambda_i(\bv A)$, and the $i$-th smallest \emph{negative} one
to $\lambda_{n-i+1}(\bv A)$, with
$|\lambda - \tilde\lambda| \leq \frac{\epsilon n}{4}$ for every matched pair $\lambda, \tilde \lambda$
(all remaining
eigenvalues of $\bv A$ are matched to $0$).
Define $\pi$ by this signed sorted matching: the output pair
carrying the $i$-th largest positive eigenvalue is mapped to $\lambda_i(\bv A)$ for $i \in [m_+]$,
and the pair carrying the $i$-th smallest negative eigenvalue to $\lambda_{n-i+1}(\bv A)$ for
$i \in [m_-]$. For each positive output pair,
$\lambda_{\pi(j)}(\bv A) \geq \tilde\lambda_j - \frac{\epsilon n}{4} \geq \frac{\epsilon n}{4} > 0$,
and symmetrically each negative output pair is matched to an eigenvalue
$\leq -\frac{\epsilon n}{4} < 0$; hence the two matched sets are disjoint and $\pi$ is a bijection
onto the multiset $\{\lambda_1(\bv A), \ldots, \lambda_{m_+}(\bv A)\} \cup
\{\lambda_{n-m_-+1}(\bv A), \ldots, \lambda_n(\bv A)\}$.

Second, from Lemma~\ref{lem:err1},
 for $s \geq \frac{c' \log n}{\epsilon^4\delta}\log^3\frac{1}{\epsilon\delta}$
with $c'$ a sufficiently large constant, every eigenpair $(\bv x, \tilde\lambda)$ of
$\bv{\bar S}^T \bv{A} \bv{\bar S}$ with $|\tilde\lambda| \geq \frac{\epsilon n}{2}$ satisfies
$\|\bv A \bv v - \tilde\lambda \bv v\|_2 \leq \frac{\epsilon n}{2}$ for
$\bv v = \frac{\bv A \Sbb \bv x}{\|\bv A \Sbb \bv x\|_2}$.

\emph{The two bounds.} For every output pair $j$, the first event gives
$|\lambda_{\pi(j)}(\bv A) - \tilde\lambda_j| \leq \frac{\epsilon n}{4} \leq \epsilon n$, and by
the second event together with the triangle inequality,
\[
\|\bv A \bv v_{j}-\lambda_{\pi(j)}(\bv A) \bv v_{j} \|_2
\leq \|\bv A \bv v_{j}-\tilde \lambda_{j} \bv v_{j} \|_2
+ |\tilde \lambda_{j}-\lambda_{\pi(j)}(\bv A)| \,\|\bv v_{j}\|_2
\leq \frac{\epsilon n}{2} + \frac{\epsilon n}{4} \leq \epsilon n.
\]

\emph{No large eigenvalue is missed.} Suppose $\lambda_i(\bv A) \geq \epsilon n$ for some
$i > m_+$. Its matched value under the signed sorted matching satisfies
$\tilde\lambda \geq \lambda_i(\bv A) - \frac{\epsilon n}{4} \geq \frac{3\epsilon n}{4} >
\frac{\epsilon n}{2}$; in particular $\tilde\lambda > 0$, so it is the $i$-th largest positive
eigenvalue of $\bv{\bar S}^T\bv A\bv{\bar S}$ and is always output by the algorithm.
Symmetrically, every eigenvalue $\lambda_{n-i+1}(\bv A) \leq
-\epsilon n$ satisfies $i \leq m_-$. Hence every eigenvalue outside the matched set has
magnitude less than $\epsilon n$.

\end{proof}

We now state the final theorem regarding approximating
the top eigenvector using Algorithm~\ref{alg:eigvec}.
\begin{reptheorem}{thm:eigtop}
Let $\bv A \in \R^{n \times n}$ be a symmetric matrix such that $\|\Ab \|_{\infty} \leq 1$, with
eigenvalues $\lambda_1(\bv A) \geq \ldots \geq \lambda_n(\bv A)$, and let $\epsilon, \delta \in (0,1)$.
Set $s \geq \frac{c \log n}{\epsilon^2 \delta}
\log^5 \frac{1}{\epsilon \delta}$ in Algorithm~\ref{alg:eigvec}
for some sufficiently large constant $c$. Let $\tilde\lambda_1$ and $\tilde\lambda_2$ be the largest
positive and smallest negative eigenvalues output by Algorithm~\ref{alg:eigvec}, and let $\bv v_1$
and $\bv v_2$ be their corresponding approximate (unit) eigenvectors, respectively. If no positive
(resp.\ negative) eigenvalue is output by Algorithm~\ref{alg:eigvec}, set $\tilde\lambda_1 = 0$
(resp.\ $\tilde\lambda_2 = 0$). Then, with probability at
least $1-\delta$, the following hold simultaneously:
\begin{enumerate}
\item If $\|\bv A \|_2 \geq 2\epsilon n$ and
$\lambda_1(\bv A) \geq \frac{\|\bv A\|_2}{2}$, then $\tilde \lambda_1>0$, and:
 \begin{align*}
 |\lambda_1(\bv A)-\tilde \lambda_1| \leq \epsilon n \text{, and }\quad
 \|\bv A \bv v_1-\lambda_1(\bv A) \bv v_1 \|_2 \leq \epsilon n.
 \end{align*}
\item If $\|\bv A \|_2 \geq 2\epsilon n$ and
$\lambda_n(\bv A) \leq -\frac{\|\bv A\|_2}{2}$, then $\tilde \lambda_2<0$, and:
 \begin{align*}
 |\lambda_n(\bv A)-\tilde \lambda_2| \leq \epsilon n \text{, and }\quad
 \|\bv A \bv v_2-\lambda_n(\bv A) \bv v_2 \|_2 \leq \epsilon n.
 \end{align*}
\item If $\|\bv A\|_2 < 2\epsilon n$, $\tilde\lambda_1 = 0$ and $\tilde\lambda_2 = 0$, and for
\emph{any} unit vectors $\bv v_1$ and $\bv v_2$ the pairs $(\tilde\lambda_1, \bv v_1)$ and
$(\tilde\lambda_2, \bv v_2)$ satisfy
$|\lambda_1(\bv A) - \tilde\lambda_1| \leq 2\epsilon n$,
$|\lambda_n(\bv A) - \tilde\lambda_2| \leq 2\epsilon n$,
$\|\bv A \bv v_1 - \lambda_1(\bv A)\bv v_1\|_2 \leq 2\epsilon n$, and
$\|\bv A \bv v_2 - \lambda_n(\bv A)\bv v_2\|_2 \leq 2\epsilon n$ trivially.
\end{enumerate}
Moreover, Algorithm~\ref{alg:eigvec} samples $\frac{c \log n}{\epsilon^2 \delta}
\log^5 \frac{1}{\epsilon \delta}$ columns from $\bv A$ in expectation.
\end{reptheorem}
\begin{proof}

\noindent We first assume that $\|\bv A \|_2 \geq 2\epsilon n$ and prove items 1 and item 2.

\smallskip
\noindent \textbf{Assume} $\|\bv A \|_2 \geq 2\epsilon n.$
Since $\|\bv A \|_2=\max (|\lambda_1(\bv A)|, |\lambda_n(\bv A)|)$, at least one of
$\lambda_1(\bv A)$ or $\lambda_n(\bv A)$ must be $\geq \frac{\|\bv A \|_2}{2} \geq \epsilon n$
in magnitude.
We assume $\lambda_1(\bv A) \geq \frac{\|\bv A \|_2}{2} \geq \epsilon n$
and prove item 1.
Then, item 2 follows by a symmetric argument.
We condition on the following two events, each holding with probability at least $1-\delta/2$,
and take a union bound.
First, the event of Imported Theorem~\ref{thm:main-eigval} applied with error parameter
$\frac{\epsilon}{2}$: the signed sorted matching assigns the $i$-th largest positive eigenvalue
of $\bv{\bar S}^T\bv A\bv{\bar S}$ to $\lambda_i(\bv A)$ and the $i$-th smallest negative one to
$\lambda_{n-i+1}(\bv A)$, with each matched pair satisfying
$|\lambda - \tilde\lambda| \leq \frac{\epsilon n}{2}$.
Second, the conclusions of Lemma~\ref{lem:asx} and Lemma~\ref{lem:err_top}, applied with error
parameter $\frac{\epsilon}{2}$, for every eigenpair $(\bv x, \tilde\lambda)$ of
$\Sbb^T\bv A\Sbb$ satisfying $|\tilde\lambda| \geq \max\big(\frac{\epsilon n}{4},
\frac{\|\bv A\|_2}{4}\big)$: the proofs of these lemmas use the assumption that $\tilde\lambda$
is the largest magnitude eigenvalue of $\Sbb^T\bv A\Sbb$ only through the two bounds
$|\tilde\lambda| \geq \frac{\epsilon n}{4}$ and $|\tilde\lambda| = \Omega(\|\bv A\|_2)$, and
therefore apply, with the constant $c$ in the sample complexity adjusted, to every such
eigenpair simultaneously. In particular, for every such pair, the vector
$\bv v = \frac{\bv A\Sbb\bv x}{\|\bv A\Sbb\bv x\|_2}$ satisfies
$\|\bv A \bv v - \tilde\lambda \bv v\|_2 \leq \frac{\epsilon n}{2}$.

\emph{Eigenvalue guarantee.}
As $\lambda_1(\bv A) \geq \epsilon n$, the largest magnitude
eigenvalue $\tilde\lambda$ of $\Sbb^T\bv A\Sbb$ matched to
$\lambda_1(\bv A)$ satisfies $\tilde\lambda \geq \lambda_1(\bv A) - \frac{\epsilon n}{2} \geq
\frac{\epsilon n}{2} > 0$, so it clears the threshold and a positive eigenvalue is satisfying
error output (thus
$\tilde\lambda_1 > 0$).

\emph{Eigenvector guarantee.} We have $ \|\bv A \bv v_1-\lambda_1(\bv A) \bv v_1\|_2 \leq \frac{\epsilon n}{2}$
and $|\tilde\lambda_1-\lambda_1(\bv A)| \leq \frac{\epsilon n}{2}$ as explained above. Then, byy triangle inequality:
\begin{align*}
 \|\bv A \bv v_1-\lambda_1(\bv A) \bv v_1\|_2 \leq \|\bv A \bv v_1-\tilde\lambda_1 \bv v_1\|_2
+ |\tilde\lambda_1-\lambda_1(\bv A)| \,\|\bv v_1\|_2 \leq \epsilon n
\end{align*}
This completes the proof of item 1. We can similarly prove item 2 when $\|\bv A \|_2 \geq 2\epsilon n$.

\smallskip
\noindent \textbf{Assume} $\|\bv A \|_2 < 2\epsilon n.$ We now prove item 3.
Both $|\lambda_1(\bv A))|$ and $|\lambda_n(\bv A))|$ are $<\epsilon n$ in this case.
Then, both the largest magnitude positive and negative eigenvalue of $\Sbb^T \bv A \Sbb$
will be $<\frac{\epsilon n}{2}$ in magnitude,
with probability at least $1-\delta$, by choosing the error parameter $\frac{\epsilon}{2}$.
So, we have $\tilde \lambda_1=0$ and $\tilde \lambda_2=0$. Moreover, we
always have $\|\bv A \bv v-\lambda_i(\bv A)\bv A \|_2$ for any eigenvalue $\lambda_i(\bv A)$
after adjusting $\epsilon$ by constant factors.
\end{proof}

\paragraph{Fast Computation of Entries of Approximate Eigenvector.}
Finally, we state the following lemma, which
shows that the entries of the approximate eigenvectors given
by Algorithm~\ref{alg:eigvec} can be computed
in $\textrm{poly}(\log n, 1/\epsilon, 1/\delta)$ time.

\begin{repcorollary}{cor:local}
Let $\bv A \in \R^{n \times n}$ be a symmetric matrix such that
$\|\Ab \|_{\infty} \leq 1$, and let $\epsilon, \delta \in (0,1)$.
Then, there is an algorithm that,
given $\bv A$ and any index
$j \in [n]$, reads
$O(s^2)$
entries in expectation from $\bv A$ and takes $O(s^{\omega})$ time, where $s=O\left(\frac{\log n}{\epsilon^4 \delta}
\log^3 \frac{1}{\epsilon \delta} \right)$ and $\omega$ is the matrix multiplication exponent, and outputs the pairs
$\{(\tilde \lambda_i, \bv v_{ij})\}_{i=1}^{m}$, where $\bv v_{ij}$ is the
$j$\textsuperscript{th} entry of a vector $\bv v_i \in \R^{n}$ with $\|\bv v_i \|_2 \in [\frac{1}{C}, C]$, for an absolute
constant $C > 1$. The pairs
$\{(\tilde \lambda_i, \bv v_{i})\}_{i=1}^{m}$
satisfy the error guarantees of Theorem~\ref{thm:main_bounded}
with probability at least $1-\delta$.
\end{repcorollary}
\begin{proof}

Recall that
the $i$-th output vector from Theorem~\ref{thm:main_bounded} is
$\bv v_i=\frac{\bv A \Sbb \bv x_i}{\|\bv A \Sbb \bv x_i \|_2}$,
where $(\bv x_i, \tilde\lambda_i)$ is the $i$-th returned eigenpair of
$\bv{\bar S}^T \bv{A} \bv{\bar S}$.
The algorithm reads the $s \times s$ submatrix $\Sbb^T \bv A \Sbb$, computes the eigenvectors and eigenvalues $(\bv x_i, \tilde\lambda_i)$ in $O(s^{\omega})$ time, and outputs $\frac{\bv A_{j,:}\Sbb \bv x_i}{|\tilde \lambda_i|}$ for each $i \in [m]$. The numerator $\bv A_{j,:}\Sbb \bv x_i$,
can be computed for all $i \in [m]$ by reading the (at most) $s$ entries of the
$j$\textsuperscript{th} row of $\bv A$.

By Lemma~\ref{Lem:asxl}, with probability at least $1-\delta$ we have
$\frac{1}{C}|\tilde \lambda_i| \leq \|\bv A \Sbb \bv x_i \|_2 \leq C |\tilde \lambda_i|$
for all $i \in [m]$. Let $\bv v'=\frac{\bv A \Sbb \bv x}{|\lambda|}$. Then, $\bv v'=c_1 \bv v$ for a constant $c_1 \in [\frac{1}{C}, C]$. Note that, from Theorem~\ref{thm:main_bounded} ,$\|\bv A \bv v' - \lambda \bv v'\|_2 =
c_1\|\bv A \bv v - \lambda \bv v \|_2 \leq c_1\epsilon n$. Thus, $\bv v'$ also satisfies the guarantee of Theorem~\ref{thm:main_bounded}
after adjusting $\epsilon$ by a constant $C$.
\end{proof}

\subsection{Eigenvector Approximation Using Squared Column-Norm Sampling}\label{app:sqsamp}

\begin{algorithm}[t]
\caption{Eigenvector Approximation Using Squared Column-Norm Sampling}
\label{alg:eigvec-rowsamp}
\begin{algorithmic}[1]
\State {\bfseries Input:} Symmetric $\bv A \in \mathbb{R}^{n\times n}$,
row norms $\{\|\bv A_{i,:}\|_2\}_{i=1}^n$,
accuracy $\epsilon \in (0,1)$, expected sample size $s$
\State \label{step:pois} Sample $K \sim \mathrm{Poisson(s)}$.
\State \label{step:samp} \textbf{(Sampling)} Sample $K$ times from the distribution of the squared
column norms $(p_1,....,p_n)$ where $p_i=\frac{\|\bv A_i \|^2_2}{\| \bv A\|^2_F}$. Let $S=\{i_1, \ldots, i_K \}$ ($i_p \in [n]$ for all $p \in [K]$) be the ordered list of sampled indices.

\State \label{step:zero} \textbf{(Forming $\bv A'_{[:,S]}$)}
Form the matrix $\bv A'_{:,S} \in \mathbb{R}^{n \times K}$ such that for all $(j,q) \in [n] \times [K]$ and for a sufficiently large absolute constant $c>0$:
\[
\left(\bv A'_{:,S} \right)_{jq} = \begin{cases}
0 & \text{if }
\|\bv A_{j,:}\|_2^2 \|\bv A_{i_{q},:}\|_2^2
\leq \dfrac{\epsilon^2 \|\bv A\|_F^2 |\bv A_{j i_{q}}|^2}{c\log^4 n},\\[4pt]
\frac{\bv A_{j i_{q}}}{\sqrt{sp_{i_q}}} & \text{otherwise. }
\end{cases}
\]
\State \label{step:zero1} \textbf{(Forming $\bv A'_{S,S}$)} Form the matrix $\bv A'_{S,S} \in \mathbb{R}^{K \times K}$ such that for all $(p,q) \in [K] \times [K]$ and and for the same constant $c$ in the previous step:
\[
\left(\bv A'_{S,S} \right)_{pq} = \begin{cases}
0 & \text{if } p = q,\\[4pt]
0 & \text{if } p \neq q \text{ and }
\|\bv A_{i_{p},:}\|_2^2 \|\bv A_{i_{q},:}\|_2^2
\leq \dfrac{\epsilon^2 \|\bv A\|_F^2 |\bv A_{i_{p}i_{q}}|^2}{c\log^4 n},\\[4pt]
\frac{\bv A_{i_{p}i_{q}}}{s\sqrt{p_{i_p}p_{i_q}}} & \text{otherwise, }
\end{cases}.
\]

\State Compute unit norm eigenvectors $\bv x_1, \bv x_2, \ldots, \bv x_{|S|}$
of $\bv A'_{S,S}$
corresponding to its eigenvalues
$\lambda_1 \ge \lambda_2 \ge \cdots \ge \lambda_{|S|}$.

\State For all $i$ with $|\lambda_i| \geq \frac{\epsilon \|\bv A\|_F}{2}$,
compute the approximate eigenvector
$\bv v_i = \frac{\left(\bv A'_{:,S} \right)\bv x_i}{\left\|\left(\bv A'_{:,S} \right) \bv x_i \right\|_2}.$

\State {\bfseries Return:} All eigenpairs $(\lambda_i, \bv v_i)$ computed in the previous step.
\end{algorithmic}
\end{algorithm}

\begin{algorithm}[t]
\caption{Restricted Squared Column Norm Sampling}
\label{alg:rest-eigvec-sqsamp}
\begin{algorithmic}[1]
\State {\bfseries Input:} Symmetric $\bv A \in \mathbb{R}^{n\times n}$,
row norms $\{\|\bv A_{i,:}\|_2\}_{i=1}^n$,
accuracy $\epsilon \in (0,1)$,
expected sample size $s$ such that $\frac{s\|\bv A_i \|^2_2}{\|\bv A \|_F^2} \leq 1$ for $i \in [n]$.

\State Let $\Sbb$ be the scaled sampling matrix
which samples column $i$
with probability $p_i=\frac{\|\bv A_{i,:}\|_2^2}{\|\bv A\|_F^2}$
and scales it by $\frac{1}{\sqrt{sp_i}}$.
\State \textbf{(Zeroing out)}
Implicitly form the matrix $\bv A'$ from $\bv A$ by setting
\[
\bv A'_{ij} = \begin{cases}
0 & \text{if } i = j \text{ and } \|\bv A_{i,:}\|_2^2 < \dfrac{\epsilon^2}{4}\|\bv A\|_F^2,\\[4pt]
0 & \text{if } i \neq j \text{ and }
\|\bv A_{i,:}\|_2^2 \|\bv A_{j,:}\|_2^2
\leq \dfrac{\epsilon^2 \|\bv A\|_F^2 |\bv A_{ij}|^2}{c\log^4 n},\\[4pt]
\bv A_{ij} & \text{otherwise, }
\end{cases}
\]
where $c > 0$ is a sufficiently large absolute constant.

\State Compute unit norm eigenvectors $\bv x_1, \bv x_2, \ldots, \bv x_{|S|}$
of $\Sbb^T \bv A' \Sbb$
corresponding to its eigenvalues
$\lambda_1 \ge \lambda_2 \ge \cdots \ge \lambda_{|S|}$.

\State For all $i$ with $|\lambda_i| \geq \frac{\epsilon \|\bv A\|_F}{2}$,
compute the approximate eigenvector
$\bv v_i = \frac{\bv A' \Sbb \bv x_i}{\| \bv A' \Sbb \bv x_i \|_2}.$

\State {\bfseries Return:} All eigenpairs $(\lambda_i, \bv v_i)$ computed in the previous step.
\end{algorithmic}
\end{algorithm}

We now analyze Algorithm~\ref{alg:eigvec-rowsamp}
which samples columns of $\bv A$ with
probability proportional to their squared column norms. This sampling model has been previously studied in classical work on randomized linear algebraic methods ~\cite{frieze2004fast,
drineas2004clustering, drineas2006fast}, and recently under the \emph{quantum-inspired} machine learning
literature~\cite{tang2019quantum, chia2020sampling, gilyen2018quantum,chepurko2022quantum}.
Instead of directly analyzing Algorithm~\ref{alg:eigvec-rowsamp},
following~\cite{swartworth2025tight},
we will first analyze
a slightly modified version of the algorithm,
Algorithm~\ref{alg:rest-eigvec-sqsamp}, where
we assume that $\frac{s\| \bv A_{i,:}^2\|_2^2}{\| \bv A \|_F^2} \leq 1$ for all $i \in [n]$
for our choice of $s$.
Note that we can make a new matrix where
the rows/columns of $\bv A$ are repeated
$M$ times and scaled by $\frac{1}{\sqrt{M}}$ for
some $M \geq s$
to ensure that
the assumption $\frac{s\| \bv A_{i,:}^2\|_2^2}{\| \bv A \|_F^2} \leq 1$
for all rows and columns of this new matrix.
As we will show in Theorem~\ref{thm:sqnorm1}, this new \emph{inflated}
matrix has the same nonzero spectrum as that of $\bv A$ and can be used
to approximate the eigenvectors and eigenvalues of $\bv A$. Basically, we will
show that running Algorithm~\ref{alg:eigvec-rowsamp} with
$\bv A$ as input is almost the same
as running Algorithm~\ref{alg:rest-eigvec-sqsamp}
with the \emph{inflated} matrix as input.

We first state some imported results from~\cite{Bhattacharjee:2021wl}
and~\cite{swartworth2025tight}. We will again use
same definitions of $\bv A_o$ and $\bv A_m$ as in
Definition~\ref{def:ao-am}
but with the split at $L=\epsilon\sqrt{\delta}\|\bv A \|_F$.
Following~\cite{Bhattacharjee:2021wl}, we will
use the following entrywise zeroing out procedure to
ensure the middle eigenvalues are bounded after sampling:
\begin{definition}[Definition 7.1
of~\cite{swartworth2025tight}]\label{def:zeroing}
Given a symmetric matrix $\bv A \in \mathbb{R}^{n \times n}$,
let $\bv A'$ be the matrix formed by zeroing out all entries $A_{ij}$ satisfying one of the following conditions:
\begin{enumerate}
 \item $i = j$ and $\|\bv A_{i,:}\|_2^2 < \dfrac{\epsilon^2}{4}\|\bv A\|_F^2$,
 \item $i \neq j$ and $\|\bv A_{i,:}\|_2^2 \|\bv A_{j,:}\|_2^2 < \dfrac{\epsilon^2 \|\bv A\|_F^2 |A_{ij}|^2}{c \log^4 n}$,
\end{enumerate}
for a sufficiently large absolute constant $c$.
\end{definition}
In Lemma 11 of~\cite{Bhattacharjee:2021wl}, it is proven
that $\bv A'$ is spectrally close to $\bv A$:
\begin{implemma}[From Lemma 11 of~\cite{Bhattacharjee:2021wl}]\label{lem:aprime}
For $\bv A'$ as defined in~\ref{def:zeroing}, we have:
\begin{align*}
 \|\bv A'-\bv A \|_2 \leq \epsilon \| \bv A\|_F.
\end{align*}
\end{implemma}
Thus, for any unit vector $\bv v$,
it is enough to show that $\|\bv A' \bv v-\lambda \bv v \|_2 \leq \epsilon n$
to claim that it also satisfies a similar bound for
$\bv A$ after adjusting by constants.
Hence, in the subsequent results,
we will prove results prove everything using $\bv A'$.
Most of the proofs in this section
are very similar to those in Sections~\ref{sec:asx2}
and~\ref{app:uniform samp}.
So we will only provide sketches of the proofs and refer to the proofs in
Section~\ref{sec:asx2} and~\ref{app:uniform samp} for more details.

As in Section~\ref{sec:asx2}, first we show that $\| \bv A \Sbb \bv x\|_2$
is close to $|\lambda|$ where $\lambda$ is the
eigenvalue of $\Sbb^T \bv A' \Sbb$
corresponding to the eigenvector $\bv x$.
We start by applying Lemma~\ref{Lem:aosum}
to the squared column-norm sampling setting.
Note that we define $\bv A'_o$ and $\bv A'_m$
according to definition~\ref{def:ao-am}
by setting the threshold as $L=\epsilon \sqrt{\delta}\|\bv A \|_F$.
\begin{lemma}\label{Lem:spec_norm2}
Let $\Sbb$ be the scaled sampling matrix
which samples column $i$ of $\bv A$ with
probability $p_i=\min\left(1, \frac{s\|\bv A_{i,:} \|^2_2}{\|\bv A \|^2_F}\right)$
 and rescales it by $\frac{1}{\sqrt{p_i}}$.
Let $\bv A'$ be the matrix with zeroed out entries as defined in
Definition~\ref{def:zeroing}.
Then, for $s \geq \frac{c}{\epsilon^2 \delta}
\log^3 \frac{1}{\epsilon}$, with probability
at least $1-\delta$, we have:
\begin{align*}
\|\Sbb^T\Ab'_o \Sbb \Sbb^T \Ab'_o\Sbb-\Sbb^T(\Ab_o')^2\Sbb \|_2
\leq \epsilon \|\bv A \|_F \|\bv A \|_2
\end{align*}
\end{lemma}
\begin{proof}
Observe that $r=\lceil \log(\|\bv A \|_2/\epsilon
\sqrt{\delta}\|\bv A \|_F) \rceil \leq \lceil \log(1/\epsilon \sqrt{\delta}) \rceil$.
The proof follows from Imported Lemma~\ref{lem:subspace-embedding-norm}
and Lemma~\ref{Lem:aosum}
after adjusting $\epsilon$
by $\log (1/\epsilon \delta)$ and constants.
\end{proof}

Similar to Lemma~\ref{lem:lambda_bound}, we can show that $\bv x$ is
an approximate eigenvector of $(\Sbb^T \Ab'_o \Sbb^T)^2$
with eigenvalue $\lambda^2$:
\begin{lemma}\label{lem:lambda_bound1}
Consider the setting of Lemma~\ref{Lem:spec_norm2}.
Let $\bv x$ be an eigenvector of $\Sbb^T \Ab' \Sbb^T$ with eigenvalue $\lambda$
such that $|\lambda| \geq \frac{\epsilon \|\bv A \|_F}{2}$.
Then, for some $s \geq \frac{c \log^4 n }{\epsilon^2\delta} \log^3 \frac{1}{\epsilon \delta}$
for some sufficiently large constant $c$, if
$\frac{s \|\bv A_{i,:} \|_2^2}{\|\bv A \|_F^2} \leq 1$ for all $i \in [n]$,
with probability at least $1-\delta$, we have:
 $$|\bv x^T(\Sbb^T \Ab'_o \Sbb)^2\bv x -\lambda^2| \leq \epsilon \|\bv A \|_F|\lambda|$$
\end{lemma}
\begin{proof}
The proof is exactly the same as that of Lemma~\ref{lem:lambda_bound},
with the only difference being we use Imported Lemma~\ref{lem:ao-sampling}.
So we skip it.
\end{proof}

Next, we show that $\|\bv A'_m \Sbb \|_2$ is bounded
by $\epsilon \|\bv A \|_F$. We will use this to bound the
contribution of $\bv A'_m$ to $\|\bv A' \Sbb \bv x \|_2$.
\begin{lemma}\label{lem:as_norm}
Consider the setting of Lemma~\ref{Lem:spec_norm2}.
For $s \geq \frac{c\log^2 n}{\epsilon^2\delta}$ where $c$
is a sufficiently large constant, if
$\frac{s \|\bv A_{i,:} \|_2^2}{\|\bv A \|_F^2} \leq 1$ for all $i \in [n]$,
with probability at least $1-\delta$, we have that:
\begin{align*}
 \| \bv{A}'_m \bv{\bar S} \|_2 \leq \epsilon \|\bv A \|_F.
\end{align*}
\end{lemma}
\begin{proof}
The proof follows from Lemma 14 of~\cite{Bhattacharjee:2021wl}.
We provide a brief proof sketch.
From Theorem 5 of~\cite{Bhattacharjee:2021wl}, we have:
\begin{align*}
 \E_2 \|\bv{A'_m}\Sbb \|_2 \leq 5\sqrt{\log n} \cdot\E_2 \|\bv{A'_m}\Sbb \|_{1 \rightarrow 2}+ \|\bv{A_m}'\|_2.
\end{align*}
We have $\|\bv{A_m}'\|_2\leq \epsilon \sqrt{\delta}\|\bv A\|_F$.
We also have:
\begin{align*}
 E_2 \|\bv{A'_m }\Sbb \|_{1 \rightarrow 2}
\leq \max_{i \in [n]}\frac{\|(\bv A'_m)_{i,:} \|_2}{\sqrt{p_i}}
\leq \max_{i \in [n]} \frac{\|\bv A \|_F \|(\bv A'_m)_{i.:} \|_2}{\sqrt{s}\|\bv A_{i,:} \|_2}
\leq \frac{\|\bv A \|_F}{\sqrt{s}}.
\end{align*}
Thus, for $s \geq \frac{ \log^2 n}{\epsilon^2 \delta}$, we have
$E_2 \|\bv{A'_m }\Sbb \|_{1 \rightarrow 2} \leq \epsilon\sqrt{\delta}\|\bv A \|_F $.
We get the final bound by adjusting $\epsilon$ by constant factors.
\end{proof}

We now show that $\|\bv A' \Sbb \bv x \|_2$ is close to $\|\bv A \|_2$
when $\bv x$ is the eigenvector corresponding to the largest eigenvalue
of $\Sbb^T \Ab' \Sbb$.
\begin{lemma}\label{lem:asx_norm}
Consider the setting of Lemma~\ref{Lem:spec_norm2}
and
let $\|\bv A' \|_2 \geq \frac{\epsilon \|\bv A \|_F}{4}$.
Let $\bv{x}$ be the eigenvector of $\Sbb^T \Ab' \Sbb$ corresponding to
its largest magnitude eigenvalue.
For $s \geq \frac{c \log^4 n}{\epsilon^2\delta} \log^3 \frac{1}{\epsilon \delta}$,
for a sufficiently large constant $c$, if
$\frac{s \|\bv A_{i,:} \|_2^2}{\|\bv A \|_F^2} \leq 1$ for all $i \in [n]$,
with
probability at least $1-\delta$, we have (for some constant $C$):
\begin{align*}
 \frac{1}{C}\|\bv A \|_2 \leq \|\bv A' \Sbb \bv x \|_2 \leq C|\bv A \|_2
\end{align*}
\end{lemma}
\begin{proof}
From Imported Theorem~\ref{thm:main-eigval-norm},
choosing $\epsilon$ to be sufficiently small,
we have that,
for the largest magnitude eigenvalue $\Tilde{\lambda}_1$ of $\Sbb^T \Ab' \Sbb$,
we get
$||\Tilde{\lambda}_1| - \|\bv A \|_2|\leq \frac{\epsilon \|\bv A \|_F}{8} \leq \frac{\|\bv A \|_2}{2}$.
The rest of the proof follows the proof of Lemma~\ref{lem:asx}.
\end{proof}

When $\bv{x}$ be an eigenvector of
$\bv{\bar S}^T \bv{A}' \bv{\bar S}$
corresponding to an outlying eigenvalue $\lambda$,
we can also show that $\|\bv A\bv \Sbb \bv x\|_2$ is close to $|\lambda|$ like Lemma~\ref{Lem:asxl}.
\begin{lemma}\label{lem:asxl_norm}
Consider the setting of Lemma~\ref{Lem:spec_norm2}.
Let $\bv{x}$ be an eigenvector of $\bv{\bar S}^T \bv{A}' \bv{\bar S}$
corresponding to its eigenvalue $\lambda$ such that $|\lambda| \geq \epsilon \|\bv A \|_F$.
For $s \geq \frac{c \log^4 n}{\epsilon^4 \delta}\log^3 \frac{1}{\epsilon \delta} $,
for a sufficiently large constant $c$, if
$\frac{s \|\bv A_{i,:} \|_2^2}{\|\bv A \|_F^2} \leq 1$ for all $i \in [n]$,
with
probability at least $1-\delta$, we have (for some constant $C$):
\begin{align*}
 \frac{1}{C}|\lambda| \leq \|\bv A' \Sbb \bv x \|_2 \leq C|\lambda|
\end{align*}
\end{lemma}

We now state the key lemma regarding approximating
the top eigenvector using squared column-norm sampling.
\begin{lemma}\label{lem:err_top_norm}
Consider the setting of Lemma~\ref{Lem:spec_norm2}.
Let $\bv{x}$ be the eigenvector of $\Sbb^T \Ab' \Sbb$ corresponding to
its largest magnitude eigenvalue $\lambda$ such that $|\lambda| \geq \frac{\epsilon \|\bv A \|_F}{2}$.
Let $\bv{v}=\frac{\Ab' \Sbb \xb}{\|\Ab' \Sbb \xb \|_2}$.
Then, for $s \geq \frac{c\log^4 n}{\epsilon^2\delta} \log^5 \frac{1}{\epsilon \delta}$,
if $\frac{s \|\bv A_{i,:} \|_2^2}{\|\bv A \|_F^2} \leq 1$ for all $i \in [n]$,
with probability at least $1-\delta$, we have:
\begin{equation*}
 \|\Ab' \bv{v}-\lambda \bv{v}\|_2 \leq \epsilon \|\bv A \|_F.
\end{equation*}
\end{lemma}
\begin{proof}
The proof is almost exactly the same as
the proof of Lemma~\ref{lem:err_top}.
We recall the key steps of the proof for completeness. We have:
\begin{align*}
 \bv A' \bv v=\bv A'_o \bv v+\bv e',
\end{align*}
where $\| \bv e'\|_2 \leq \epsilon \sqrt{\delta}\|\bv A \|_F$.
So, it is enough to prove $\|\Ab'_o \bv{v}-\lambda \bv{v}\|_2
\leq \epsilon \|\bv A \|_F$.
Using Lemma~\ref{Lem:spec_norm2} and Imported Lemma~\ref{lem:ao-sampling}, similar to~\eqref{eq:spliteq}
we get:
\begin{align*}
 \bv A'_o \bv v &= \lambda \bv v-\frac{\lambda \bv A'_m \Sbb \bv x}{\|\Ab' \Sbb \bv x\|_2}
 -\frac{\sum_{i=1}^r \sum_{j \neq i} \Ab'_{o,i}\Sbb \Sbb^T \bv A'_{o,j} \Sbb \bv x}{\|\Ab' \Sbb \bv x\|_2}
 -\frac{\sum_{i=1}^r \Ab'_{o,i}\Sbb \Sbb^T \bv A'_{m} \Sbb \bv x}{\|\Ab' \Sbb \bv x\|_2} \nonumber\\
 &+ \frac{\sum_{i=1}^r \bv U'_{o,i} \bv E_i (\bv U'_{o,i})^T \Bar{\Sb} \bv x}{\|\Ab' \Sbb \bv x\|_2}.
\end{align*}
Here the eigenvalues of $\bv A'_{o,i}$ (for $i \in \log (\|\bv A' \|_2/\epsilon \sqrt{\delta} \|\bv A \|_F)$)
are in $[\|\bv A'_o\|_2 \sqrt{\delta}2^{-i},\| \bv A'_o\|_2 \sqrt{\delta} 2^{-i+1}]$
and $\| \bv E_i \|_2 \leq \epsilon \sqrt{\delta}\|\bv A \|_F \|\bv A \|_2 \sqrt{\delta} 2^{-i+2}$.
We can now bound each term on the right above by $O(\epsilon \|\bv A \|_F
\log (1/\epsilon))$
using Lemmas~\ref{lem:as_norm},~\ref{lem:asx_norm}, and~\ref{Lem:spec_norm2}
as in the proof of Lemma~\ref{lem:err_top}.
We get the final bound after adjusting $\epsilon$ by
$\log (1/\epsilon)$ and constant factors.
\end{proof}

Similar to the uniform sampling setting,
we can use Algorithm~\ref{alg:rest-eigvec-sqsamp}
to approximate all eigenvectors corresponding
to large eigenvalues at the cost of
a $\Tilde{O}(1/\epsilon^2)$ factor
increase in sampling complexity. We state the results
formally below.
\begin{lemma}\label{lem:eignorm}
Consider the setting of Lemma~\ref{Lem:spec_norm2}.
Let $\bv{x}$ be an eigenvector of $\bv{\bar S}^T \bv{A}' \bv{\bar S}$
corresponding
to its eigenvalue $\lambda$ such that $|\lambda| \geq \epsilon n$.
Let $\bv{v}=\frac{\Ab' \Sbb \xb}{\|\Ab' \Sbb \xb \|_2}$.
Then, for $s \geq \frac{c\log^4 n}{\epsilon^4\delta} \log^3 \frac{1}{\epsilon \delta}$,
if $\frac{s \|\bv A_{i,:} \|_2^2}{\|\bv A \|_F^2} \leq 1$ for all $i \in [n]$,
with probability at least $1-\delta$, we have:
\begin{equation*}
 \|\Ab' \bv{v}-\lambda \bv{v}\|_2 \leq \epsilon \|\bv A \|_F
\end{equation*}
\end{lemma}
\begin{proof}
The proof follows exactly the same structure as that of Lemma~\ref{lem:err1}
by writing $(\bv U'_o)^T \bv U'_o$ as $(\bv U'_o)^T\Sbb \Sbb^T \bv U'_o+\bv E$
where $\|\bv E \|_2\leq \epsilon$ from the error bound in Lemma~\ref{lem:subspace-embedding-norm}
by setting $L=\epsilon^2 \|\bv A \|_F$.
We also use Lemma~\ref{lem:asxl_norm} to bound the denominator $\|\bv A' \Sbb \bv x \|_2$.
We omit the details here
\end{proof}

We are now ready to state the main results the top
eigenvector approximation using squared column-norm sampling.
First, we state the eigenvector approximation result using
Algorithm~\ref{alg:rest-eigvec-sqsamp}
assuming we have $\frac{s\|\bv A_i \|^2_2}{\| \bv A\|_F^2} \leq 1$
for each $i \in [n]$.
\begin{lemma}\label{lem:sq1}
Let $\bv A \in \R^{n \times n}$ be a symmetric matrix,
with eigenvalues $\lambda_1(\bv A) \geq \ldots \geq \lambda_n(\bv A)$,
and let $\epsilon, \delta \in (0,1)$.
Let Algorithm~\ref{alg:rest-eigvec-sqsamp} output the pairs
$\{(\tilde \lambda_j, \bv v_j)\}_{j=1}^{m}$, where $\tilde \lambda_j \in \R$, $\bv v_j \in \R^n$,
$\|\bv v_j \|_2 = 1$, and let $m_+$ (resp.\ $m_-$) be the number of output pairs with
$\tilde\lambda_j > 0$ (resp.\ $\tilde\lambda_j < 0$), so that $m = m_+ + m_-$.
Then, if $s \geq
\frac{c \log^4 n}{\epsilon^4\delta} \log^3 \frac{1}{\epsilon \delta}$ for some constant $c$
and
$\frac{s \|\bv A_{i,:} \|_2^2}{\|\bv A \|_F^2} \leq 1$ for all $i \in [n]$,
with probability at least $1-\delta$:
\begin{enumerate}
\item There is a bijection $\pi$ between the output pairs and the $m_+$ largest together with the
$m_-$ smallest eigenvalues of $\bv A$, i.e.\ the set
$\{\lambda_1(\bv A), \ldots, \lambda_{m_+}(\bv A)\} \cup \{\lambda_{n-m_-+1}(\bv A), \ldots, \lambda_n(\bv A)\}$,
such that every output pair $j \in [m]$ satisfies
\begin{align*}
 |\lambda_{\pi(j)}(\bv A)-\tilde \lambda_{j}| \leq \epsilon \| \bv A\|_F \text{, and } \quad
 \|\bv A \bv v_{j}-\lambda_{\pi(j)}(\bv A) \bv v_{j} \|_2 \leq \epsilon \| \bv A\|_F.
\end{align*}
\item Every eigenvalue of $\bv A$ outside this matched multiset has magnitude less than $\epsilon \|\bv A \|_F$.
\end{enumerate}
Moreover, Algorithm~\ref{alg:rest-eigvec-sqsamp} samples
$\frac{c \log^4 n}{\epsilon^4 \delta}
\log^3 \frac{1}{\epsilon \delta}$
columns from $\bv A$ in expectation,
for some sufficiently large constant $c$.
\end{lemma}
\begin{proof}
We use Imported Theorem~\ref{thm:main-eigval-norm}
for proving the eigenvalue approximation result along with
Lemma~\ref{lem:eignorm} for bounding the eigenvector approximation error.
We omit the details as the proof follows
exactly the same steps as the
proof of Theorem~\ref{thm:main-eigval-norm}.
\end{proof}

Now, we will relax the $\frac{s\|\bv A_i \|^2_2}{\| \bv A\|_F^2} \leq 1$ assumption and use Algorithm~\ref{alg:eigvec-rowsamp} to approximate the eigenvectors for any symmetric matrix. We accomplish this in Theorem~\ref{thm:sqnorm1} via two main steps: first, following~\cite{swartworth2025tight}, we analyze a binomial sampling scheme applied to an \emph{inflated} matrix $\bv A_M$---where each column of $\bv A$ is duplicated $M$ times and scaled by $1/\sqrt{M}$ to ensure all sampling probabilities are at most $1$. Second, we transfer this error guarantee to the Poisson sampling scheme used in Algorithm~\ref{alg:eigvec-rowsamp} by bounding the total variation distance between the two sampling distributions (Lemma~\ref{lem:TV}).
Note that in Algorithm~\ref{alg:eigvec-rowsamp},
we first sample a Poisson distribution with mean $s$
$K \sim \mathrm{Poisson(s)}$, and then sample $K$ times independently from
the squared column-norm distribution.
We use this sampling scheme, instead of simply sampling each column $i$
$\mathrm{Binomial}\left(M, \frac{s\| \bv A_i\|^2_2}{M\|\bv A \|_F^2} \right)$
times as in~\cite{swartworth2025tight}, so that we can adapt
Algorithm~\ref{alg:eigvec-rowsamp} to the
quantum-inspired framework in Section~\ref{sec:qram}.
In this framework, we can query entries of the matrix and
draw independent samples according to the distribution
of the squared column-norms (see Section~\ref{sec:qram} for details).
The most natural thing for adapting our algorithm
to this quantum-inspired framework
would be to
just sample $s$ times from the squared column-norm distribution.
However, for the bounds from~\cite{tropp2008norms}, which
are used to bound the
small magnitude eigenvalues of the sampled principal submatrix
in~\cite{swartworth2025tight,Bhattacharjee:2021wl}
(Imported Lemma~\ref{lem:ao-sampling}), it is crucial
that each row and column is sampled independently
of other rows and columns. This won't
be the case if we simply sample $s$ times from
the squared column-norm distribution. We briefly explain the sampling scheme based on the Poisson distribution below and prove some results which will help in proving Theorem~\ref{thm:sqnorm1}.

\paragraph{Poisson Sampling:} Note that the two-step sampling in Algorithm~\ref{alg:eigvec-rowsamp}
is exactly the same as sampling index $i$ $\mathrm{Poisson}(sp_i)$
times independently of other indices (Lemma~\ref{lem:pois})

.
Let $n_i \sim \mathrm{Poisson}(sp_i)$.
By the Poisson splitting property,
$(n_1, \ldots, n_n)$ are \emph{mutually independent} (Lemma~\ref{lem:pois}).
Thus, we can use Algorithm~\ref{alg:eigvec-rowsamp} in the quantum-inspired
model.

To transfer the error guarantee from Binomial to Poisson sampling, we again use the inflated-matrix argument described above:
we compare $\mathrm{Poisson}(sp_i)$ to $\mathrm{Binomial}(M, sp_i/M)$,
which arises from sampling the inflated matrix $\bv Q_M \bv A \bv Q_M^T$
where each row/columns of $\bv A$ is repeated $M$ times
and then scaled down by $\frac{1}{\sqrt{M}}$.
(see the proof of Theorem~\ref{thm:sqnorm1}).
By Lemma~\ref{lem:TV},
the TV distance between the two joint distributions of $\mathrm{Poisson}(sp_i)$
and $\mathrm{Binomial}(M, sp_i/M)$
is at most $s^2/M$,
which is bounded by $\delta$ for $M \geq \frac{s^2}{\delta}$.

We now state the lemma bounding the Total Variation (TV) distance between
these distributions.

\begin{lemma}[TV Distance Bound]\label{lem:TV}
Let $P_M$ be the joint distribution of the independent counts $n_i^{(M)} \sim \mathrm{Binomial}(M, \frac{sp_i}{M})$ for all $i \in [n]$. Let $P_{Pois}$ be the joint distribution of independent random variables $n_i \sim \mathrm{Poisson}(sp_i)$ for all $i \in [n]$. Then the Total Variation distance is bounded by:
\[ d_{TV}(P_M, P_{Pois}) \le \frac{s^2}{M} \]
\end{lemma}
\begin{proof}
By Le Cam's inequality~\cite{le1960approximation}, the TV distance between a $\mathrm{Binomial}(M, p)$ and $\mathrm{Poisson}(Mp)$ distribution is at most $M p^2$. Because the variables are mutually independent across $i \in [n]$ in both distributions, the joint TV distance is bounded by the sum of the marginal TV distances:
\[ d_{TV}(P_M, P_{Pois}) \le \sum_{i=1}^n M \left(\frac{sp_i}{M}\right)^2 = \frac{s^2}{M} \sum_{i=1}^n p_i^2 \]
Since $p_i$ is a probability distribution ($\sum_{i=1}^n p_i = 1$), we have $\sum_{i=1}^n p_i^2 \le 1$, yielding $d_{TV} \le \frac{s^2}{M}$.
\end{proof}

We next state the Poisson splitting property which guarantees that the columns
are sampled mutually independent of each other.
This is a standard property of the Poisson distribution.
A $\mathrm{Multinomial}(K, p_1, \dots, p_n)$ distribution where the
number of trials $K$ is distributed as $\mathrm{Poisson}(s)$ results
in marginal counts $n_i$ that are mutually independent and
distributed exactly as $\mathrm{Poisson}(sp_i)$~\cite{mitzenmacher2017probability}.
\begin{lemma}[Poisson Splitting]\label{lem:pois}
Let $(p_1, \ldots, p_n)$ with $\sum_i p_i=1$ and $p_i \geq 0$ be a distribution over n indices such that the probability of drawing the $i$\textsuperscript{th} index is $p_i$. Let $P_{Pois}$ be a distribution defined in the following way: we draw $K \sim \mathrm{Poisson}(s)$ and then independently draw $K$ samples from the distribution $(p_1, \ldots, p_n)$. Let $n_i$ be the total number of times index $i$ is drawn. Then, each $n_i$ is distributed as $\mathrm{Poisson}(sp_i)$. The joint distribution of $(n_1, \dots, n_n)$ is exactly $P_{Pois}$.
\end{lemma}

\paragraph{Final Bounds.} We are now ready to state our main result for eigenvector
approximation using Algorithm~\ref{alg:eigvec-rowsamp}.
\begin{theorem}\label{thm:sqnorm1}
Let $\bv A \in \R^{n \times n}$ be a symmetric matrix,
with eigenvalues $\lambda_1(\bv A) \geq \ldots \geq \lambda_n(\bv A)$,
and let $\epsilon, \delta \in (0,1)$.
Let Algorithm~\ref{alg:eigvec-rowsamp} output the pairs
$\{(\tilde \lambda_j, \bv v_j)\}_{j=1}^{m}$ with $\bv A$ and $s=\frac{c \log^4 n}{\epsilon^4 \delta}
\log^3 \frac{1}{\epsilon \delta}$ as inputs (for a sufficiently large constant $c$), where $\tilde \lambda_j \in \R$, $\bv v_j \in \R^n$,
$\|\bv v_j \|_2 = 1$, and let $m_+$ (resp.\ $m_-$) be the number of output pairs with
$\tilde\lambda_j > 0$ (resp.\ $\tilde\lambda_j < 0$), so that $m = m_+ + m_-$. Then,
with probability at least $1-\delta$:
\begin{enumerate}
\item There is a bijection $\pi$ between the output pairs and the $m_+$ largest together with the
$m_-$ smallest eigenvalues of $\bv A$, i.e.\ the set
$\{\lambda_1(\bv A), \ldots, \lambda_{m_+}(\bv A)\} \cup \{\lambda_{n-m_-+1}(\bv A), \ldots, \lambda_n(\bv A)\}$,
such that every output pair $j \in [m]$ satisfies
\begin{align*}
 |\lambda_{\pi(j)}(\bv A)-\tilde \lambda_{j}| \leq \epsilon \| \bv A\|_F \text{, and } \quad
 \|\bv A \bv v_{j}-\lambda_{\pi(j)}(\bv A) \bv v_{j} \|_2 \leq \epsilon \| \bv A\|_F.
\end{align*}
\item Every eigenvalue of $\bv A$ outside this matched multiset has magnitude less than $\epsilon \|\bv A \|_F$.
\end{enumerate}
Moreover, Algorithm~\ref{alg:eigvec-rowsamp} samples
$O\left(\frac{ \log^4 n}{\epsilon^4 \delta}
\log^3 \frac{1}{\epsilon \delta} \right)$
columns from $\bv A$ in expectation.
\end{theorem}
\begin{proof}
From Lemma~\ref{lem:sq1}, the guarantee holds for if we run Algorithm~\ref{alg:rest-eigvec-sqsamp}
with any matrix
with $\frac{s\|\bv A_i \|^2_2}{\|\bv A \|^2_F} \leq 1$ for all $i \in [n]$ as input.
We now relax this assumption by using a
similar argument as in~\cite{swartworth2025tight} of
making multiple copies if the rows and columns of $\bv A$ and then scaling down the norms.

\noindent\textbf{Reduction to Binomial sampling.} Let $\bv Q_M$ be a vertical stack of consisting of
$M =\frac{s^2}{\delta}$ copies of
$\frac{1}{\sqrt{M}}\bv I_n$.
Now, consider the matrix $\bv Q_M \bv A \bv Q_M^T$.
Note that~\cite{swartworth2025tight} make this argument for $M=s$
while we argue for this for $M =\frac{s^2}{\delta}$,
the reasons for which will become clear.
The nonzero spectrum of $\bv Q_M \bv A \bv Q_M^T$
coincides with the nonzero spectrum of $\bv A$ and $\|\bv Q_M \bv A \bv Q_M^T \|_F=\|\bv A \|_F$.
Also, for any row $i$ of $\bv Q_M \bv A \bv Q_M^T$, we have:
 $$\frac{s\|(\bv Q_M \bv A \bv Q_M^T)_{i,:} \|_2^2}{\|\bv Q_M \bv A \bv Q_M^T \|_F^2}
=\frac{1}{M}\frac{s\|\bv A_{i',:} \|_2^2}{\|\bv A \|^2_F} = \frac{sp_i}{M}$$ for some
row $i'$ of $\bv A$. Now for any $M>s$, the row norms of
$\bv Q_M \bv A \bv Q_M^T$ are bounded by $1$.
Thus, Lemma~\ref{lem:sq1} should hold for
sampling from $\bv Q \bv A \bv Q^T$ when sample column $i$ with probability
$\frac{s\|(\bv Q_M \bv A \bv Q_M^T)_{i,:} \|_2^2}{\|\bv Q \bv A \bv Q^T \|_F^2}=\frac{sp_i}{M}$.
Since $\bv Q_M \bv A \bv Q_M $ has $nM=\frac{ns^2}{\delta}$ columns,
in order to achieve error $\epsilon \|\bv Q_M \bv A \bv Q_M \|_F=\epsilon \|\bv A \|_F$,
we need to sample $s \geq \frac{c\log^4 (ns^2/\delta )}{\epsilon^4 \delta}\log ^3 \frac{1}{\epsilon \delta}$
columns i.e. $s \geq \frac{c\log^4 n}{\epsilon^4 \delta}\log ^3 \frac{1}{\epsilon \delta}$
columns after adjusting the constant $c$.
To simulate Algorithm~\ref{alg:rest-eigvec-sqsamp}
with $\bv Q_M \bv A \bv Q_M^T $ as input, we can sample each column $i$ from $\bv A$
$\textrm{Binom}(M, \frac{sp_i}{M})$ many times.
Note that we will still draw $s$ columns in expectation.
Also, note that for the matrix $\bv Q_M \bv A \bv Q_M^T$, the diagonal zeroing out condition
in Algorithm~\ref{alg:rest-eigvec-sqsamp} is always satisfied since $s\geq \frac{c}{\epsilon^4}$, so exactly the copy positions
$(i',i')$ are zeroed; two \emph{distinct} copies of the same index of $\bv A$ form an
off-diagonal pair and are governed by the off-diagonal condition.
On the other hand, the off-diagonal zeroing condition in Algorithm~\ref{alg:rest-eigvec-sqsamp}
is still the same for $\bv Q_M \bv A \bv Q_M^T$.
Thus, the final algorithm applied to $\bv A$ is Algorithm~\ref{alg:eigvec-rowsamp},
but instead of the sampling scheme in Steps~\ref{step:pois} and~\ref{step:samp}, we sample each column
$i$, $\textrm{Binom}(M, \frac{sp_i}{M})$ times independently.
Under this correspondence Step~\ref{step:zero1} reproduces the sampled principal submatrix
of $\bv Q_M \bv A \bv Q_M^T$ exactly, while Step~\ref{step:zero} differs only at the single
row-copy of $i_q$ per sampled slot $q$; writing $\bv D$ for this difference,
$\|\bv D\|_F^2 \le \sum_{q \in [K]} \frac{1}{M}\frac{\bv A_{i_q i_q}^2}{s p_{i_q}}
\le \frac{K \|\bv A\|_F^2}{Ms} \le \frac{2\delta \|\bv A\|_F^2}{s^2}$,
using $\bv A_{jj}^2 \le p_j \|\bv A\|_F^2$, $M=\frac{s^2}{\delta}$ and $K \le 2s$.
Since $s \geq \frac{c}{\epsilon^4}$, we get $\|\bv D\|_F = O(\epsilon^4 \|\bv A\|_F)$,
which is absorbed into the $\epsilon \|\bv A\|_F$ error.

\noindent\textbf{Transfer to Poisson sampling.}
Now, we show that using the sampling scheme of Algorithm~\ref{alg:eigvec-rowsamp}
should also work instead of picking column $i$ $\textrm{Binom}(M, \frac{sp_i}{M})$ times.
First, from the Poisson splitting argument in Lemma~\ref{lem:pois},
the sampling scheme from Algorithm~\ref{alg:eigvec-rowsamp} is
equivalent to picking column each $i$ $\mathrm{Poisson}(sp_i)$ many times.
Let $P_M$ be the joint distribution of the independent counts
$n_i^{(M)} \sim \mathrm{Binomial}(M, \frac{sp_i}{M})$ for all $i \in [n]$.
Let $P_{Pois}$ be the joint distribution of
independent random variables $n_i \sim \mathrm{Poisson}(sp_i)$ for all $i \in [n]$.
Then, from Lemma~\ref{lem:TV}, we have $ d_{TV}(P_M, P_{Pois}) \le \frac{s^2}{M}$.

Let $\mathcal{E}_{err}$ be the bad case when any approximate
eigenvector given by Algorithm~\ref{alg:eigvec-rowsamp}, with the binomial sampling scheme
specified above, has residual
error exceeding $\epsilon \|\mathbf{A}\|_F$. From our argument above,
we know that $\mathbb{P}_M(\mathcal{E}_{err}) \le \delta$.
By the definition of total variation distance, we have:
\[ \mathbb{P}_{Pois}(\mathcal{E}_{err}) \le \mathbb{P}_M(\mathcal{E}_{err})
+ d_{TV}(P_M, P_{Pois}) \le \delta + \frac{s^2}{M} \]
By setting $M=\frac{s^2}{\delta}$ we have
$\mathbb{P}_{Pois}(\mathcal{E}_{err}) \leq 2\delta$.
Thus, we have bounded the case when any approximate eigenvector
given by Algorithm~\ref{alg:eigvec-rowsamp} has error exceeding $2\delta$.
Finally, we bound the sample complexity of Algorithm~\ref{alg:eigvec-rowsamp}.
Note that we sample $K \sim \mathrm{Poisson}(s)$ columns. By standard Poisson tail bounds,
we always have $K<2s$ with probability at least $1-\delta$.
Thus, we still pick $O(s)$ columns in expectation. We get the final result
by taking a union bound over the complement of the events $\mathcal{E}_{err}$ under
Algorithm~\ref{alg:eigvec-rowsamp}
(i.e. the approximate eigenvector has bounded residual error) and
$K<2s$ and
adjusting $\epsilon$, $\delta$ by constant factors.

\end{proof}

Similar to Corollary~\ref{cor:local} for uniform sampling,
we state a corollary that Algorithm~\ref{alg:eigvec-rowsamp}
gives the first sublinear time algorithm to output a single entry of an approximate
eigenvector of $\bv A$ by reading $\tilde{O}(1/\epsilon^4)$ entries of $\bv A$.
First, we show that $\leq \|\bv A'_{[:,S]} \bv x \|_2 \approx |\lambda|$ in
Algorithm~\ref{alg:eigvec-rowsamp}. This is basically Lemma~\ref{lem:asxl_norm}
without the assumption $\frac{s\|\bv A_{i,:}\|_2^2}{\|\bv A\|_F^2} \le 1$ for all $i \in [n]$.

\begin{lemma}\label{lem:asxl_norm_unconditional}
Let $\bv A'_{S,S}$ be the $|S| \times |S|$ matrix and $\bv A'_{[:,S]}$
be the $n \times |S|$ matrix constructed in Steps~\ref{step:zero} and~\ref{step:zero1} of Algorithm~\ref{alg:eigvec-rowsamp}
with a symmetric $\bv A \in \R^{n \times n}$ as input.
Let $\bv{x}$ be an eigenvector of $\bv A'_{S,S}$ with eigenvalue $\lambda$
such that $|\lambda| \geq \epsilon \|\bv A \|_F$.
For $s \geq \frac{c \log^4 n}{\epsilon^4 \delta}\log^3 \frac{1}{\epsilon \delta} $,
for a sufficiently large constant $c$, with probability at least $1-\delta$, we have (for some constant $C$):
\begin{align*}
 \frac{1}{C}|\lambda| \leq \|\bv A'_{[:,S]} \bv x \|_2 \leq C|\lambda|
\end{align*}
\end{lemma}

\begin{proof}
We remove the assumption
$\frac{s\|\bv A_{i,:}\|_2^2}{\|\bv A\|_F^2} \le 1$ for all $i \in [n]$ using
the matrix inflation argument described in Theorem~\ref{thm:sqnorm1}.
We will now go over. the key steps.
Let $\hat{\bv A}=\bv Q_M \bv A \bv Q_M^T$ be the $nM \times nM$ \emph{inflated}
matrix where each row/column is repeated $M$ times
and scaled by $\frac{1}{\sqrt{M}}$ for $M=\frac{s^2}{\delta}$,
as described in Theorem~\ref{thm:sqnorm1}.
Then, as proved in Theorem~\ref{thm:sqnorm1},
if a column $i'$ of $\hat{\bv A}$ is the scaled copy of column $i$ of $\bv A$,
its squared column-norm sampling probability satisfies
$\frac{s\|\hat{\bv A}_{i',:}\|_2^2}{\|\hat{\bv A}\|_F^2}=\frac{sp_i}{M} \leq 1$,
where $p_i=\frac{\|\bv A_{i,:}\|_2^2}{\|\bv A\|_F^2}$; hence $\hat{\bv A}$
satisfies the assumption of Lemma~\ref{lem:asxl_norm}.
Let $\Sbb'$ be the scaled sampling matrix that samples column $i' \in [nM]$ of
$\hat{\bv A}$ with probability $s\,p_{i'}$, where
$p_{i'}=\frac{\|\hat{\bv A}_{i',:}\|_2^2}{\|\hat{\bv A}\|_F^2}$, and scales it by
$\frac{1}{\sqrt{s\,p_{i'}}}$. Let $S'$ be the set of indices of the columns sampled from
$\hat{\bv A}$, and let $S_1$ be the corresponding multiset of indices of
the columns of $\bv A$.

Let $\bv A'_{[:,S_1]}$ and $\bv A'_{S_1, S_1}$
be the zeroed-out, sampled, and scaled matrices formed from the multiset $S_1$
according to Steps~\ref{step:zero} and~\ref{step:zero1} of Algorithm~\ref{alg:eigvec-rowsamp}.
Following the proof of Theorem~\ref{thm:sqnorm1}, and because the off-diagonal
zeroing condition evaluates identically for $\hat{\bv A}$ and $\bv A$, the
$\frac1M$ from the entries of $\hat{\bv A}'$ cancels against the sampling scale
$\frac{1}{\sqrt{s\,p_{i'}}}=\sqrt{\tfrac{M}{sp_i}}$ entrywise
$\big(\tfrac1M(\bv A')_{ij}\cdot\sqrt{\tfrac{M}{sp_j}}
=\tfrac{1}{\sqrt M}\,(\bv A'_{[:,S_1]})_{ij}\big)$, giving
\begin{align}
 \bv A'_{S_1, S_1}=(\Sbb')^T\hat{\bv A}' \Sbb'
 \qquad\text{and}\qquad
 \hat{\bv A}' \Sbb' = \bv Q_M\,\bv A'_{[:,S_1]}+\bv D.
 \label{eq:inflation-identities}
\end{align}
The first identity means $\hat{\bv A}'$'s sampled principal submatrix is exactly
$\bv A'_{S_1,S_1}$. The second says $\hat{\bv A}' \Sbb'$ is just each row of
$\bv A'_{[:,S_1]}$ repeated $M$ times and scaled by $\frac{1}{\sqrt M}$, up to the
difference $\bv D$ bounded in the proof of Theorem~\ref{thm:sqnorm1}; since
$\bv Q_M$ has orthonormal columns, for \emph{any} unit vector $\bv y$ we have
\begin{align}
 \Big| \big\|\bv A'_{[:,S_1]}\,\bv y\big\|_2
 - \big\|\hat{\bv A}'\Sbb'\,\bv y\big\|_2 \Big|
 \le \|\bv D\|_F \le \frac{\sqrt{2\delta}}{s}\|\bv A\|_F.
 \label{eq:norm-preserved-13}
\end{align}

\paragraph{A count-measurable good event.}
For a sample multiset $S$ produced by counts $(n_1,\dots,n_n)$ (where $n_i$ is
the number of times column $i$ is drawn), we define the following event
based on the matrices $\bv A'_{S,S}$ and
$\bv A'_{[:,S]}$ from
Algorithm~\ref{alg:eigvec-rowsamp}:
\[
 \mathcal G := \Big\{\ \text{for every eigenpair } (\bv y, \mu)
 \text{ of } \bv A'_{S,S} \text{ with } |\mu|\ge \epsilon\|\bv A\|_F:\quad
 \tfrac1C|\mu| \le \|\bv A'_{[:,S]}\,\bv y\|_2 \le C|\mu|\ \Big\},
\]
for the constant $C$ of Lemma~\ref{lem:asxl_norm}. Since $\bv A'_{S,S}$,
$\bv A'_{[:,S]}$, and hence all of their eigenpairs are deterministic functions
of the counts $(n_1,\dots,n_n)$, the event $\mathcal G$ is a function of the
counts alone. Crucially, $\mathcal G$ does not reference any fixed eigenvalue,
so it is well defined under \emph{any} sampling scheme.

\paragraph{Bounding $\mathcal G$ under Binomial sampling.}
The multiset $S_1$ is exactly the multiset obtained by drawing each column $i$
of $\bv A$ a $\mathrm{Binomial}\!\big(M,\tfrac{sp_i}{M}\big)$ number of times,
independently across $i$; equivalently, it is the sample produced by running
Algorithm~\ref{alg:rest-eigvec-sqsamp} on $\hat{\bv A}$ with the matrix
$\Sbb'$ above. From Lemma~\ref{lem:asxl_norm},
$\|\hat{\bv A}'\Sbb'\bv y\|_2=\Theta(|\mu|)$ holds with probability at least $1-\delta$
 for every eigenpair $(\bv y,\mu)$ of
$(\Sbb')^T\hat{\bv A}'\Sbb'$ with $|\mu|\ge \epsilon\|\bv A\|_F$. By~\eqref{eq:inflation-identities},
these are exactly the outlying eigenpairs of
$\bv A'_{S_1,S_1}$, and by~\eqref{eq:norm-preserved-13},
$\|\bv A'_{[:,S_1]}\bv y\|_2=\|\hat{\bv A}'\Sbb'\bv y\|_2
\pm\frac{\sqrt{2\delta}}{s}\|\bv A\|_F=\Theta(|\mu|)$, since
$|\mu|\ge\epsilon\|\bv A\|_F$ and $s\ge\frac{c}{\epsilon^4}$. Writing $\mathbb{P}_M$ for the probability under the
Binomial counts, we have
\[
 \mathbb{P}_M(\neg\mathcal G) \le \delta.
\]

\paragraph{Transfer to Poisson sampling.}
Algorithm~\ref{alg:eigvec-rowsamp} draws its sample via the Poisson scheme,
which by Lemma~\ref{lem:pois} produces counts $n_i\sim\mathrm{Poisson}(sp_i)$
independently across $i$; let $\mathbb{P}_{Pois}$ denote their joint sampling distribution.
By
Lemma~\ref{lem:TV}, $d_{TV}(\mathbb{P}_M,\mathbb{P}_{Pois})\le \frac{s^2}{M}=\delta$
for $M=\frac{s^2}{\delta}$. Since $\mathcal G$ is a function of the counts,
\[
 \mathbb{P}_{Pois}(\neg\mathcal G)
 \le \mathbb{P}_M(\neg\mathcal G) + d_{TV}(\mathbb{P}_M,\mathbb{P}_{Pois})
 \le \delta + \delta = 2\delta.
\]
Thus, with probability at least $1-2\delta$, the sample $S$ produced by
Algorithm~\ref{alg:eigvec-rowsamp} lies in $\mathcal G$. The eigenpair
$(\bv x, \lambda)$ in the statement is an eigenpair of $\bv A'_{S,S}$ with
$|\lambda|\ge \epsilon\|\bv A\|_F$, hence one of the eigenpairs quantified over
in $\mathcal G$; therefore $\frac1C|\lambda|\le\|\bv A'_{[:,S]}\bv x\|_2\le C|\lambda|$.
Rescaling $\delta$ by a constant factor completes the proof.
\end{proof}

We now state the corollary for fast computation of the entries of the approximate eigenvectors of Theorem~\ref{thm:sqnorm1}.
\begin{corollary}\label{cor:sqnorm}
Let $\bv A \in \R^{n \times n}$ be a symmetric matrix such that
$\|\Ab \|_{\infty} \leq 1$, and let $\epsilon, \delta \in (0,1)$.
Then, there is an algorithm that,
given $\bv A$ and any index
$j \in [n]$, reads
$O(s^2)$
entries in expectation from $\bv A$ and takes $O(s^{\omega})$ time, where $s=O\left(\frac{\log^4 n}{\epsilon^4 \delta}
\log^3 \frac{1}{\epsilon \delta} \right)$ and $\omega$ is the matrix multiplication exponent, and outputs the pairs
$\{(\tilde \lambda_i, \bv v_{ij})\}_{i=1}^{m}$, where $\bv v_{ij}$ is the
$j$\textsuperscript{th} entry of a vector $\bv v_i \in \R^{n}$ with $\|\bv v_i \|_2 \in [\frac{1}{C}, C]$, for an absolute
constant $C > 1$. The pairs
$\{(\tilde \lambda_i, \bv v_{i})\}_{i=1}^{m}$
satisfy the error guarantees of Theorem~\ref{thm:sqnorm1}
with probability at least $1-\delta$.
\end{corollary}
\begin{proof}
The proof is basically the same as the proof of Corollary~\ref{cor:local}
so we omit the details.
To compute the $j$\textsuperscript{th} entry of an approximate eigenvector
$\bv v_j=\frac{\bv A'_{[j,S]} \bv x}{\|\bv A'_{[:,S]} \bv x \|_2}$, we
compute the numerator $\bv A'_{[j,S]} \bv x$ exactly, and for computing the denominator $\|\bv A'_{[:,S]} \bv x \|_2$, we use
Lemma~\ref{lem:asxl_norm_unconditional}.
\end{proof}

We now state the main theorem about approximating
the top eigenvector using squared column-norm sampling.
We need the result from Lemma~\ref{lem:err_top_norm} for this proof.
We omit the proof since it will be very similar to the
proof of Theorem~\ref{thm:sqnorm1}.
\begin{theorem}\label{thm:sqnorm2}
Let $\bv A \in \R^{n \times n}$ be a symmetric matrix such that $\|\Ab \|_{\infty} \leq 1$, with
eigenvalues $\lambda_1(\bv A) \geq \ldots \geq \lambda_n(\bv A)$, and let $\epsilon, \delta \in (0,1)$.
Set $s \geq \frac{c \log^4 n}{\epsilon^2 \delta}
\log^5 \frac{1}{\epsilon \delta}$ in Algorithm~\ref{alg:eigvec}
for some sufficiently large constant $c$. Let $\tilde\lambda_1$ and $\tilde\lambda_2$ be the largest
positive and smallest negative eigenvalues output by Algorithm~\ref{alg:eigvec}, and let $\bv v_1$
and $\bv v_2$ be their corresponding approximate (unit) eigenvectors, respectively. If no positive
(resp.\ negative) eigenvalue is output by Algorithm~\ref{alg:eigvec}, set $\tilde\lambda_1 = 0$
(resp.\ $\tilde\lambda_2 = 0$). Then, with probability at
least $1-\delta$, the following hold simultaneously:
\begin{enumerate}
\item If $\|\bv A \|_2 \geq 2\epsilon n$ and
$\lambda_1(\bv A) \geq \frac{\|\bv A\|_2}{2}$, then $\tilde \lambda_1>0$, and:
 \begin{align*}
 |\lambda_1(\bv A)-\tilde \lambda_1| \leq \epsilon \|\bv A \|_F \text{, and }\quad
 \|\bv A \bv v_1-\lambda_1(\bv A) \bv v_1 \|_2 \leq \epsilon \|\bv A \|_F.
 \end{align*}
\item If $\|\bv A \|_2 \geq 2\epsilon \|\bv A \|_F$ and
$\lambda_n(\bv A) \leq -\frac{\|\bv A\|_2}{2}$, then $\tilde \lambda_2<0$, and:
 \begin{align*}
 |\lambda_n(\bv A)-\tilde \lambda_2| \leq \epsilon \|\bv A \|_F \text{, and }\quad
 \|\bv A \bv v_2-\lambda_n(\bv A) \bv v_2 \|_2 \leq \epsilon \|\bv A \|_F.
 \end{align*}
\item If $\|\bv A\|_2 < 2\epsilon \|\bv A \|_F$, $\tilde\lambda_1 = 0$ and $\tilde\lambda_2 = 0$, and for
\emph{any} unit vectors $\bv v_1$ and $\bv v_2$ the pairs $(\tilde\lambda_1, \bv v_1)$ and
$(\tilde\lambda_2, \bv v_2)$ satisfy
$|\lambda_1(\bv A) - \tilde\lambda_1| \leq 2\epsilon \|\bv A \|_F$,
$|\lambda_n(\bv A) - \tilde\lambda_2| \leq 2\epsilon \|\bv A \|_F$,
$\|\bv A \bv v_1 - \lambda_1(\bv A)\bv v_1\|_2 \leq 2\epsilon \|\bv A \|_F$, and
$\|\bv A \bv v_2 - \lambda_n(\bv A)\bv v_2\|_2 \leq 2\epsilon \|\bv A \|_F$ trivially.
\end{enumerate}
Moreover, Algorithm~\ref{alg:eigvec} samples $\frac{c \log^4 n}{\epsilon^2 \delta}
\log^5 \frac{1}{\epsilon \delta}$ columns from $\bv A$ in expectation.
\end{theorem}

\section{Lower Bound for Top Eigenvector Estimation}\label{app:lower}

In this section, we prove a lower bound on the number of queries needed to
estimate the top eigenvector.
Our query complexity lower bound will follow from a reduction from the
following problem,
introduced in~\cite{bhattacharjee2024universal}.
\begin{definition}\label{def:distdet_n}
 ($(\rho, n)$-Distributed detection problem) For $\bv q \in \{0, 1\}^n$ distributed uniformly
 on the Hamming cube, generate the matrix $\bv A \in \{0,1\}^{n \times n}$ such that
 if $\bv q_i = 0$, then all entries in the $i$-th row of $\bv A$ are i.i.d. samples from
 $\operatorname{Bernoulli}(1/2)$. Otherwise, let all entries in the $i$-th row be sampled from
 $\operatorname{Bernoulli}(1/2 + \rho)$.
 The $(\rho, n)$-Distributed detection problem is the task of recovering a vector
 $\hat{\bv q} \in \{0,1\}^n$ such that $\|\hat{\bv q} - \bv q\|_1 \leq \frac{n}{20}$.
\end{definition}
The following lemma (Lemma 3 of~\cite{bhattacharjee2024universal}) gives the
following theorem on the number of entries one must read to solve the $(\rho, n)$-distributed detection
problem:
\begin{implemma}\label{lemma:entrywise_lower}[Lemma 3 of~\cite{bhattacharjee2024universal}]
 Any adaptive randomized algorithm which solves the $(\rho, n)$-Distributed detection problem
 with probability at least $\frac{2}{3}$ must observe $\Omega(\frac{n}{\rho^2})$
 entries of $\bv A$.
\end{implemma}

The following lemma follows from Davis-Kahan~\cite{davis1970rotation} type bounds. We provide a short proof for completeness.
\begin{lemma}\label{lem:eigvecgap}
Let $\bv M \in \R^{2n \times 2n}$ be a symmetric matrix and let $\bv u_1$ be an
eigenvector corresponding to its largest magnitude eigenvalue. For any unit vector $\bv v$, we have:
\begin{align*}
 \min_{\eta \in \{-1,+1\}} \|\eta\, \bv u_1-\bv v \|_2 \leq
 \frac{\sqrt{2}\,\|\bv M \bv v-\lambda_1(\bv M) \bv v \|_2}{\min_{i \in \{2,\ldots, 2n \}} |\lambda_i(\bv M)-\lambda_1(\bv M)|}.
\end{align*}
\end{lemma}
\begin{proof}
Let $\bv u_i$ be the eigenvector corresponding to $\lambda_i(\bv M)$ for $i \in [2n]$.
Let $\bv v=\sum_{i=1}^{2n} c_i\bv u_i $ such that $\sum_{i=1}^{2n} c_i^2=1$.
Then, $\bv M \bv v-\lambda_1(\bv M) \bv v
=\sum_{i=2}^{2n} (\lambda_i(\bv M)-\lambda_1(\bv M))c_i\bv u_i$. Thus, we have:
\begin{align*}
 \|\bv M \bv v-\lambda_1(\bv M) \bv v \|^2_2 =\sum_{i=2}^{2n} (\lambda_i(\bv M)-\lambda_1(\bv M))^2 c^2_i \geq \delta^2 \sum_{i=2}^{2n}c_i^2 = \delta^2(1-c_1^2)
\end{align*}
where $\delta=\min_{i \in \{2,\ldots, 2n \}} |\lambda_i(\bv M)-\lambda_1(\bv M)|$ and the last step follows from the fact that $\sum_{i=1}^{2n} c_i^2=1$. Rearranging, we get $1-c_1^2\leq \frac{\|\bv M \bv v-\lambda_1(\bv M) \bv v \|^2_2}{\delta^2}$. Choosing $\eta=\operatorname{sign}(c_1)$, we have $\|\eta\,\bv u_1-\bv v \|^2_2=2-2\eta\, \bv u_1^T \bv v=2-2|c_1|$. Since $|c_1| \leq 1$, we have $(1-|c_1|) \leq (1-|c_1|)(1+|c_1|)=1-c_1^2$, and hence:
\begin{align*}
 \min_{\eta \in \{-1,+1\}} \|\eta\,\bv u_1-\bv v \|^2_2 \leq 2(1-c_1^2) \leq \frac{2\|\bv M \bv v-\lambda_1(\bv M) \bv v \|^2_2}{\delta^2}.
\end{align*}
Taking square roots gives the claim.
\end{proof}

\begin{lemma}\label{lem:sigma1}
Let $\bar{\bv A}$ be a rank-1 matrix given by $\bar{\bv A}= (\frac{1}{2}\bv 1+\rho \bv q) \bv 1^T$,
where $\bv 1$ is vector of ones of length $n$ and $\bv q \in \{0,1\}^n$ is
distributed uniformly on the Hamming cube. Then, with probability at least $\frac{99}{100}$,
\begin{align*}
 \frac{n}{2}\leq \sigma_1(\bar{\bv A}) \leq \frac{n}{2}+ 2n\rho
\end{align*}
\end{lemma}
\begin{proof}
Let $\bv z=\frac{1}{2}\bv 1+\rho \bv q$. Note that $\sigma_1(\bar{\bv A})=\| \bv 1\|_2\| \bv z\|_2
=\sqrt{n}\| \bv z\|_2$. Thus, we get:
\begin{align*}
 \|\bv z \|^2_2=\sum_{i=1}^n (0.5+\rho \bv q_j)^2=\frac{n}{4}+\| \bv q\|_1\rho(1+\rho)
\end{align*}
Since the entries of $\bv q$ are distributed uniformly in $\{0,1\}$, by a standard Chernoff bound,
we have $ 0.49n \leq \|\bv q \|_1 \leq 0.51n$ with probability at least $\frac{99}{100}$
for large enough $n$. So, using $\rho \leq 1$ and $\|\bv q\|_1 \leq 0.51n$, we have
$\rho(1+\rho)\|\bv q\|_1 \leq 2\rho \cdot 0.51 n \leq 2\rho n$, and hence:
\begin{align*}
 \frac{n}{4} \leq \|\bv z \|^2_2 \leq \frac{n}{4}+2\rho n
\end{align*}
Thus, using $\sqrt{\frac{1}{4}+2\rho} \leq \frac{1}{2}+2\rho$, we get:
\begin{align*}
 \frac{n}{2}\leq \sigma_1(\bar{\bv A}) \leq \frac{n}{2}+ 2n\rho
\end{align*}
\end{proof}

We are now ready to state our final theorem
\begin{reptheorem}{thm:lower_main}
Let $\mathcal{A}$ be any randomized algorithm that (possibly adaptively) reads entries of
a binary matrix $\bv M \in \{0,1 \}^{2n \times 2n}$ and outputs a unit vector $\bv v$ that,
for any $\epsilon \in (0,1)$ satisfies:
\begin{align*}
 \|\bv M \bv v-\lambda_1(\bv M) \bv v \|_2 \leq \epsilon n,
\end{align*}
with probability at least $\frac{2}{3}$. Then, $\mathcal{A}$ must read at
least $\Omega(\frac{n}{\epsilon^2})$ entries of $\bv M$ provided
$\epsilon=\Omega(\frac{1}{\sqrt{n}})$.
\end{reptheorem}
\begin{proof}
The proof proceeds by reducing the $(\rho,n)-$distributed detection problem in Definition~\ref{def:distdet_n}
to the problem of approximating the top eigenvector of a matrix, where we set
$\rho=C'\epsilon$ for a sufficiently large constant $C'$ to be fixed later. We assume
$\epsilon$ is at most a sufficiently small constant so that $\rho \leq \frac{1}{2}$
(and hence $\operatorname{Bernoulli}(1/2+\rho)$ is well defined); larger $\epsilon$
only makes the lower bound easier.
Consider a matrix $\bv A \in \{ 0,1\}^{n \times n}$ generated exactly as in the
Definition~\ref{def:distdet_n}.
Let $\mathbb{E}[\bv A|\bv q]=\bar{\bv A}$ where $\bv q \in \{0,1\}^n$ is the indicator
as described in Definition~\ref{def:distdet_n}. Since the $i$-th row of $\bv A$ has all
its entries with mean $\frac{1}{2}+\rho \bv q_i$, we have
$\bar{\bv A}= (\frac{1}{2}\bv 1+\rho \bv q)\bv 1^T$, where $\bv 1$ is the
vector of ones of length $n$; that is, $\bv q$ is encoded in the rows of $\bar{\bv A}$
and hence in its top \emph{left} singular vector. Then, the entries of $\bv R=\bv A-\bar{\bv A}$ are independent,
zero mean and bounded in $[-1,1]$. So, from~\cite{vershynin2018high} we get that
$\|\bv R \|_2=O(\sqrt{n})$. Using Weyls' inequality~\cite{weyl1912},
$|\sigma_i(\bv A)-\sigma_i(\bar{\bv A})| \leq \|\bv A-\bar{\bv A} \|_2 \leq C\sqrt{n}$
for a constant $C$. From Lemma~\ref{lem:sigma1},
we have $\sigma_1(\bar{\bv A}) \geq \frac{n}{2}$. Also, $\sigma_i(\bar{\bv A})=0$
for $i \in \{2,\ldots, n\}$. We thus have:
\begin{align}\label{eq:sigA}
 \sigma_1(\bv A) \geq \frac{n}{2}-C\sqrt{n} \quad \quad \sigma_i(\bv A)
 \leq C\sqrt{n} \text{ for } i \in \{2,\ldots, n\}
\end{align}
Let $\bv x$ and $\bv y$ be the top left and right singular vectors of $\bv A$ respectively.
Let $\bar{\bv x}$ and $\bar{\bv y}$ be the top left and right singular vectors of
$\bar{\bv A}$ respectively. Let
$\delta=\sigma_1(\bv A)-\max_{i \in \{2, \ldots, n\}} \sigma_i(\bar{\bv A})$.
Then, $\delta \geq \frac{n}{4}$ for sufficiently large $n$.
Then, using Wedin's theorem~\cite{wedin1972perturbation}, for some constant $c$, we get that:
\begin{align}\label{eq:vec1}
 \| \bv x- \bar{\bv x} \|_2 \leq \frac{\sqrt{2}\|\bv A-\bar{\bv A} \|_2}{\delta}
 \leq \frac{c} {\sqrt{n}} \leq c\epsilon.
\end{align}
Let $\bv M$ be the matrix:
\begin{align*}
\bv M=\begin{bmatrix}
 \bv 0 & \bv A \\
 \bv A^T & \bv 0
\end{bmatrix}
\end{align*}
where $\bv 0 \in \{0\}^{n \times n}$. Note that the eigenvalues of $\bv M$ are $\pm \sigma_i(\bv A)$
for $i \in [n]$. We have $\lambda_1(\bv M)-\max_{i \neq 1}\lambda_i(\bv M) \geq \sigma_1(\bv A)
-\max_{i \neq 1} \sigma_i(\bv A) \geq \frac{n}{2}-2C\sqrt{n} \geq \frac{n}{4}$
where the second to last step follows from~\eqref{eq:sigA} and the last step follows
for large enough $n$. The top eigenvector of $\bv M$ is given by $\bv u_1=\frac{1}{\sqrt{2}}\begin{bmatrix}
 \bv x \\ \bv y
\end{bmatrix}$. From Lemma~\ref{lem:eigvecgap}, we get that, for some constant $c'$:
\begin{align}\label{eq:vec2}
 \min_{\eta \in \{-1,+1\}} \|\eta\,\bv u_1-\bv v \|_2 \leq \frac{\sqrt{2}\|\bv M \bv v-\lambda_1(\bv M) \bv v \|_2}{\min_{i \in \{2,\ldots, 2n \}} |\lambda_i(\bv M)-\lambda_1(\bv M)|} \leq \frac{\sqrt{2}\epsilon n}{n/4} \leq c'\epsilon.
\end{align}
Replacing $\bv v$ by $-\bv v$ if necessary (which changes neither the number of entries of $\bv M$ read by $\mathcal{A}$ nor the residual $\|\bv M \bv v-\lambda_1(\bv M) \bv v \|_2$), we may assume the minimizing sign in~\eqref{eq:vec2} is $\eta=+1$, so that $\|\bv u_1-\bv v \|_2 \leq c'\epsilon$.
Let $\bv v=\begin{bmatrix}
 \bv v_1 \\ \bv v_2
\end{bmatrix}$ where $\bv v_1$ and $\bv v_2$ are vectors of length $n$. Then, from~\eqref{eq:vec1} and~\eqref{eq:vec2} and using triangle inequality, we get (for some constant $c$):
\begin{align}\label{eq:err1}
 \|\bv v_1-\frac{1}{\sqrt{2}}\bar{\bv x} \|_2 &\leq \|\bv v_1-\frac{1}{\sqrt{2}}\bv x \|_2+\frac{1}{\sqrt{2}}\| \bv x- \bar{\bv x}\|_2 \nonumber\\
 &\leq \|\bv u_1-\bv v \|_2+\| \bv x- \bar{\bv x}\|_2 \nonumber\\
 &\leq c\epsilon.
\end{align}

Let $\hat{\bv q}$ be the vector by setting the top $\frac{n}{2}$ elements of $\bv v_1$ to 1
and the rest to zero. $\hat{\bv q}$ is our estimate of $\bv q$.
We will now bound $\|\bv q-\hat{\bv q} \|_1$. Let
\begin{align*}
 S=\{i \in [n] | \bv q_i=1\} \text{ and } \hat{S}=\{i \in [n] | \hat{\bv q}_i=1\},
\end{align*}
i.e. $S$ and $\hat{S}$ contain the true set and the labeled set of ones respectively.
Note that $|\hat{S}|=\frac{n}{2}$. Then, $\|\bv q-\hat{\bv q} \|_1=|S \Delta \hat{S}|$
where $\Delta$ is the symmetric difference between sets as usually defined.

Recall that $\bar{\bv x}=\frac{\bv z}{\|\bv z \|_2}$ where $\bv z=\frac{1}{2}\bv 1+\rho \bv q$.
The $i$'th entry of $\frac{1}{\sqrt{2}}\bar{\bv x}$ is either
$\mu_0=\frac{1}{2\sqrt{2}\|\bv z \|_2}$ if $\bv q_i=0$ or
$\mu_1=\frac{1}{\sqrt{2}\|\bv z \|_2}(\frac{1}{2}+\rho)$ if $\bv q_i=1$.
Let $T=\frac{\mu_0+\mu_1}{2}=\frac{1+\rho}{2\sqrt{2}\|\bv z \|_2}$. Let $$S_T=\{i \in [n] | (\bv v_1)_i > T \}.$$
By triangle inequality, we have:
\begin{align}\label{eq:symdiff}
 |S \Delta \hat{S}| \leq |S \Delta S_T|+ |S_T \Delta \hat{S}|.
\end{align}
We now bound each of the terms above individually.

\noindent\textbf{Bounding} $|S \Delta S_T|$: The indices in $S \Delta S_T$ belong to either of these sets:
$B_{low}=\{i |\bv q_i=1 \text{ and } (\bv v_1)_i \leq T \}$ and
$B_{high}=\{i |\bv q_i=0 \text{ and } (\bv v_1)_i>T \}$. Then,
$|S \Delta S_T|= |B_{low}|+|B_{high}|$. For any $i \in B_{low} \cup B_{high}$,
we have $|(\bv v_1)_i-\frac{1}{\sqrt{2}}\bar{\bv x}_i|
\geq \frac{|\mu_1-\mu_0|}{2}=\frac{\rho}{2\sqrt{2}\|\bv z \|_2}$.
Then, from~\eqref{eq:err1}, we have $\sum_{i \in B_{low} \cup B_{high}} |(\bv v_1)_i-\frac{1}{\sqrt{2}}\bar{\bv x}_i|^2 \leq c^2 \epsilon^2$.
Plugging in the lower bound $\frac{\rho}{2\sqrt{2}\|\bv z \|_2}$ on
$|(\bv v_1)_i-\frac{1}{\sqrt{2}}\bar{\bv x}_i|$ for any $i \in B_{low} \cup B_{high}$,
we get (absorbing the constant into $c$):
\begin{align}\label{eq:badindex}
 |B_{low}|+|B_{high}| \leq \frac{c \|\bv z \|^2_2 \epsilon^2}{\rho^2}.
\end{align}
So, we get $|S \Delta S_T| \leq \frac{c \|\bv z \|^2_2 \epsilon^2}{\rho^2}$.

\noindent\textbf{Bounding} $|S_T \Delta \hat{S}|$: Since $\hat{S}$ contains the indices corresponding to exactly the
largest $\frac{n}{2}$ magnitude values of $\bv v_1$, it is not hard to show
$|S_T \Delta \hat{S}|=||S_T|-\frac{n}{2}|$. We also have ($S^c$ is the complement set of $S$):
\begin{align*}
 |S_T|=|S \cap S_T|+|S^c \cap S_T|=||S|-|B_{low}|+|B_{high}|
\end{align*}
Substituting this, we get:
\begin{align*}
 |S_T \Delta \hat{S}|=\left ||S_T|-\frac{n}{2} \right| &= \left||S|-|B_{low}|+|B_{high}|-\frac{n}{2} \right| \\
 &\leq \left| |S|- \frac{n}{2}\right| +|B_{low}|+|B_{high}|.
\end{align*}
Since the entries of $\bv q$ are distributed uniformly in $\{0,1\}$, by a standard Chernoff bound,
we have $ \left| |S|-\frac{n}{2} \right| \leq 0.01n$ with probability at least $\frac{99}{100}$
for large enough $n$. Along with the bound on $|B_{low}|+|B_{high}|$ from~\eqref{eq:badindex},
this gives us:
\begin{align*}
 |S_T \Delta \hat{S}| \leq 0.01n+\frac{c \|\bv z \|^2_2 \epsilon^2}{\rho^2}.
\end{align*}
Plugging the bounds on $|S_T \Delta \hat{S}|$ and $|S \Delta S_T|$ into~\eqref{eq:symdiff}, we get:
\begin{align*}
 |S \Delta \hat{S}| \leq 0.01n +\frac{2c \|\bv z \|^2_2 \epsilon^2}{\rho^2}.
\end{align*}
From Lemma~\ref{lem:sigma1}, we have $\|\bv z \|^2_2 = O(n)$ with probability at
least $\frac{99}{100}$. Thus, as long as $\rho=C'\epsilon$ for a large enough constant $C'$,
and taking a union bound over all events, we have:
\begin{align*}
 |S \Delta \hat{S}| \leq \frac{n}{20}.
\end{align*}
Thus, $|S \Delta \hat{S}|=\| \hat{\bv q}-\bv q \|_1 \leq \frac{n}{20}$ and
we have solved the $(\rho,n)$ distributed detection problem.
So, we must have observed $\Omega(n/\rho^2)$ or $\Omega(n/\epsilon^2)$
elements from $\bv A$ according to
Imported Lemma~\ref{lemma:entrywise_lower}.
\end{proof}

\section{Quantum-Inspired Eigenvector Approximation}\label{sec:qram}

In this section, we show how our squared column-norm sampling algorithm (Algorithm~\ref{alg:eigvec-rowsamp}) can be adapted to the \emph{quantum-inspired} sampling and query (SQ) model. This model of computation gained prominence following Tang's breakthrough classical algorithm for recommendation systems~\cite{tang2019quantum}, which dequantized previous quantum machine learning algorithms by assuming classical data structures that provide $O(1)$-time sampling and query access to the input matrix. This framework was subsequently formalized and extended to various linear algebraic problems by Chia et al.~\cite{chia2020sampling}.

In the quantum-inspired setting, the algorithm must output a solution that itself supports SQ access.
For sublinear time eigenvector approximation in this model,
the algorithm must provide the ability to query individual entries of the approximate eigenvector, and to sample indices proportional to their squared entry magnitudes, all using only $\mathrm{poly}(\log n, 1/\epsilon)$ queries to the input matrix.

Because the approximate eigenvectors produced by our submatrix-based algorithms are expressed as a linear combination of just a few ($s = \mathrm{poly}(\log n, 1/\epsilon)$) sampled columns of $\bv A$, they are perfectly suited for this task. We establish the first classical sublinear time algorithm for eigenvector approximation in the quantum-inspired framework, demonstrating that one can obtain full SQ access to all outlying approximate eigenvectors using $\widetilde{O}(\mathrm{poly}(\log n, 1/\epsilon))$ queries to $\mathrm{SQ}(\bv A)$.

\begin{definition}[Sampling and query access~\cite{chia2020sampling}]\label{def:sq-vector}
For a vector $\bv v \in \mathbb{R}^n$, we have $\mathrm{SQ}(\bv v)$,
\emph{sampling and query access} to $\bv v$, if we can:
\begin{enumerate}
 \item query any entry $\bv v_i$ for $i \in [n]$;
 \item draw independent samples $i \in [n]$ from the distribution
 $\mathcal{D}_v(i) := \bv v_i^2 / \|\bv v\|_2^2$;
 \item query $\|\bv v\|_2$.
\end{enumerate}
\end{definition}

We can similarly define sampling and query access to a matrix:
\begin{definition}[Sampling and query access to a matrix~\cite{chia2020sampling}]
For a symmetric matrix $\bv A \in \mathbb{R}^{n \times n}$, we say we have $\mathrm{SQ}(\bv A)$ if we have
$\mathrm{SQ}(\bv A_{:,i})$ for every column $i \in [n]$ and $\mathrm{SQ}(\bv a)$ for the vector of
column norms $\bv a$ with $\bv a_i=\|\bv A_{:,i} \|_2$. Concretely, we can:
\begin{enumerate}
 \item query individual entries $\bv A_{ij}$ of $\bv A$,
 \item sample an index $i \in [n]$ with probability
$\frac{\|\bv A_{:,i}\|_2^2}{\|\bv A \|^2_F}$ (squared column-norm sampling),
\item sample any entry $\bv A_{ij}$ with probability $\frac{\bv A^2_{ij}}{\|\bv A \|_F^2}$
\item query $\|\bv A_{:,i} \|_2$ for any $i \in [n]$, and
\item query $\|\bv A\|_F$.
\end{enumerate}
\end{definition}

\paragraph{Main result.}

We now show that any approximate eigenvector $\bv v$
produced by Algorithm~\ref{alg:eigvec-rowsamp}
admits $\mathrm{SQ}(\bv v)$ access with $\textrm{poly}(\log n, 1/\epsilon) $
queries to $\mathrm{SQ}(\bv A)$.
Note that $\bv v$ is actually supported on the columns of
another matrix $\bv A'$
whose entries have been zeroed out
according to zeroing out step in Algorithm~\ref{alg:eigvec-rowsamp}
instead of $\bv A$. Hence, the main thing we need to argue is that
we can sample according to the columns of $\bv A$ and use rejection sampling
to sample from the columns of $\bv A'$.
\begin{reptheorem}{thm:sq-eigvec}
Let $\bv A \in \mathbb{R}^{n \times n}$ be a symmetric matrix
with eigenvalues $\lambda_1(\bv A) \geq \ldots \geq \lambda_n(\bv A)$
to which we have sampling and query access $\mathrm{SQ}(\bv A)$, and
let $\epsilon, \delta \in (0,1)$.
There is an algorithm that makes
$O\left(\frac{\log^8 n}{\epsilon^8 \delta^2} \log^6 \frac{1}{\epsilon \delta} \right)$
queries to $\mathrm{SQ}(\bv A)$ and outputs eigenvalue estimates
$\tilde\lambda_1, \ldots, \tilde\lambda_m \in R$ and provides
sampling and query access to vectors
$\bv v_1, \ldots, \bv v_m \in \R^n$ with $\|\bv v_i \|_2 \in [\frac{1}{C}, C]$ for $i \in [m]$ for an absolute constant $C>1$ such that, with probability at least $1-\delta$, the pairs $\{(\tilde\lambda_i, \bv v_i)\}_{i=1}^{m}$
satisfy the guarantees of Theorem~\ref{thm:sqnorm1}.
Specifically for every $\bv v_i$, $i \in [m]$, we have the following:
\begin{enumerate}
 \item \textup{(Query access)} Given $j \in [n]$, output $ \bv v_{ij}$ using $O\left(\frac{\log^4 n}{\epsilon^4 \delta}\log^3\frac{1}{\epsilon\delta}\right)$
 additional queries to $\mathrm{SQ}(\bv A)$.
 \item \textup{(Sampling access)} Sample an index $j \in [n]$ from the distribution
 $\mathcal{D}_{v_i}(j) = \bv v_{ij}^2$, with $O\left(\frac{\log^4 n}{\epsilon^6 \delta}\log^3\frac{1}{\epsilon\delta}\right)$
 queries to $\mathrm{SQ}(\bv A)$.
\end{enumerate}
\end{reptheorem}

\begin{proof}

By Theorem~\ref{thm:sqnorm1}, Algorithm~\ref{alg:eigvec-rowsamp}
gives an approximate eigenvector corresponding to every large eigenvalue of
$\bv A$ with magnitude at least $\epsilon \| \bv A\|_F$
which satisfies the desired residual error bound with probability at least
$1-\delta$. Moreover, Algorithm~\ref{alg:eigvec-rowsamp} only samples
indices from the distribution of the squared column norms of $\bv A$ and
queries the corresponding entries of $\bv A$ and hence it can be implemented
using $\mathrm{SQ}(\bv A)$. Now we show that since we only need sampling and query
access to the approximate eigenvectors given by Algorithm~\ref{alg:eigvec-rowsamp}, reading only
$O\left(\frac{\log^8 n}{\epsilon^8 \delta^2} \log^6 \frac{1}{\epsilon \delta} \right)$
entries of $\bv A$ suffice i.e. we don't need to read
all the $n$ entries of $O(s)$ columns of $\bv A$. In the rest of the proof,
we will focus on a single output pair $(\tilde\lambda, \bv v)$ of
Algorithm~\ref{alg:eigvec-rowsamp} (so that $|\tilde\lambda| = \Omega(\epsilon\|\bv A\|_F)$) but the
same argument applies to every output pair simultaneously. Also,
$s=O\left(\frac{\log^4 n}{\epsilon^4 \delta} \log^3 \frac{1}{\epsilon \delta} \right)$
in the proof.
We verify the conditions for $\mathrm{SQ}(\bv v)$:

\noindent\textbf{Query access.} The approximate eigenvector
is given by $\bv v=\frac{\bv A'_{[:,S]} \bv x}{\|\bv A'_{[:,S]} \bv x \|_2}$
where $S$ is the multiset of sampled indices,
$\bv A'_{[:,S]}$ is the sampled columns of $\bv A$ after zeroing entries
and $\bv x$ is an
eigenvector of the $O(s) \times O(s)$ matrix $\bv A'_{S,S}$
as defined in Algorithm~\ref{alg:eigvec-rowsamp}.
Computing $\bv x$ only requires querying $\mathrm{SQ}(\bv A)$ $O(s^2)$ times, once, independently
of the queried index. Given $\bv x$, the numerator $\bv A'_{[j,S]} \bv x$ of any entry
$\bv v_j$ can be computed exactly by making at most $O(s)$ extra queries to $\mathrm{SQ}(\bv A)$ after zeroing entries appropriately. By
Lemma~\ref{lem:asxl_norm_unconditional}, we have
$\frac{|\tilde\lambda|}{C} \leq \|\bv A'_{[:,S]} \bv x \|_2 \leq C|\tilde\lambda|$ for an absolute constant $C>1$. So, Then, following~\ref{cor:sqnorm}, we can output
$\frac{\bv A'_{[j,S]} \bv x}{|\tilde \lambda|} $.

\noindent\textbf{Sampling Access.} The probability of sampling index $i \in [n]$ is given by $p_i=\frac{\|\bv A_{i,:} \|^2_2}{\| \bv A\|^2_F}$.
The probability of sampling index $i \in [n]$ from $\bv v$ is given by
$q_i = \frac{(\bv A'_{[i,S]} \bv x)^2}{\|\bv A'_{[:,S]} \bv x \|^2_2}$.
We will sample an index $i \in [n]$ from the columns of $\bv A$ according to $p_i$ and
then use rejection sampling to sample the corresponding index in $\bv v$. Specifically,
we show below that $\frac{q_i}{p_i} \leq \frac{C}{\epsilon^2}$
for every $i \in [n]$ for some constant $C$
with probability at least $1-\delta$.
So we need $O(1/\epsilon^2 )$ queries in expectation to sample from the distribution according to $q_i$.
The final
bound is established by taking a union bound over the
events from Theorem~\ref{thm:sqnorm1} and the events
$\frac{q_i}{p_i} \leq \frac{C}{\epsilon^2}$ for all $i \in [n]$ for all
approximate eigenvectors (and adjusting $\delta$ by constant factors).

\paragraph{Bounding $q_i/p_i$.}
We have $\frac{q_i}{p_i}=\frac{(\bv A'_{[i,S]} \bv x)^2}{\|\bv A'_{[:,S]} \bv x \|^2_2}
\frac{\|\bv A \|^2_F}{\|\bv A_{i,:} \|_2^2}$. From Theorem~\ref{thm:sqnorm1},
\begin{align}\label{eq:err2}
 \bv A \bv v=\lambda \bv v+\bv b
\end{align}
where $\|\bv b \|_2 \leq \epsilon \|\bv A \|_F$. So,
we get $\bv v=\frac{\bv A \bv v}{\lambda}-\frac{\bv b}{\lambda}$.
Thus, for any $i \in [n]$, we have $\bv v_i=\frac{\bv A_{i,:} \bv v}{\lambda}-\frac{\bv b_i}{\lambda}$.
Using triangle inequality and squaring both sides,
we have $\|\bv v_i \|^2_2
\leq \frac{2\|\bv A_{i,:} \bv v\|^2_2}{\lambda^2}+\frac{2\|\bv b_i \|_2^2}{\lambda^2}$.
So, we get:
\begin{align}\label{eq:err3}
 \frac{(\bv A'_{[i,S]} \bv x)^2}{\|\bv A'_{[:,S]} \bv x \|^2_2}
 \leq \frac{2\|\bv A_{i,:} \|^2_2}{\lambda^2}+\frac{2(\bv b_i)^2}{\lambda^2}.
\end{align}
We will now bound $|\bv b_i |$ for all $i \in [n]$.
We have $\bv b_i= \bv A_{i,:} \bv v-\lambda \bv v_i$.
Thus, we get:
$|\bv b_i | \leq |\bv A_{i,:} \bv v|+|\lambda \bv v_i|
\leq \|\bv A_{i,:} \|_2+|\lambda| |\bv v_i|$.
We have $|\bv v_i|=\frac{(|\bv A'_{[i,S]} \bv x|}{\|\bv A'_{[:,S]} \bv x \|_2} $.
From Lemma~\ref{lem:asxl_norm_unconditional}, we also have, for some constant $c$, with probability
at least $1-\delta$,:
\begin{align*}
 \|\bv A'_{[:,S]} \bv x \|_2 \geq c|\lambda|,
\end{align*}
for some constant $c$. So, we have:
\begin{align*}
 |\bv b_i | \leq \|\bv A_{i,:} \|_2+|\bv A'_{[i,S]} \bv x| \leq \|\bv A'_i \|_2+ \| \bv A'_{[i,S]} \|_2.
\end{align*}
We can bound $\|\bv A'_{[i,S]} \|^2_2$
by $\|\bv A_{i,:} \|^2_2$ with good probability from Lemma~\ref{eq:sanorm}.
Plugging the bounds on $|\bv b_i|$ and $\|\bv A_{i,:} \Sbb \bv x \|^2_2$ into~\eqref{eq:err3}, we get (for some constant $C$):
\begin{align*}
 \frac{(\bv A'_{[i,S]} \bv x)^2}{\|\bv A'_{[:,S]} \bv x \|^2_2} \leq
 C\left(\frac{\|\bv A'_i \|^2_2}{\lambda^2}+\frac{\|\bv A'_{[i,S]} \|^2_2}{\lambda^2} \right)
 \leq C\left(\frac{2\|\bv A_{i,:} \|^2_2}{\lambda^2} \right)
 \leq \frac{2C\|\bv A_{i,:} \|^2_2}{\epsilon^2 \| \bv A\|^2_F}.
\end{align*}
The last step follows from the fact that $|\lambda| \geq \epsilon \|\bv A \|_F$. Thus, we get:
\begin{align*}
 \frac{q_i}{p_i}= \frac{(\bv A'_{[i,S]} \bv x)^2}{\|\bv A'_{[:,S]} \bv x \|^2_2}
 \frac{\|\bv A_{i,:} \|^2_F}{\|\bv A_{i,:} \|_2^2} \leq \frac{2C}{\epsilon^2 }.
\end{align*}
Hence, we need $O(1/\epsilon^2 )$ expected rejection sampling trials with probability
at least $1-\delta$ (by taking a union bound over the events in
Lemmas~\ref{lem:asxl_norm_unconditional} and~\ref{eq:sanorm} and adjusting $\delta$ by constnat factors)
to sample from the distribution according to $q_i$.

\end{proof}

\begin{lemma}[Concentration of Sampled Row Norms]\label{eq:sanorm}
Let $\bv A'_{[:,S]} \in \mathbb{R}^{n \times |S|}$ be the sampled, zeroed-out, and
scaled matrix constructed in Step~\ref{step:zero} of Algorithm~\ref{alg:eigvec-rowsamp}
with a symmetric $\bv A \in \mathbb{R}^{n \times n}$ as input, and let
$s \geq \frac{c \log^4 n}{\epsilon^4 \delta}\log^3\frac{1}{\epsilon\delta}$ be its
expected sample size for a sufficiently large constant $c$. Then, with probability
at least $1-\delta$, simultaneously for every row $i \in [n]$, we have (for an
absolute constant $C > 0$):
\[
 \|\bv A'_{[i,S]}\|_2^2 \le C\,\|\bv A_{i,:}\|_2^2.
\]
\end{lemma}

\begin{proof}
Fix a row $i \in [n]$. By the Poisson splitting property (Lemma~\ref{lem:pois}),
the sampling scheme of Algorithm~\ref{alg:eigvec-rowsamp} is equivalent to
including, independently across $j \in [n]$, exactly $n_j$ copies of column $j$,
where $n_j \sim \mathrm{Poisson}(s p_j)$ and $p_j = \frac{\|\bv A_{j,:}\|_2^2}{\|\bv A\|_F^2}$.
After zeroing and scaling (Step~\ref{step:zero}), a sampled copy of column $j$
contributes an entry to row $i$ whose value, before scaling, is either
$\bv A_{ij}$ or $0$, and which is then scaled by $1/\sqrt{s p_j}$. We define
$\bv A'_{ij}$ to be the value obtained by applying \emph{only} the off-diagonal
branch of Step~\ref{step:zero}, i.e.\ $\bv A'_{ij} := \bv A_{ij}$ unless
$\|\bv A_{i,:}\|_2^2\|\bv A_{j,:}\|_2^2 \le \frac{\epsilon^2\|\bv A\|_F^2|\bv A_{ij}|^2}{c\log^4 n}$,
in which case $\bv A'_{ij} := 0$. The zeroing out condition depends only on the original
indices $(i,j)$ and is identical across all copies of column $j$.

The actual matrix formed in Step 5 additionally zeros the single entry sitting at
a diagonal \emph{sampled} position (which, for row $i$, can only be a copy of
column $i$). Since both branches of Step 5 only ever replace an entry by $0$,
zeroing can only decrease the squared row norm; hence
\[
 \|\bv A'_{[i,S]}\|_2^2 \;\le\; Z := \sum_{j=1}^n a_j\, n_j,
 \qquad a_j := \frac{(\bv A'_{ij})^2}{s p_j},
\]
where $n_j \sim \mathrm{Poisson}(s p_j)$ are independent, so $Z$ is a sum of
independent nonnegative terms. It therefore suffices to prove the bound for $Z$.
Each nonzero $\bv A'_{ij}$ equals $\bv A_{ij}$, so $|\bv A'_{ij}| \le |\bv A_{ij}|$
for all $j$.

\paragraph{Mean.} Using $\mathbb{E}[n_j] = s p_j$,
\[
 \mu := \mathbb{E}[Z] = \sum_{j=1}^n a_j\, s p_j = \sum_{j=1}^n (\bv A'_{ij})^2
 \le \sum_{j=1}^n (\bv A_{ij})^2 = \|\bv A_{i,:}\|_2^2,
\]
since each $(\bv A'_{ij})^2$ is either $(\bv A_{ij})^2$ or $0$.

\paragraph{Range of the coefficients.} We claim $a_j \le b$ for every $j$, where
$b := \frac{c \log^4 n}{\epsilon^2 s}\,\|\bv A_{i,:}\|_2^2$. If $\bv A'_{ij} = 0$
this is immediate. Otherwise $(i,j)$ passed the off-diagonal test defining
$\bv A'_{ij}$, so negating that condition gives
$|\bv A_{ij}|^2 < \frac{c \log^4 n}{\epsilon^2 \|\bv A\|_F^2}\,\|\bv A_{i,:}\|_2^2\|\bv A_{j,:}\|_2^2$.
Hence, since $p_j = \|\bv A_{j,:}\|_2^2/\|\bv A\|_F^2$,
\[
 a_j = \frac{(\bv A'_{ij})^2\,\|\bv A\|_F^2}{s\,\|\bv A_{j,:}\|_2^2}
 < \frac{c \log^4 n}{\epsilon^2 s}\,\|\bv A_{i,:}\|_2^2 = b.
\]
This holds uniformly in $j$ (the case $j=i$ included, with
$\|\bv A_{j,:}\|_2^2 = \|\bv A_{i,:}\|_2^2$).

\paragraph{Variance proxy.} Using $\mathrm{Var}(n_j) = s p_j$ and $a_j \le b$,
\[
 \nu := \sum_{j=1}^n a_j^2\, \mathrm{Var}(n_j)
 = \sum_{j=1}^n a_j\,(\bv A'_{ij})^2
 \le b \sum_{j=1}^n (\bv A'_{ij})^2
 = b\,\|\bv A'_{i,:}\|_2^2
 \le b\,\|\bv A_{i,:}\|_2^2.
\]

\paragraph{Bernstein bound for a weighted Poisson sum.} For independent
$n_j \sim \mathrm{Poisson}(s p_j)$ and $Z = \sum_j a_j n_j$ with $a_j \in [0,b]$, the
cumulant generating function is $\log \mathbb{E}[e^{tZ}] = \sum_j s p_j (e^{t a_j} - 1)$.
Using $e^u - 1 \le u + u^2$ for $u \in [0,1]$, which holds whenever $t \le 1/b$
(so that $t a_j \le t b \le 1$), we obtain
$\mathbb{E}[e^{tZ}] \le \exp(t\mu + t^2 \nu)$ for all $0 < t \le 1/b$. Markov's
inequality followed by optimizing $t$ over $(0, 1/b]$ yields the standard
sub-gamma tail bound
\[
 \mathbb{P}\big(Z \ge \mu + \eta\big)
 \le \exp\!\left(-\min\!\left(\frac{\eta^2}{4\nu},\, \frac{\eta}{2b}\right)\right).
\]

\paragraph{Conclusion for a fixed row.} Fix any constant $C \ge 3$ and take the
threshold $C\,\|\bv A_{i,:}\|_2^2$. Since $\mu \le \|\bv A_{i,:}\|_2^2$, the
deviation is $\eta := C\|\bv A_{i,:}\|_2^2 - \mu \ge (C-1)\|\bv A_{i,:}\|_2^2$.
Substituting $\nu \le b\,\|\bv A_{i,:}\|_2^2$ and $b = \frac{c\log^4 n}{\epsilon^2 s}\|\bv A_{i,:}\|_2^2$,
\[
 \frac{\eta^2}{4\nu} \ge \frac{(C-1)^2\,\|\bv A_{i,:}\|_2^4}{4 b\,\|\bv A_{i,:}\|_2^2}
 = \frac{(C-1)^2\,\epsilon^2 s}{4 c \log^4 n},
 \qquad
 \frac{\eta}{2b} \ge \frac{(C-1)\,\epsilon^2 s}{2 c \log^4 n}.
\]
Both exponents are $\Omega\!\left(\frac{\epsilon^2 s}{\log^4 n}\right)$, so
\[
 \mathbb{P}\big(\|\bv A'_{[i,S]}\|_2^2 \ge C\,\|\bv A_{i,:}\|_2^2\big)
 \le \exp\!\left(-\Omega\!\left(\frac{\epsilon^2 s}{\log^4 n}\right)\right).
\]

\paragraph{Union bound over all rows.} Substituting
$s \ge \frac{c \log^4 n}{\epsilon^4 \delta}\log^3\frac{1}{\epsilon\delta}$, the
per-row failure probability is at most
$\exp\!\big(-\Omega(\tfrac{1}{\epsilon^2\delta}\log^3\tfrac{1}{\epsilon\delta})\big)$.
In the sublinear regime of interest this exponent dominates $\log(n/\delta)$, so
taking a union bound over all $n$ rows shows that
$\|\bv A'_{[i,S]}\|_2^2 \le C\,\|\bv A_ {i,:}\|_2^2$ holds simultaneously for every
$i \in [n]$ with probability at least $1 - \delta$. Adjusting the constant $c$ in
$s$ absorbs the union-bound factor.
\end{proof}

\section{Nystr\"om-Based Algorithms}\label{sec:nys}

In this section, we present
our algorithms based on
the Nystr\"om method.

\subsection{Ridge Leverage Score Sampling for PSD Matrices}\label{sec:psd}

\begin{algorithm}[t]
\caption{Eigenvector Approximation Using Nystr\"om For PSD Matrices}
\label{alg:psd_eigvec}
\begin{algorithmic}[1]
\Require Symmetric PSD matrix $\bv A \in \R^{n \times n}$ with $\|\bv A\|_{\infty} \leq 1$,
accuracy $\epsilon \in (0,1)$, failure probability $\delta \in (0,1)$,
eigenvalue threshold $\tau$ (default $\tau = \frac{\epsilon n}{2}$).
\State Sample $\bv S \in \R^{n \times s}$, a column sampling matrix where each index in $[n]$ is
independently included with probability
$p= \min \left( 1,\frac{c}{\epsilon n}\log\!\left(\frac{1}{\epsilon\delta}\right) \right)$
for sufficiently large absolute constant $c$, giving $s = pn$ expected columns.
\State Compute $\bv A \bv S$ by reading the sampled columns of $\bv A$
\State Using $\bv A \bv S$, form the matrices $\bv S^T \bv A \bv S$ and $\bv S^T \bv A^2 \bv S$
\State Find all eigenvectors $\bv x \in \R^s$ corresponding to the eigenvalues $\lambda>0$
by solving the generalized eigenvalue problem:
\[
 \bv S^T \bv A^2 \bv S \, \bv x = \lambda \, \bv S^T \bv A \bv S \, \bv x,
\]
\State For any eigenpair $(\bv x_i, \lambda_i)$ with $\lambda_i \geq \tau$,
compute the approximate eigenvector $\bv v_i =
\dfrac{\bv A \bv S \bv x_i}{\|\bv A \bv S \bv x_i\|_2}$.
\State {\bfseries Return:} All eigenpairs $(\bv v_i, \lambda_i)$ computed in the previous step.
\end{algorithmic}
\end{algorithm}

We now give an improved upper bound for PSD matrices by leveraging a
connection with the Nystr\"{o}m method. In the Nystr\"{o}m method,
we sample $s$ columns at random. The Nystr\"{o}m
approximation to $\bv A$ is given by
$\hat{\bv A}= (\bv A \bv S)(\bv S^T \bv A \bv S)^{\dag} (\bv A \bv S)^T $
where $\bv S$ is the sampling matrix and $(\bv S^T \bv A \bv S)^{\dag}$ is the
pseudoinverse of $\bv S^T \bv A \bv S$. Let $(\bv x, \lambda)$
be a solution to the generalized eigenvalue problem given by
$\bv S^T \bv A^2 \bv S \bv x=\lambda \bv S^T \bv A \bv S \bv x$.
First, we prove that the approximate eigenvector
$\bv v=\frac{\bv A \bv S \bv x}{\|\bv A \bv S \bv x \|_2}$ is
exactly an eigenvector of $\hat{\bv A}$.

\begin{lemma}\label{lem:nyseig}
Let $\bv x \in \R^{s}$ and $\lambda > 0$ be a solution to the eigenvalue
problem
$\bv S^T \bv A^2 \bv S \bv x=\lambda \bv S^T \bv A \bv S \bv x$ for a PSD
matrix $\bv A \in \R^{n \times n}$
and a matrix $\bv S \in \R^{n \times s}$.
Then, $\bv v=\frac{\bv A \bv S \bv x}{\|\bv A \bv S \bv x \|_2}$ is an eigenvector
of $(\bv A \bv S)(\bv S^T \bv A \bv S)^{\dag} (\bv A \bv S)^T$
corresponding to the eigenvalue $\lambda$.
\end{lemma}
\begin{proof}
 $(\bv A \bv S)(\bv S^T \bv A \bv S)^{\dag} (\bv A \bv S)^T\bv v
 = \frac{(\bv A \bv S)(\bv S^T \bv A \bv S)^{\dag} (\bv S^T \bv A^2 \bv S) \bv x}{\|\bv A \bv S \bv x \|_2}
 =\frac{\lambda (\bv A \bv S)(\bv S^T \bv A \bv S)^{\dag}(\bv S^T \bv A \bv S) \bv x}{\|\bv A \bv S \bv x \|_2}= \frac{\lambda\bv A \bv S \bv x}{\|\bv A \bv S \bv x \|_2}=\lambda \bv v $.
 Thus, $\bv v$ is an eigenvector of $(\bv A \bv S)(\bv S^T \bv A \bv S)^{\dag} (\bv A \bv S)^T$
 corresponding to the eigenvalue $\lambda$.
\end{proof}
Hence, by the above lemma, $\bv v$ is an eigenvector of the Nystr\"{o}m approximation
$\hat{\bv A}$ with eigenvalue $\lambda$.
To relate this to an approximate eigenvector of $\bv A$ itself, it suffices to bound
$\|\bv A - \hat{\bv A}\|_2 \leq \epsilon n$.
We do this by bounding the Nystr\"{o}m approximation error using ridge leverage score sampling.
We first recall the relevant definitions.

\begin{definition}[Ridge Leverage Score]\label{def:rls}
For a PSD matrix $\bv A \in \R^{n \times n}$, regularization parameter $\beta > 0$,
and index $i \in [n]$, the $i$\textsuperscript{th} ridge leverage score is defined as
\[
 l_i^{\beta}(\bv A) = \left[\bv A(\bv A + \beta \bv I)^{-1}\right]_{ii}.
\]
\end{definition}

\begin{definition}[Effective Dimension]\label{def:effective_dimension}
For a PSD matrix $\bv A \in \R^{n \times n}$ and regularization parameter $\beta > 0$,
the effective dimension is defined as
\[
 d_{\mathrm{eff}}^{\beta}(\bv A) = \sum_{i=1}^n l_i^{\beta}(\bv A).
\]
\end{definition}
It is known that sampling columns according to their ridge leverage scores (RLS) suffices to achieve
this spectral approximation guarantee of $\beta$~\cite{musco2017recursive}.
Formally, Theorem 3 of~\cite{musco2017recursive} states the following:

\begin{implemma}[Theorem 3 of~\cite{musco2017recursive}]\label{lem:rls}
For any $\beta>0$, let index $i \in [n]$ be sampled independently with probability
$p_i=\min(1, \tilde{l}_i^{\beta} \cdot 16\log (\sum_{j=1}^n \tilde{l}_j^{\beta}/\delta))$
where $\tilde{l}_i^{\beta} \geq l_i^{\beta}$ is an over-approximation of the
$i$\textsuperscript{th} ridge leverage score $l_i^{\beta}$.
Let $\bv S \in \R^{n \times s}$ be the corresponding sampling matrix.
Let $\hat{\bv A}=(\bv A \bv S)(\bv S^T \bv A \bv S)^{\dag} (\bv A \bv S)^T$.
Then, with probability at least $1-\delta$, we have:
\begin{align*}
 \hat{\bv A}\preceq \bv A \preceq \hat{\bv A}+\beta \bv I
\end{align*}

\end{implemma}
While computing the ridge leverage scores exactly can be expensive, Imported Lemma~\ref{lem:rls}
tells us that we only need an over-approximation of the scores for RLS sampling.
We show that for any PSD matrix $\bv A$,
the $i$\textsuperscript{th} ridge leverage score is bounded by $\bv A_{ii}/\beta$ for any $\beta > 0$. So,
sampling each column independently with probability $\tilde{O}(\bv A_{ii}/\beta)$ suffices
to get the spectral approximation guarantee of $\beta$.
\begin{lemma}\label{lem:rls_bound}
Let $\bv A \in \R^{n \times n}$ be a PSD matrix.
Then for any $\beta > 0$ and any $i \in [n]$,
\[
 l_i^{\beta}(\bv A) \leq \frac{\bv A_{ii}}{\beta}.
\]
Consequently, $d_{\mathrm{eff}}^{\beta}(\bv A) \leq \frac{\tr(\bv A)}{\beta}$.
\end{lemma}
\begin{proof}
Since $\bv A$ is PSD, we have $\bv A(\bv A + \beta \bv I)^{-1}
\preceq \frac{1}{\beta}\bv A$.
Examining the $i$th diagonal entry gives
$l_i^{\beta}(\bv A) = [\bv A(\bv A+\beta \bv I)^{-1}]_{ii} \leq \frac{1}{\beta}\bv A_{ii}$.
Summing over all $i \in [n]$ yields $d_{\mathrm{eff}}^{\beta}(\bv A) \leq \frac{\tr(\bv A)}{\beta}$.
\end{proof}

Setting $\beta = \epsilon n$ and using the fact that $0 \leq \bv A_{ii} \leq 1$
(since $\| \bv A \|_{\infty} \leq 1$) in Lemma~\ref{lem:rls_bound} gives
$l_i^{\epsilon n}(\bv A) \leq \frac{1}{\epsilon n}$
and $d_{\mathrm{eff}}^{\epsilon n}(\bv A) \leq \frac{1}{\epsilon}$.
Hence, by Imported Lemma~\ref{lem:rls} with
$\tilde{l}_i^{\beta} = \frac{1}{\epsilon n}$,
sampling each index independently with probability $\tilde{O}(1/\epsilon n)$,
i.e., $\tilde{O}(1/\epsilon)$ columns in expectation,
suffices to obtain $\|\bv A - \hat{\bv A}\|_2 \leq \epsilon n$.
We now state our main result below:
\begin{theorem}\label{thm:psd_upper}
Let $\bv A \in \R^{n \times n}$ be a PSD matrix such that
$\|\bv A \|_{\infty} \leq 1$, with eigenvalues
$\lambda_1(\bv A) \geq \ldots \geq \lambda_n(\bv A) $,
and let
$\epsilon, \delta \in (0,1)$.
Then, Algorithm~\ref{alg:psd_eigvec}, run with
sampling probability
$p = \min\!\left(1, \frac{c}{\epsilon n}\log\!\left(\frac{1}{\epsilon\delta}\right)\right)$
for a sufficiently large
absolute constant $c$ and
threshold $\tau = \frac{\epsilon n}{2}$, outputs the
pairs $\{(\tilde \lambda_i, \bv v_i) \}_{i=1}^m$,
where $\tilde \lambda_i \in \R$, $\bv v_i \in \R^n$, $\|\bv v_i \|_2=1$
for all $i \in [m]$,
and
$\tilde\lambda_1\geq \ldots \geq \tilde\lambda_m$, such that, with
probability at least $1-\delta$, for all $i \in [m]$,
\begin{align*}
 |\lambda_i(\bv A) - \tilde\lambda_i| \leq \epsilon n \text{, and } \quad
 \|\bv A \bv v_i - \lambda_i(\bv A) \bv v_i \|_2 \leq \epsilon n.
\end{align*}
Moreover, Algorithm~\ref{alg:psd_eigvec} samples
$pn = O \left(\frac{1}{\epsilon}\log \left(\frac{1}{\epsilon\delta}\right) \right)$
columns from $\bv A$ in expectation.
\end{theorem}
\begin{proof}
Let $\hat{\bv A} = (\bv A \bv S)(\bv S^T \bv A \bv S)^{\dag}(\bv A \bv S)^T$.
By Imported Lemmas~\ref{lem:rls} and~\ref{lem:rls_bound},
sampling each index in $[n]$ with probability
$p = \min\!\left(1, \frac{c}{\epsilon n}\log\!\left(\frac{1}{\epsilon\delta}\right)\right)$
ensures that, with probability at least $1-\delta$,
\[
 \|\bv A - \hat{\bv A}\|_2 \leq \frac{\epsilon n}{2}.
\]
We condition on this event for the rest of the proof; it fixes $\bv S$ (and hence $\hat{\bv A}$).
Recall that $\hat{\bv A}$ is a Nystr\"om approximation, so $\bv 0 \preceq \hat{\bv A} \preceq \bv A$.
By Lemma~\ref{lem:nyseig}, each solution $(\bv x, \lambda)$ with $\lambda>0$
of $\bv S^T \bv A^2 \bv S \bv x = \lambda \bv S^T \bv A \bv S \bv x$ yields an eigenpair
$(\bv v, \lambda)$ of $\hat{\bv A}$ with $\bv v = \frac{\bv A \bv S \bv x}{\|\bv A \bv S \bv x\|_2}$;
the algorithm returns exactly those with $\lambda \geq \frac{\epsilon n}{2}$, i.e. the eigenvectors
of $\hat{\bv A}$ whose eigenvalue is at least $\frac{\epsilon n}{2}$.

\emph{Eigenvalue bound.} Since $\bv A$ and $\hat{\bv A}$ are symmetric, Weyl's inequality~\cite{weyl1912}
gives $|\lambda_j(\bv A) - \lambda_j(\hat{\bv A})| \leq \|\bv A - \hat{\bv A}\|_2 \leq \frac{\epsilon n}{2}$
for all $j$. Fix an eigenvalue $\lambda_i$ of $\bv A$ with $\lambda_i \geq \epsilon n$ (as $\bv A$ is
PSD, all such eigenvalues are positive), and let $\tilde\lambda_i := \lambda_i(\hat{\bv A})$ be the
corresponding eigenvalue of $\hat{\bv A}$. Then $|\lambda_i - \tilde\lambda_i| \leq \frac{\epsilon n}{2}$,
and $\tilde\lambda_i \geq \lambda_i - \frac{\epsilon n}{2} \geq \frac{\epsilon n}{2}$, so
$\tilde\lambda_i$ clears the algorithm's threshold and is therefore returned,
together with a unit eigenvector $\bv v_i$ of $\hat{\bv A}$ satisfying $\hat{\bv A}\bv v_i = \tilde\lambda_i \bv v_i$.

\emph{Eigenvector residual.} For this pair, by spectral subadditivity,
\[
 \|\bv A \bv v_i - \tilde\lambda_i \bv v_i\|_2 =
 \|(\bv A - \hat{\bv A})\bv v_i\|_2 \leq
 \|\bv A - \hat{\bv A}\|_2 \|\bv v_i\|_2 =
 \|\bv A - \hat{\bv A}\|_2 \leq \frac{\epsilon n}{2},
\]
using that $\bv v_i$ is a unit eigenvector of $\hat{\bv A}$. Combining with the eigenvalue bound via
the triangle inequality,
\[
 \|\bv A \bv v_i - \lambda_i \bv v_i\|_2 \leq
 \|\bv A \bv v_i - \tilde\lambda_i \bv v_i\|_2 + |\tilde\lambda_i - \lambda_i|\,\|\bv v_i\|_2
 \leq \frac{\epsilon n}{2} + \frac{\epsilon n}{2} = \epsilon n.
\]
Both bounds hold for every $i \in [k]$ on the single event above, hence simultaneously with
probability at least $1-\delta$. Finally, since each index is sampled independently with probability
$p$, the expected number of columns sampled by Algorithm~\ref{alg:psd_eigvec} is $pn = \tilde{O}(1/\epsilon)$.
\end{proof}
We now give an improved error guarantee of $\sqrt{\epsilon n}$ for the $l_{\infty}$-norm error of the approximate eigenvectors obtained
by Algorithm~\ref{alg:psd_eigvec}.
\begin{corollary}\label{cor:psd_eiginf}
Let $\bv A $ be a PSD matrix such that $\|\bv A \|_{\infty} \leq 1$, and let $\epsilon, \delta \in (0,1)$.
Run Algorithm~\ref{alg:psd_eigvec} as in Theorem~\ref{thm:psd_upper}. Then, with probability at least
$1-\delta$, every output pair $(\tilde\lambda_i, \bv v_i)$ with $\tilde\lambda_i > 0$ additionally satisfies:
\begin{align*}
 \|\bv A \bv v_i-\tilde\lambda_i \bv v_i \|_{\infty} \leq \sqrt{\epsilon n}.
\end{align*}
\end{corollary}
\begin{proof}
Let $\hat{\bv A} = (\bv A \bv S)(\bv S^T \bv A \bv S)^{\dag}(\bv A \bv S)^T$.
Since $\hat{\bv A}$ is a Nystr\"{o}m approximation to $\bv A$, we have
$\hat{\bv A} \preceq \bv A$ or equivalently
$\bv E=\bv A - \hat{\bv A}$ is PSD~\cite{gittens2011spectral}.
We have $\bv E \bv v=\bv A \bv v-\lambda \bv v $ as $\bv v$ is an eigenvector of $\hat{\bv A}$ corresponding to the eigenvalue $\lambda$ (Lemma~\ref{lem:nyseig}).
So, we need to bound $|\bv e_i^T\bv E \bv v|$ for any $i \in [n]$ where $\bv e_i$ is the $i$\textsuperscript{th} standard basis vector.
By spectral submultiplicativity, we have $|\bv e_i^T \bv E \bv v| \leq \|\bv e_i^T\bv E\|_2 \|\bv v\|_2 \leq \|\bv e_i^T\bv E\|_2 $.
Next, observe that $\|\bv e_i^T\bv E\|^2_2= \|\bv E \bv e_i\|^2_2 =\bv e_i^T\bv E^T \bv E \bv e_i=(\bv E^2)_{ii}$.
Since $\bv E$ is PSD, we have $\bv E^2 \preceq \|\bv E\|_2 \bv E$ and thus $(\bv E^2)_{ii} \leq \|\bv E\|_2 \bv E_{ii}$.
Finally, since $\|\bv E\|_2 \leq \epsilon n$ and
$\bv E_{ii} \leq \bv A_{ii} \leq 1$, we have $|\bv e_i^T\bv E \bv v| \leq \sqrt{\epsilon n}$.
Thus, $\|\bv A \bv v-\lambda \bv v\|_{\infty} = \|\bv E \bv v\|_{\infty} \leq \sqrt{\epsilon n}$.

\end{proof}

\subsubsection{Improved Guarantees By Sampling According To Diagonal Entries}

We now state an improved guarantee for approximating the eigenvectors of a PSD matrix $\bv A$ to error $\epsilon \tr(\bv A)$
by sampling indices with probability proportional to the diagonal entry.
The algorithm is the same as Algorithm~\ref{alg:psd_eigvec}, except that instead of sampling columns uniformly,
column $i$ is sampled with probability proportional to the diagonal entry $\bv A_{ii}$.
The proof for the approximation guarantee is similar to that of Theorem~\ref{thm:psd_upper} and is stated below. Note that we
don't need the bounded entry assumption for this guarantee, and the error is in terms of $\tr(\bv A)$ instead of $n$.
For a PSD matrix with $\|\bv A \|_{\infty} \leq 1$, we always have $\tr(\bv A) \leq n$,
so the error guarantee of $\epsilon \tr(\bv A)$ is at least as good as the guarantee of Theorem~\ref{thm:psd_upper}
and can be significantly better for many matrices.
\begin{theorem}\label{thm:diag_samp}
Let $\bv A \in \R^{n \times n}$ be a PSD matrix such that
$\|\bv A \|_{\infty} \leq 1$, with eigenvalues
$\lambda_1(\bv A) \geq \ldots \geq \lambda_n(\bv A) $,
and let
$\epsilon, \delta \in (0,1)$.
Then, Algorithm~\ref{alg:psd_eigvec}, run with
sampling probability
$p_i = \min\!\left(1, \frac{c\, \bv A_{ii}}{\epsilon \tr(\bv A)}\log\!\left(\frac{1}{\epsilon\delta}\right)\right)$
for a sufficiently large
absolute constant $c$ and
threshold $\tau = \frac{\epsilon \tr(\bv A)}{2}$, outputs the
pairs $\{(\tilde \lambda_i, \bv v_i) \}_{i=1}^m$,
where $\tilde \lambda_i \in \R$, $\bv v_i \in \R^n$, $\|\bv v_i \|_2=1$
for all $i \in [m]$,
and
$\tilde\lambda_1\geq \ldots \geq \tilde\lambda_m$, such that, with
probability at least $1-\delta$, for all $i \in [m]$,
\begin{align*}
 |\lambda_i(\bv A) - \tilde\lambda_i| \leq \epsilon \tr(\bv A) \text{, and } \quad
 \|\bv A \bv v_i - \lambda_i(\bv A) \bv v_i \|_2 \leq \epsilon \tr(\bv A).
\end{align*}
Moreover, the expected number of columns sampled is
$O \left(\frac{1}{\epsilon}\log \left(\frac{1}{\epsilon\delta}\right) \right)$.
\end{theorem}
\begin{proof}
We set $\beta=\epsilon \tr(\bv A)$; then from Lemma~\ref{lem:rls_bound}, $l_i^{\beta}(\bv A) \leq \frac{\bv A_{ii}}{\epsilon \tr(\bv A)}$.
Let $\hat{\bv A} = (\bv A \bv S)(\bv S^T \bv A \bv S)^{\dag}(\bv A \bv S)^T$, where $\bv S$ is
the sampling matrix obtained by sampling each index $i \in [n]$ with probability
$p_i= \min \left(1, \frac{c \bv A_{ii}}{\epsilon \tr(\bv A)}\log\!\left(\frac{1}{\epsilon\delta}\right)\right)$.
Then, from Imported Lemma~\ref{lem:rls}, with probability at least $1-\delta$,
\[
 \|\bv A - \hat{\bv A}\|_2 \leq \frac{\epsilon \tr(\bv A)}{2}.
\]
We condition on this event, which fixes $\bv S$ (and hence $\hat{\bv A}$). As in the proof of
Theorem~\ref{thm:psd_upper}, $\bv 0 \preceq \hat{\bv A} \preceq \bv A$, and by Lemma~\ref{lem:nyseig}
the algorithm returns exactly the eigenvectors of $\hat{\bv A}$ whose eigenvalue is at least
$\tau = \frac{\epsilon \tr(\bv A)}{2}$.

For the eigenvalue bound, Weyl's inequality~\cite{weyl1912} gives
$|\lambda_j(\bv A) - \lambda_j(\hat{\bv A})| \leq \|\bv A - \hat{\bv A}\|_2 \leq \frac{\epsilon \tr(\bv A)}{2}$
for all $j$. Fix $\lambda_i \geq \epsilon \tr(\bv A)$ and set $\tilde\lambda_i := \lambda_i(\hat{\bv A})$;
then $|\lambda_i - \tilde\lambda_i| \leq \frac{\epsilon \tr(\bv A)}{2}$ and
$\tilde\lambda_i \geq \frac{\epsilon \tr(\bv A)}{2} = \tau$, so $\tilde\lambda_i$ clears the threshold
and is returned with a unit eigenvector $\bv v_i$ of $\hat{\bv A}$. By spectral subadditivity,
$\|\bv A \bv v_i - \tilde\lambda_i \bv v_i\|_2 = \|(\bv A - \hat{\bv A})\bv v_i\|_2 \leq \|\bv A - \hat{\bv A}\|_2 \leq \frac{\epsilon \tr(\bv A)}{2}$,
and hence by the triangle inequality
$\|\bv A \bv v_i - \lambda_i \bv v_i\|_2 \leq \frac{\epsilon \tr(\bv A)}{2} + \frac{\epsilon \tr(\bv A)}{2} = \epsilon \tr(\bv A)$.
Both bounds hold for every $i \in [k]$ on the single event above. Finally, since index $i$ is sampled
independently with probability $p_i$, the expected number of columns sampled is
$\sum_{i=1}^n p_i \leq O((1/\epsilon) \log(1/\epsilon \delta))$.
\end{proof}
We can also get a similar $l_{\infty}$-norm error guarantee for the approximate eigenvectors
as in Corollary~\ref{cor:psd_eiginf}.
\begin{corollary}
Let $\bv A $ be a PSD matrix, and let $\epsilon, \delta \in (0,1)$. Run Algorithm~\ref{alg:psd_eigvec}
as in Theorem~\ref{thm:diag_samp}, with column $i \in [n]$ sampled independently with probability
$p_i= \min \left(1, \frac{c \bv A_{ii}}{\epsilon \tr(\bv A)}\log\!\left(\frac{1}{\epsilon\delta}\right)\right)$
for a sufficiently large absolute constant $c$. Then, with probability at least $1-\delta$, every output
pair $(\tilde\lambda, \bv v)$ with $\tilde\lambda > 0$ additionally satisfies, for every $j \in [n]$,
\begin{align*}
 |\bv A_{j} \bv v-\tilde\lambda \bv v_j | \leq \sqrt{\epsilon \tr(\bv A)\, \bv A_{jj}}.
\end{align*}
Moreover, for any $j \in[n]$ such that index $j$ is in the sampled set of columns, we have
$\bv A_j \bv v=\tilde\lambda \bv v_j$.
\end{corollary}
\begin{proof}
The proof is similar to that of Corollary~\ref{cor:psd_eiginf}, applied to each output pair
$(\tilde\lambda, \bv v)$. Let $\bv E = \bv A - \hat{\bv A}$. Then, following the proof of
Corollary~\ref{cor:psd_eiginf},
we get $|\bv A_j \bv v - \tilde\lambda \bv v_j| = |\bv e_j^T \bv E \bv v| \leq \sqrt{(\bv E^2)_{jj}}$.
We also have $(\bv E^2)_{jj} \leq \| \bv E \|_2 \bv E_{jj}$.
Since $\|\bv E\|_2 \leq \epsilon \tr(\bv A)$ and $\bv E_{jj} \leq \bv A_{jj}$, we have $|\bv A_j \bv v - \tilde\lambda \bv v_j| \leq \sqrt{\epsilon \tr(\bv A)\, \bv A_{jj}}$.
\end{proof}

\subsubsection{Eigenvalue Approximation of Bounded-Entry PSD Matrices}
By combining Imported Lemmas~\ref{lem:rls} and~\ref{lem:rls_bound} with known
randomized sampling algorithms (in~\cite{Bhattacharjee:2021wl} and~\cite{swartworth2025tight})
, we can approximate the eigenvalues of a bounded entry PSD matrix $\bv A$ to error $\epsilon n$ by sampling
just $\tilde{O}(1/\epsilon^3)$ entries of $\bv A$.
In fact, the observation that one can combine RLS sampling with
randomized sampling algorithms to obtain improved eigenvalue approximation guarantees for bounded entry PSD matrices
was first made in~\cite{swartworth2025tight}. But
the authors of~\cite{swartworth2025tight} claim
that the algorithm needs to be nonadaptive to achieve this guarantee,
since the RLS sampling probabilities are not known a priori and need to be estimated.
However, as we have shown in Lemma~\ref{lem:rls_bound}, for bounded entry PSD matrices, the RLS sampling probabilities are uniformly bounded by $1/\epsilon n$, so
uniform column sampling suffices.
Thus, the algorithm suggested by~\cite{swartworth2025tight} (Theorem 8.3) can be made non-adaptive.
We state this result formally below for completeness, and give a proof sketch.
\begin{theorem}\label{thm:psd_eigs}
Let $\bv A$ be a PSD matrix such that $\|\bv A \|_{\infty} \leq 1$. Then, there is a non-adaptive algorithm
that reads $\tilde{O}\left(\frac{1}{\epsilon^3 } \right)$ entries of $\bv A$ and
with probability $2/3$, approximates all the eigenvalues of $\bv A$ to additive error $\epsilon n$.
\end{theorem}
\begin{proof}
We first observe that by Imported Theorem~\ref{thm:main-eigval}, we can produce
a rescaled sampling matrix
$\frac{n}{s}\bv S^T \bv A \bv S$ of size $s \times s$, where $s=\tilde{O}(1/\epsilon^2)$,
which approximates
all eigenvalues of $\bv A$ to additive error $\epsilon n$ with good (constant) probability.
Thus, it suffices to approximate the eigenvalues $\bv S^T \bv A \bv S$
to error $\epsilon s$. Note that we do not actually produce the matrix $\bv S^T \bv A \bv S$ explicitly,
but we we have query access to its entries simply by querying entries of $\bv A$.
Also note that $\bv S^T \bv A \bv S$ is also a bounded entry PSD matrix.
Setting $\beta = \epsilon s$ in Lemma~\ref{lem:rls_bound}
gives that the RLS sampling probabilities of $\bv S^T \bv A \bv S$ are
bounded by $1/\epsilon s$. Thus, by Imported Lemma~\ref{lem:rls}, by sampling
$\tilde{O}(1/\epsilon)$ columns from $\bv S^T \bv A \bv S$ in total,
we get an $\hat{\bv A}$ such that $\|\bv S^T \bv A \bv S- \hat{\bv A}\|_2\leq \epsilon s$
with good probability. So, the total number of entries we read are
$\tilde{O}(1/\epsilon^3)$ ($\bv S^T \bv A \bv S$ has $s=\tilde{O}(1/\epsilon^2)$ rows and we read
only $\tilde{O}(1/\epsilon)$ columns from it).
We finally use Weyls' inequality~\cite{weyl1912} and
adjust $\epsilon$ by constant factors
to get the desired guarantee.
\end{proof}

\subsubsection{Lower Bound for PSD Matrices}

We now show that sampling $\Omega(1/\epsilon)$ columns are necessary to approximate even the top
eigenvector of a bounded entry PSD matrix to error $\epsilon n$. To see this,
suppose there is an $n \times n$ matrix with a set of indices $s \in [n]$ of size $|s|=\epsilon n$
such that $\bv A_{ij}=1$ if $i,j \in s$ or $i=j$ and $\bv A_{ij}=0$ otherwise, i.e. $\bv A$ is a matrix
with ones on the diagonal and a block of ones of size
$\epsilon n \times \epsilon n$. Note that $\bv A$ has one eigenvalue equal to $\epsilon n$,
and its other eigenvalues are either 0 or 1.
Any algorithm that samples fewer than $O(1/\epsilon)$ columns will
miss the block of ones with constant probability, leading to an error of $\Omega(\epsilon n)$ in approximating the top eigenvector.
This simple lower bound shows that the upper bound of $\tilde{O}(1/\epsilon)$ columns of Theorem~\ref{thm:psd_upper} is tight up to logarithmic factors. We state the lower bound formally below:
\begin{theorem}\label{thm:lower_psd}
Let $\mathcal{A}$ be any randomized algorithm that
(possibly adaptively) reads columns of a
symmetric PSD matrix $\bv A \in \R^{n \times n}$ with $\|\bv A\|_{\infty} \leq 1$ and largest magnitude
eigenvalue $\lambda_1(\bv A)$, and
outputs a unit vector $\bv v$ that, for any $\epsilon \in (0,1)$ satisfies:
\begin{align*}
 \|\bv A \bv v-\lambda_1(\bv A) \bv v \|_2 \leq \epsilon n,
\end{align*}
with probability at least $\frac{2}{3}$.
Then, $\mathcal{A}$ must read at least $\Omega(\frac{1}{\epsilon})$ columns of $\bv A$ provided $\epsilon=\Omega(\frac{1}{n})$.
\end{theorem}

\subsection{Improved Algorithms Using Rank-Truncated Nystr\"{o}m}\label{sec:gennys}

In this section, we present algorithms based on a rank-truncated Nystr\"{o}m method for
general symmetric (and potentially indefinite) matrices.
The algorithms we present here achieve \emph{optimal sample complexity}
(up to logarithmic factors) for approximating all outlying eigenvectors,
requiring only $\widetilde{O}(1/\epsilon^2)$ sampled columns.
As proved in Theorem~\ref{thm:lower_main}, approximating even the top eigenvector
of a bounded-entry matrix to $\epsilon n$ residual error requires sampling
$\Omega(1/\epsilon^2)$ columns.

For a general symmetric matrix $\bv A$, the scaled principal submatrix $\Sbb^T \bv A \Sbb$ is no longer positive semidefinite and may contain small eigenvalues that cause the standard Nystr\"{o}m approximation $\bv A \Sbb(\Sbb^T \bv A \Sbb)^\dagger \Sbb^T \bv A$ to have large spectral error. To address this, we construct a rank-truncated Nystr\"{o}m approximation $\hat{\bv A} = \bv C [\Sbb^T \bv A \Sbb]_\tau^\dagger \bv C^T$, where $\bv C = \bv A \bar{\bv S}$ and $[\Sbb^T \bv A \Sbb]_\tau^\dagger$ is a \emph{truncated} pseudoinverse restricted to eigenvalues of magnitude at least $\tau$. By proving that $\|\bv A - \hat{\bv A}\|_2$ is small, we show that we can extract high-quality approximate eigenvectors for $\bv A$ by solving a generalized eigenvalue problem involving the smaller $s \times s$ sampled matrices.

\subsubsection{Uniform Sampling For Bounded-Entry Matrices.}

We first describe the approach for bounded-entry matrices using uniform sampling, and then extend it to general matrices using squared column-norm sampling.
We construct the approximation as follows.
Let $\bar{\bv S} \in \mathbb{R}^{n \times s}$
be the scaled uniform sampling matrix,
which samples $s = \widetilde{O}(1/\epsilon^2)$
columns uniformly at random and rescales them by $\sqrt{n/s}$.
Let $\bv C = \bv A \bar{\bv S}$
and let $\hat{\bv A} = \bv C \bv U \bv C^T$
be the rank-truncated Nystr\"om approximation~\cite{nakatsukasa2020fast},
where $\bv U = [\Sbb^T \bv A \Sbb]_{\tau}^\dagger$ is the
truncated pseudoinverse of $\Sbb^T \bv A \Sbb$,
restricted to eigenvalues with magnitude at least
$\tau \geq \frac{\epsilon n}{2}$.
For the full pseudocode, see Algorithm~\ref{alg:genNystr\"{o}m}.

Throughout this section, we will set $\bv A_o, \bv A_m$ as in Definition~\ref{def:ao-am} with $L+\epsilon \sqrt{\delta}n$. We start with a definition regarding the eigendecomposition of $\bv W=\Sbb^T \bv A \Sbb$ we will use throughout our proofs.
\begin{definition}\label{def:w}
Let $\bv W = \bv V \bv \Sigma \bv V^T$ be the eigendecomposition
of $\bv W = \Sbb^T \bv A \Sbb$. Then $\bv W = \bv V_1 \bv \Sigma_1 \bv V_1^T + \bv V_2 \bv \Sigma_2 \bv V_2^T$, where $\bv V_1$ has the eigenvectors of $\bv W$ corresponding to its eigenvalues $\Sigma_1$ with at least $\tau=\frac{\epsilon n}{2}$ in magnitude. $\bv V_2$ contains the eigenvectors corresponding to $\Sigma_2$ which are $<\tau=\frac{\epsilon n}{2}$ in magnitude.
\end{definition}
Throughout our proofs, the only two properties of the sampling matrix $\Sbb$ we will require are those from Imported Lemma~\ref{lem:subspace-embedding-bounded} with $R=\epsilon n$, Imported Lemma~\ref{lem:am}, and Lemma~\ref{lem:as}. We formally state this as a property of $\Sbb$:
\begin{property}\label{prop:s}
For $\bv A_o=\bv U_o \bv \Lambda_o \bv U_o, \bv A_m= \bv U_m \bv \Lambda_m \bv U_m$ with $L=\epsilon \sqrt{\delta} n$ in Definition~\ref{def:ao-am}, let $\Sbb$ be a matrix which satisfies the following:
\begin{enumerate}
 \item $\Sbb$
is a constant-factor subspace embedding for $\bv U_o$, which gives:
\begin{align*}
 \sigma_{\min}(\Sbb^T \bv U_o) = \Omega(1) \quad \text{and thus,}
 \quad \|(\Sbb^T \bv U_o)^\dagger\|_2 = O(1).
\end{align*}
 \item $\|\Sbb^T \bv A_m \Sbb \|_2 \leq \epsilon n$.
 \item $\|\bv A_m \Sbb \|_2 \leq \epsilon n$.
\end{enumerate}
\end{property}

\begin{algorithm}[t]
\caption{Rank-Truncated Nystr\"om Using Uniform Sampling}
\label{alg:genNystr\"{o}m}
\begin{algorithmic}[1]
\Require Symmetric matrix $\bv A \in \R^{n \times n}$ with
$\|\bv A\|_{\infty} \leq 1$,
accuracy $\epsilon \in (0,1)$, expected sample size $s$.
\State Let $\bv{\bar S} \in \R^{n \times |S|}$ be the scaled sampling matrix
which samples each column of $\bv A$ independently with probability $\frac{s}{n}$
and then rescales each sampled column by $\sqrt{\frac{n}{s}}$.
\State Find all eigenvectors $\bv x \in \R^s$ corresponding to the
eigenvalues $\lambda$
(with magnitude at least $ \epsilon n$)
by solving the generalized eigenvalue problem:
\[
 \Sbb^T \bv A^2 \Sbb \, \bv x =
 \lambda \, [\Sbb^T \bv A \Sbb]_{\tau} \, \bv x,
\]
where $[\Sbb^T \bv A \Sbb]_{\tau}$ is the truncated eigendecomposition
of the matrix $\Sbb^T \bv A \Sbb $ restricted to the eigenvalues
with magnitude at least $\tau \geq \frac{\epsilon n}{2}$.
\State Compute the approximate eigenvectors $\bv v$
corresponding to the eigenvectors $\bv x$ found in the previous step as:
\[
 \bv v = \dfrac{\bv A \Sbb \bv x}{\|\bv A \Sbb \bv x\|_2}.
\]
\State \Return all approximate eigenvector, eigenvalue pairs
$(\bv v, \lambda)$, found in the previous step.
\end{algorithmic}
\end{algorithm}

We start by showing that for any such $\Sbb$, the spectral norm of $\|\bv A_o \Sbb \bv U \|_2$ and $\|\bv A_o \Sbb \bv V_2\|_2$ is at most a constant.
\begin{lemma}\label{lem:asu}
Let $\bv U = [\Sbb^T \bv A \Sbb]_{\tau}^\dagger$
be the truncated pseudoinverse of $\Sbb^T \bv A \Sbb$,
restricted to eigenvalues with magnitude at least $\tau \geq \frac{\epsilon n}{2}$. Let $\Sbb$ satisfy the conditions in Property~\ref{prop:s}. Let $\bv W = \bv V_1 \bv \Sigma_1 \bv V_1^T + \bv V_2 \bv \Sigma_2 \bv V_2^T$ be as in Definition~\ref{def:w}. Then, for constants $C_1, C_2>0$:
\begin{align*}
 \|\bv A_o \Sbb \bv U \|_2 = C_1 \quad \text{ and, } \quad \|\bv A_o \Sbb \bv V_2\|_2 =C_2
 \quad
\end{align*}
\end{lemma}
\begin{proof}

 The truncated pseudoinverse
$\bv U = [\Sbb^T \bv A \Sbb]_\tau^\dagger$ keeps exactly the eigenvalues of $\bv W$ with magnitude at least $\tau = \frac{\epsilon n}{2}$. Then, $\bv U = [\Sbb^T \bv A \Sbb]_{\tau}^{\dagger}=\bv V_1 \bv \Sigma_1^{-1} \bv V_1^T$,
and $[\Sbb^T \bv A \Sbb]_\tau = \bv V_1 \bv \Sigma_1 \bv V_1^T$. Note that we have
\begin{align*}
 \|\bv \Sigma_1^{-1}\|_2 \le \frac{2}{\epsilon n} = O\!\left(\frac{1}{\epsilon n}\right)
 \quad \text{and }
 \|\bv \Sigma_2\|_2 < \frac{\epsilon n}{2} = O(\epsilon n).
\end{align*}
Let $\bv W=\bv W_o+\bv W_m$ where $\bv W_o=\Sbb^T \bv A_o \Sbb$ and $\bv W_m=\Sbb^T \bv A_m \Sbb$. We have $\bv W \bv V_1 = \bv V_1 \bv \Sigma_1$.
Substituting $\bv W = \Sbb^T \bv A_o \Sbb + \bv W_m$ yields:
\begin{align*}
 (\Sbb^T \bv A_o\Sbb + \bv W_m) \bv V_1 &= \bv V_1 \bv \Sigma_1
 \implies \Sbb^T \bv A_o\Sbb \bv V_1 = \bv V_1 \bv \Sigma_1 - \bv W_m \bv V_1.
\end{align*}
Since, we have $\|(\Sbb^T \bv U_o)^\dagger \|_2=O(1)$, $(\Sbb^T \bv U_o)^\dagger \Sbb^T \bv U_o = \bv I$.So, left-multiplying both sides above by $\bv U_o(\Sbb^T \bv U_o)^\dagger$, we get:
\begin{align}\label{eq:id1}
 \bv A_o\Sbb \bv V_1 &= \bv U_o(\Sbb^T \bv U_o)^\dagger \bv V_1 \bv \Sigma_1 -
 \bv U_o(\Sbb^T \bv U_o)^\dagger \bv W_m \bv V_1.
\end{align}
Similarly, starting from the fact $\bv W \bv V_2 = \bv V_2 \bv \Sigma_2$ and following the steps above, we have:
\begin{align} \label{eq:id3}
 \bv A_o \Sbb \bv V_2 =
 \bv U_o(\Sbb^T \bv U_o)^\dagger \bv V_2 \bv \Sigma_2 -
 \bv U_o(\Sbb^T \bv U_o)^\dagger \bv W_m \bv V_2.
 \end{align}
Right-multiplying~\eqref{eq:id1} by $\bv \Sigma_1^{-1}$ gives
$\bv A_o \Sbb \bv V_1 \bv \Sigma_1^{-1} = \bv U_o(\Sbb^T \bv U_o)^\dagger \bv V_1
- \bv U_o(\Sbb^T \bv U_o)^\dagger \bv W_m \bv V_1 \bv \Sigma_1^{-1}$, so
\begin{align}\label{eq:v1}
 \|\bv A_o \Sbb \bv V_1 \bv \Sigma_1^{-1}\|_2
 &= \|\bv U_o(\Sbb^T \bv U_o)^\dagger \bv V_1
 - \bv U_o(\Sbb^T \bv U_o)^\dagger \bv W_m \bv V_1 \bv \Sigma_1^{-1}\|_2 \nonumber \\
 &= \|(\Sbb^T \bv U_o)^\dagger \bv V_1
 - (\Sbb^T \bv U_o)^\dagger \bv W_m \bv V_1 \bv \Sigma_1^{-1}\|_2 \nonumber \\
 &\le \|(\Sbb^T \bv U_o)^\dagger \|_2 (1
 + \|\bv W_m\|_2 \|\bv \Sigma_1^{-1}\|_2) \nonumber \\
 &\le O(1) \cdot \left(1 + O(\epsilon n) \cdot O\left(\frac{1}{\epsilon n}\right) \right)
 = O(1).
\end{align}
This proves the first result. Now we prove the second bound. From~\eqref{eq:id3} and the triangle inequality:
\begin{align*}
 \|\bv A_o \Sbb \bv V_2\|_2
 &= \|\bv U_o(\Sbb^T \bv U_o)^\dagger \bv V_2 \bv \Sigma_2 - \bv U_o(\Sbb^T \bv U_o)^\dagger \bv W_m \bv V_2 \|_2\\
 &= \|(\Sbb^T \bv U_o)^\dagger \bv V_2 \bv \Sigma_2 - (\Sbb^T \bv U_o)^\dagger \bv W_m \bv V_2\|_2 \\
 &\le \|(\Sbb^T \bv U_o)^\dagger \|_2 (\|\bv \Sigma_2\|_2 + \|\bv W_m\|_2) \\
 &\le O(1) \cdot O(\epsilon n) = O(\epsilon n).
\end{align*}
The last step follows from the fact that $\|(\Sbb^T \bv U_o)^\dagger \|_2 = O(1)$
and $\|\bv \Sigma_2\|_2$ and $\|\bv W_m\|_2$ are both bounded by $O(\epsilon n)$.
\end{proof}

In the next lemma, we bound
the spectral norm error of
the $\|\bv A-\hat{\bv A} \|_2$.
\begin{lemma} \label{lem:Nystr\"{o}m_spectral}
Let $\bv A \in \mathbb{R}^{n \times n}$ be symmetric
with $\|\bv A\|_\infty \le 1$.
Let $\bv{\bar S} \in \R^{n \times |S|}$ be the scaled sampling matrix
which samples each column of $\bv A$ independently with probability $\frac{s}{n}$
and then rescales each sampled column by $\sqrt{\frac{n}{s}}$.
Let $\bv C = \bv A \bar{\bv S}$ and
let $\bv U = [\Sbb^T \bv A \Sbb]_{\tau}^\dagger$
be the truncated pseudoinverse of $\Sbb^T \bv A \Sbb$,
restricted to eigenvalues with magnitude at least $\tau \geq \frac{\epsilon n}{2}$.
Let $\hat{\bv A}=\bv C \bv U \bv C^T$.
For $s \ge \frac{c \log n}{\epsilon^2 \delta}\log \frac{1}{\epsilon\delta}$,
with probability at least $1-\delta$:
\begin{align*}
 \|\bv A - \hat{\bv A}\|_2 \le \epsilon n.
\end{align*}
\end{lemma}
\begin{proof}
Let $\bv A = \bv A_o + \bv A_m$ as in
Definition~\ref{def:ao-am} with threshold $L = \epsilon \sqrt{\delta} n$. Let
\begin{align*}
 \bv C &= \bv A \Sbb = \bv C_o + \bv C_m,
 \quad \text{where }
 \bv C_o = \bv A_o \Sbb \text{ and } \bv C_m =
 \bv A_m \Sbb. \\
 \bv W &= \Sbb^T \bv A \Sbb =
 \bv W_o + \bv W_m, \quad \text{where }
 \bv W_o = \Sbb^T \bv A_o \Sbb \text{ and }
 \bv W_m = \Sbb^T \bv A_m \Sbb.
\end{align*}
Note that by Imported Lemma~\ref{lem:subspace-embedding-bounded} (with $R=\epsilon \sqrt{\delta} n$), Imported Lemma~\ref{lem:am} and Lemma~\ref{lem:as}, for $s \ge \frac{c \log n}{\epsilon^2 \delta}\log \frac{1}{\epsilon\delta}$, $\Sbb$ satisfies the conditionsin Property~\ref{prop:s} with probability at least $1-\delta$, with a union bound and adjusting $\delta$ by constants. Hence, $\Sbb$ also satisfies the ocnditions of Lemma~\ref{lem:asu}.

We substitute $\bv C = \bv C_o + \bv C_m$ into $\hat{\bv A}$:
\begin{align}\label{eq:hata}
 \hat{\bv A} &= (\bv C_o + \bv C_m) \bv U (\bv C_o^T + \bv C_m^T) \nonumber\\
 &= \underbrace{\bv C_o \bv U \bv C_o^T}_{\bv T_1}
 + \underbrace{\bv C_o \bv U \bv C_m^T}_{\bv T_2}
 + \underbrace{\bv C_m \bv U \bv C_o^T}_{\bv T_3}
 + \underbrace{\bv C_m \bv U \bv C_m^T}_{\bv T_4}.
\end{align}
We will bound each of the terms
$\bv T_1$, $\bv T_2$, $\bv T_3$, and $\bv T_4$ individually.

\smallskip
\noindent\textbf{ Bounding $\bv T_1$:} Substituting
$\bv U =\bv V_1 \bv \Sigma_1^{-1} \bv V_1^T$ and $\bv C_o = \bv A_o\Sbb$, we get:
\begin{align}\label{eq:t1}
 \bv T_1 &= \bv A_o\Sbb \bv V_1 \bv \Sigma_1^{-1} \bv V_1^T \Sbb^T \bv A_o \nonumber \\
 &=\bv A_o \Sbb \bv V_1 \bv \Sigma_1^{-1} [\bv U_o(\Sbb^T \bv U_o)^\dagger \bv V_1 \bv \Sigma_1 -
 \bv U_o(\Sbb^T \bv U_o)^\dagger \bv W_m \bv V_1]^T \nonumber \\
 &= \bv A_o \Sbb (\bv V_1 \bv V_1^T) ((\Sbb^T \bv U_o)^\dagger )^T \bv U_o^T
 - \bv A_o \Sbb \bv V_1 \bv \Sigma_1^{-1} \bv V_1^T \bv W_m ((\Sbb^T \bv U_o)^\dagger )^T \bv U_o^T.
\end{align}
For the
second step, we substitute $(\bv A_o \Sbb \bv V_1)^T$ from~\eqref{eq:id1} in Lemma~\ref{lem:asu}. Using the fact that $\bv V_1 \bv V_1^T = \bv I - \bv V_2 \bv V_2^T$, we get:
\begin{align*}
 \bv A_o \Sbb (\bv I -
 \bv V_2 \bv V_2^T) ((\Sbb^T \bv U_o)^\dagger )^T \bv U_o^T
 = \bv A_o \Sbb ((\Sbb^T \bv U_o)^\dagger )^T \bv U_o^T -
 \bv A_o\Sbb \bv V_2 \bv V_2^T ((\Sbb^T \bv U_o)^\dagger )^T \bv U_o^T.
\end{align*}
Since $\bv U_o^T \Sbb ((\Sbb^T \bv U_o)^\dagger )^T
= ((\Sbb^T \bv U_o)^\dagger \Sbb^T \bv U_o)^T
= \bv I$,
for the first term above, we get $\bv A_o \Sbb ((\Sbb^T \bv U_o)^\dagger )^T \bv U_o^T=\bv U_o \bv \Sigma_o (\bv U_o\Sbb) ((\Sbb^T \bv U_o)^\dagger )^T \bv U_o^T=\bv U_o \bv \Lambda_o \bv U_o^T = \bv A_o$.
Thus, from~\eqref{eq:t1} we have:
\begin{align}\label{eq:t1err}
 \bv T_1 = \bv A_o - \bv E_1 - \bv E_2,
\end{align}
where the error terms are:
\begin{align*}
 \bv E_1 &= \bv A_o \Sbb \bv V_2 \bv V_2^T ((\Sbb^T \bv U_o)^\dagger )^T \bv U_o^T, \\
 \bv E_2 &= \bv A_o \Sbb \bv V_1 \bv \Sigma_1^{-1} \bv V_1^T \bv W_m ((\Sbb^T \bv U_o)^\dagger )^T \bv U_o^T.
\end{align*}
We bound $\|\bv E_1\|_2$.
Since $\|\bv U_o\|_2 = \|\bv V_2\|_2 = 1$, $\|((\Sbb^T \bv U_o)^\dagger )^T\|_2 = O(1)$, and $\|\bv A_o \Sbb \bv V_2 \|_2=O(1)$ from Lemma~\ref{lem:asu}, by spectral submultiplicativity, we have:
\begin{align*}
 \|\bv E_1\|_2 \le O(\epsilon n).
\end{align*}
We next bound $\|\bv E_2\|_2$. From Lemma~\ref{lem:asu} and pectral submultiplicativity, we have:
\begin{align*}
 \|\bv E_2\|_2 &\le \|\bv A_o \Sbb \bv V_1 \bv \Sigma_1^{-1} \|_2
 \cdot \|\bv V_1^T\|_2 \cdot \|\bv W_m \|_2
 \|((\Sbb^T \bv U_o)^\dagger )^T \|_2\|\bv U_o^T\|_2 \\
 &\le 1 \cdot O(1) \cdot 1 \cdot O(\epsilon n) \cdot O(1) \cdot 1 = O(\epsilon n).
\end{align*}
Therefore, from~\eqref{eq:t1err}, we have:
\begin{align}\label{eq:aot1}
 \|\bv A_o - \bv T_1\|_2 &\le \|\bv E_1\|_2 + \|\bv E_2\|_2 \le O(\epsilon n).
\end{align}

\noindent \textbf{Bounding the Cross Terms $\bv T_2, \bv T_3$:}
We first bound $\|\bv C_o \bv U\|_2$. From Lemma~\ref{lem:asu}, we have:
\begin{align*}
 \|\bv C_o \bv U\|_2
 &= \|\bv A_o \Sbb \bv V_1 \bv \Sigma_1^{-1} \bv V_1^T\|_2\\
 &= \|\bv A_o \Sbb \bv V_1 \bv \Sigma_1^{-1}\|_2 \|\bv V_1^T\|_2 \le O(1).
\end{align*}
Also, we have:
\begin{align*}
 \|\bv C_m\|_2 &= \|\bv A_m \Sbb \|_2 \le O(\epsilon n).
\end{align*}
Thus, we can bound $\bv T_2$ and $\bv T_3$ as:
\begin{align*}
 \|\bv T_2\|_2 &= \|\bv C_o \bv U \bv C_m^T\|_2
 \le \|\bv C_o \bv U\|_2 \|\bv C_m^T\|_2 \le O(1)
 \cdot O(\epsilon n) = O(\epsilon n).
\end{align*}

\noindent \textbf{Bounding $\bv T_4$:} Finally,
we can bound $\bv T_4$ as follows:
\begin{align*}
 \|\bv T_4\|_2
 &= \|\bv C_m \bv U \bv C_m^T\|_2
 \le \|\bv C_m\|_2^2 \|\bv U\|_2
 \le O(\epsilon n)^2 \cdot O\left(\frac{1}{\epsilon n}\right)
 = O(\epsilon n).
\end{align*}
Finally, from~\eqref{eq:hata}, we have:
\begin{align*}
 \bv A - \hat{\bv A} &= (\bv A_o + \bv A_m)
 - (\bv T_1 + \bv T_2 + \bv T_3 + \bv T_4) \\
 &= \bv A_m + (\bv A_o - \bv T_1) - \bv T_2 - \bv T_3 - \bv T_4.
\end{align*}
Applying the triangle inequality and
the bounds on $\bv A_o - \bv T_1$ from~\eqref{eq:aot1} and
the bounds on $\bv T_2$, $\bv T_3$, and $\bv T_4$
from above, we have:
\begin{align*}
 \|\bv A - \hat{\bv A}\|_2 &\le \|\bv A_m\|_2
 + \|\bv A_o - \bv T_1\|_2 + \|\bv T_2\|_2
 + \|\bv T_3\|_2 + \|\bv T_4\|_2 \\
 &\le \epsilon n + O(\epsilon n)
 + O(\epsilon n) + O(\epsilon n) + O(\epsilon n)
 = O(\epsilon n).
\end{align*}
We can adjust $\epsilon$ by constant factors. This concludes the proof.
\end{proof}

We now state the final theorem about eigenvector approximation.
\begin{theorem}\label{thm:gen-nystrom-eigvec}
Let $\bv A \in \R^{n \times n}$ be a symmetric matrix such that
$\|\Ab \|_{\infty} \leq 1$, with eigenvalues $\lambda_1(\bv A) \geq \ldots \geq \lambda_n(\bv A)$,
and let $\epsilon, \delta \in (0,1)$.
Let Algorithm~\ref{alg:genNystr\"{o}m} output the pairs
$\{(\tilde \lambda_j, \bv v_j)\}_{j=1}^{m}$, with $\bv A$ and $s=\frac{c \log n}{\epsilon^2 \delta}\log \frac{1}{\epsilon\delta}$ as inputs (for some sufficiently large constant $c$), where $\tilde \lambda_j \in \R$, $\bv v_j \in \R^n$,
$\|\bv v_j \|_2 = 1$, and let $m_+$ (resp.\ $m_-$) be the number of output pairs with
$\tilde\lambda_j > 0$ (resp.\ $\tilde\lambda_j < 0$), so that $m = m_+ + m_-$. Then,
with probability at least $1-\delta$:
\begin{enumerate}
\item There is a bijection $\pi$ between the output pairs and the $m_+$ largest together with the
$m_-$ smallest eigenvalues of $\bv A$, i.e.\ the set
$\{\lambda_1(\bv A), \ldots, \lambda_{m_+}(\bv A)\} \cup \{\lambda_{n-m_-+1}(\bv A), \ldots, \lambda_n(\bv A)\}$,
such that every output pair $j \in [m]$ satisfies
\begin{align*}
 |\lambda_{\pi(j)}(\bv A)-\tilde \lambda_{j}| \leq \epsilon n \text{, and } \quad
 \|\bv A \bv v_{j}-\lambda_{\pi(j)}(\bv A) \bv v_{j} \|_2 \leq \epsilon n.
\end{align*}
\item Every eigenvalue of $\bv A$ outside this matched multiset has magnitude less than $\epsilon n$.
\end{enumerate}
Moreover, Algorithm~\ref{alg:genNystr\"{o}m} samples
$\frac{c \log n}{\epsilon^2 \delta}\log \frac{1}{\epsilon\delta}$ columns
from $\bv A$ in expectation.
\end{theorem}
\begin{proof}
The proof parallels that of Theorem~\ref{thm:main_bounded}, using the rank-truncated Nystr\"om spectral
bound of Lemma~\ref{lem:Nystr\"{o}m_spectral} in place of the principal submatrix bounds.
Let $\hat{\bv A} = (\bv A \Sbb)[\Sbb^T \bv A \Sbb]_{\tau}^{\dagger}(\bv A \Sbb)^T$ with truncation
threshold $\tau = \frac{\epsilon n}{2}$. By Lemma~\ref{lem:Nystr\"{o}m_spectral}, for
$s \geq \frac{c \log n}{\epsilon^2 \delta}\log \frac{1}{\epsilon\delta}$, we have
\[
 \|\bv A - \hat{\bv A}\|_2 \leq \frac{\epsilon n}{2}
\]
with probability at least $1-\delta$. We condition on this event for the rest of the proof; it fixes
$\Sbb$ and hence $\hat{\bv A}$.

\emph{The output vectors are exact eigenvectors of $\hat{\bv A}$.} For any solution $(\bv x, \lambda)$
of the generalized eigenvalue problem
\[
 \Sbb^T \bv A^2 \Sbb \, \bv x = \lambda \, [\Sbb^T \bv A \Sbb]_{\tau} \, \bv x
\]
with $|\lambda| \geq \frac{\epsilon n}{2}$ (which captures all non-zero eigenvalues retained by the truncation), the associated output vector $\bv v = \frac{\bv A \Sbb \bv x}{\|\bv A \Sbb \bv x\|_2}$
is an eigenvector of $\hat{\bv A}$ with eigenvalue $\lambda$: exactly as in Lemma~\ref{lem:nyseig}
(using $\Sbb^T \bv A^2 \Sbb = (\bv A\Sbb)^T (\bv A\Sbb)$ since $\bv A$ is symmetric, and
$[\Sbb^T \bv A \Sbb]_{\tau}^{\dagger}[\Sbb^T \bv A \Sbb]_{\tau}\bv x = \bv x$ for $\bv x$ in the large
subspace),
\[
 \hat{\bv A}\bv v = \frac{(\bv A\Sbb)[\Sbb^T \bv A \Sbb]_{\tau}^{\dagger}(\Sbb^T \bv A^2 \Sbb)\bv x}{\|\bv A\Sbb\bv x\|_2}
 = \frac{\lambda\,(\bv A\Sbb)[\Sbb^T \bv A \Sbb]_{\tau}^{\dagger}[\Sbb^T \bv A \Sbb]_{\tau}\bv x}{\|\bv A\Sbb\bv x\|_2}
 = \frac{\lambda\,\bv A\Sbb\bv x}{\|\bv A\Sbb\bv x\|_2} = \lambda\bv v.
\]

\emph{The two bounds.} Since $\bv A$ and $\hat{\bv A}$ are symmetric, Weyl's
inequality~\cite{weyl1912} gives $|\lambda_i(\bv A) - \lambda_i(\hat{\bv A})| \leq \|\bv A - \hat{\bv A}\|_2 \leq \frac{\epsilon n}{2}$
for all $i \in [n]$.

Let $\tilde \lambda_1 \ge \ldots \ge \tilde \lambda_n$ be the eigenvalues of $\hat{\bv A}$.
Algorithm~\ref{alg:genNystr\"{o}m} outputs the pairs corresponding to the non-zero eigenvalues of $\hat{\bv A}$
clearing the magnitude threshold (which is $\ge \frac{\epsilon n}{2}$ as governed by the truncation).
Define the bijection $\pi$ by the natural index matching from Weyl's inequality: the output pair
carrying the $j$-th largest positive eigenvalue of $\hat{\bv A}$ is mapped to $\lambda_j(\bv A)$ for $j \in [m_+]$,
and the pair carrying the $j$-th smallest negative eigenvalue is mapped to $\lambda_{n-j+1}(\bv A)$ for $j \in [m_-]$.
For each positive output pair, $\lambda_{\pi(j)}(\bv A) \ge \tilde \lambda_j - \frac{\epsilon n}{2} \ge \frac{\epsilon n}{2} - \frac{\epsilon n}{2} = 0$,
and symmetrically each negative output pair is matched to an eigenvalue $\le 0$; hence the two matched sets
are disjoint and $\pi$ is a bijection onto the multiset $\{\lambda_1(\bv A), \ldots, \lambda_{m_+}(\bv A)\} \cup \{\lambda_{n-m_-+1}(\bv A), \ldots, \lambda_n(\bv A)\}$.

For every output pair $j \in [m]$, this matching directly gives
$|\lambda_{\pi(j)}(\bv A) - \tilde \lambda_j| \le \frac{\epsilon n}{2} \le \epsilon n$.
By the exact-eigenvector property, $\hat{\bv A}\bv v_j = \tilde \lambda_j \bv v_j$, so spectral subadditivity yields:
\[
 \|\bv A\bv v_j - \tilde \lambda_j \bv v_j\|_2 = \|(\bv A - \hat{\bv A})\bv v_j\|_2 \le \|\bv A - \hat{\bv A}\|_2 \|\bv v_j\|_2 \le \frac{\epsilon n}{2}.
\]
Combining this with the eigenvalue bound via the triangle inequality gives:
\[
 \|\bv A\bv v_j - \lambda_{\pi(j)}(\bv A)\bv v_j\|_2 \le \|\bv A\bv v_j - \tilde \lambda_j \bv v_j\|_2 + |\tilde \lambda_j - \lambda_{\pi(j)}(\bv A)| \|\bv v_j\|_2 \le \frac{\epsilon n}{2} + \frac{\epsilon n}{2} = \epsilon n.
\]

\emph{No large eigenvalue is missed.} Suppose $\lambda_i(\bv A) \ge \epsilon n$ for some $i > m_+$.
Its matched value under Weyl's inequality satisfies
$\tilde \lambda \ge \lambda_i(\bv A) - \frac{\epsilon n}{2} \ge \epsilon n - \frac{\epsilon n}{2} = \frac{\epsilon n}{2} > 0$;
in particular $\tilde \lambda > 0$, so it is the $i$-th largest positive eigenvalue of $\hat{\bv A}$
and clears the magnitude threshold $\tau = \frac{\epsilon n}{2}$, meaning its corresponding generalized eigenpair
is always output by the algorithm. Symmetrically, every eigenvalue $\lambda_{n-i+1}(\bv A) \le -\epsilon n$
satisfies $i \le m_-$. Hence every eigenvalue of $\bv A$ outside the matched set has magnitude less than $\epsilon n$.

Both bounds hold for every output pair simultaneously on the single conditioned event. The sample complexity is that of Lemma~\ref{lem:Nystr\"{o}m_spectral}.
\end{proof}

\subsubsection{Squared Column-Norm Sampling.}

We now describe how we can get an
improved error bound of $\epsilon \|\bv A \|_F$
for our approximate eigenvectors
by using squared column-norm sampling
to form the rank-truncated Nystr\"{o}m appproximation.
\begin{lemma}\label{lem:sqcur}
Let $\bv A \in \mathbb{R}^{n \times n}$ be a symmetric matrix.
Let $\bv A'$ be the matrix with zeroed out entries as defined in
Definition~\ref{def:zeroing}.
For $s \ge \frac{c \log^4 n}{\epsilon^2 \delta}\log \frac{1}{\epsilon\delta}$, if
$\frac{s\|\bv A_{i,:} \|^2_2}{\|\bv A \|^2_F} \leq 1$ for all $i \in [n]$, let
$\Sbb$ be the scaled sampling matrix
which samples column $i$ of $\bv A$ with
probability $p_i=\frac{s\|\bv A_{i,:} \|^2_2}{\|\bv A \|^2_F}$,
and rescales it by $\frac{1}{\sqrt{p_i}}$.

Let $\bv C = \bv A' \Sbb$ and
let $\bv U = [\Sbb^T \bv A' \Sbb]_{\tau}^\dagger$
be the truncated pseudoinverse of $\Sbb^T \bv A' \Sbb$,
restricted to eigenvalues with magnitude at least $\tau \geq \frac{\epsilon \|\bv A\|_F}{2}$.
Let $\hat{\bv A}=\bv C \bv U \bv C^T$.
Then, with probability at least $1-\delta$:
\begin{align*}
 \|\bv A - \hat{\bv A}\|_2 \le \epsilon \|\bv A\|_F.
\end{align*}
\end{lemma}
\begin{proof}
By Imported Lemma~\ref{lem:aprime}, $\|\bv A - \bv A'\|_2 \le \epsilon \|\bv A\|_F$, so by the
triangle inequality it suffices to show $\|\bv A' - \hat{\bv A}\|_2 \le \epsilon \|\bv A\|_F$; the
claim then follows after adjusting $\epsilon$ by a constant factor. We prove this by verifying the
steps of the proof
of Lemma~\ref{lem:Nystr\"{o}m_spectral} with $\bv A$ replaced by $\bv A'$, the threshold
$L = \epsilon n$ replaced by $L = \epsilon \|\bv A\|_F$ (so $\tau = \frac{\epsilon\|\bv A\|_F}{2}$), and
each uniform-sampling spectral estimate replaced by its squared column-norm counterpart. Concretely,
split $\bv A' = \bv A'_o + \bv A'_m$ as in Definition~\ref{def:ao-am} with $L = \epsilon\|\bv A\|_F$,
where $\bv A'_o = \bv U'_o \bv \Lambda'_o (\bv U'_o)^T$, and set
$\bv C = \bv C'_o + \bv C'_m$ with $\bv C'_o = \bv A'_o \Sbb$, $\bv C'_m = \bv A'_m \Sbb$, and
$\bv W = \Sbb^T \bv A' \Sbb = \bv W'_o + \bv W'_m$ analogously. The proof of
Lemma~\ref{lem:Nystr\"{o}m_spectral} uses only the following three facts fro Property~\ref{prop:s} about the sampling matrix $\Sbb$, each of which has a squared
column-norm analogue at scale $\epsilon\|\bv A\|_F$:
\begin{enumerate}
 \item \emph{Middle column bound} $\|\bv C'_m\|_2 = \|\bv A'_m \Sbb\|_2 \le \epsilon\|\bv A\|_F$, which
 is Lemma~\ref{lem:as_norm} (replacing Lemma~\ref{lem:as});
 \item \emph{Middle submatrix bound} $\|\bv W'_m\|_2 = \|\Sbb^T \bv A'_m \Sbb\|_2 \le \epsilon\|\bv A\|_F$,
 which is Imported Lemma~\ref{lem:ao-sampling} (replacing Imported Lemma~\ref{lem:am});
 \item \emph{Subspace embedding} for $\bv U'_o$: $\Sbb$ is a constant-factor subspace embedding, so
 $\sigma_{\min}(\Sbb^T \bv U'_o) = \Omega(1)$ and $\|(\Sbb^T \bv U'_o)^\dagger\|_2 = O(1)$, which is
 Imported Lemma~\ref{lem:subspace-embedding-norm} (replacing Imported Lemma~\ref{lem:subspace-embedding-bounded});
\end{enumerate}
With these substitutions, the algebraic identities and the term-by-term bounds on
$\bv T_1, \bv T_2, \bv T_3, \bv T_4$ in the proof of Lemma~\ref{lem:Nystr\"{o}m_spectral} go through
unchanged (every occurrence of $O(\epsilon n)$ becomes $O(\epsilon\|\bv A\|_F)$), yielding
\[
 \|\bv A' - \hat{\bv A}\|_2 \le \|\bv A'_m\|_2 + \|\bv A'_o - \bv T_1\|_2 + \|\bv T_2\|_2 + \|\bv T_3\|_2 + \|\bv T_4\|_2 \le O(\epsilon\|\bv A\|_F).
\]
Adjusting $\epsilon$ by constant factors and combining with $\|\bv A - \bv A'\|_2 \le \epsilon\|\bv A\|_F$
gives $\|\bv A - \hat{\bv A}\|_2 \le \epsilon\|\bv A\|_F$.
\end{proof}
 The final theorem using squared column-norm sampling is stated below.
\begin{theorem}\label{thm:gen-nystrom-eigvec-norm}
Let $\bv A \in \R^{n \times n}$ be a symmetric matrix,
with eigenvalues $\lambda_1(\bv A) \geq \ldots \geq \lambda_n(\bv A)$,
and let $\epsilon, \delta \in (0,1)$.
Let Algorithm~\ref{alg:nys_squared} output the pairs
$\{(\tilde \lambda_j, \bv v_j)\}_{j=1}^{m}$, with $\bv A$ and $s=frac{c \log^4 n}{\epsilon^2 \delta}\log \frac{1}{\epsilon\delta}$ as inputs (for some sufficiently large constant $c$), where $\tilde \lambda_j \in \R$, $\bv v_j \in \R^n$,
$\|\bv v_j \|_2 = 1$, and let $m_+$ (resp.\ $m_-$) be the number of output pairs with
$\tilde\lambda_j > 0$ (resp.\ $\tilde\lambda_j < 0$), so that $m = m_+ + m_-$. Then,
with probability at least $1-\delta$:
\begin{enumerate}
\item There is a bijection $\pi$ between the output pairs and the $m_+$ largest together with the
$m_-$ smallest eigenvalues of $\bv A$, i.e.\ the set
$\{\lambda_1(\bv A), \ldots, \lambda_{m_+}(\bv A)\} \cup \{\lambda_{n-m_-+1}(\bv A), \ldots, \lambda_n(\bv A)\}$,
such that every output pair $j \in [m]$ satisfies
\begin{align*}
 |\lambda_{\pi(j)}(\bv A)-\tilde \lambda_{j}| \leq \epsilon \| \bv A\|_F \text{, and } \quad
 \|\bv A \bv v_{j}-\lambda_{\pi(j)}(\bv A) \bv v_{j} \|_2 \leq \epsilon \| \bv A\|_F.
\end{align*}
\item Every eigenvalue of $\bv A$ outside this matched multiset has magnitude less than $\epsilon \|\bv A \|_F$.
\end{enumerate}
Moreover, Algorithm~\ref{alg:nys_squared} samples
$s \geq \frac{c \log^4 n}{\epsilon^2 \delta}\log \frac{1}{\epsilon\delta}$ columns
from $\bv A$ in expectation.
\end{theorem}
\begin{proof}
The proof follows the same steps as that of Theorem~\ref{thm:gen-nystrom-eigvec}, using the squared column-norm
spectral bound of Lemma~\ref{lem:sqcur} in place of the uniform one. Let $\bv A'$ be the zeroed matrix
of Definition~\ref{def:zeroing}, let $\bv C = \bv A' \Sbb$ (the matrix $\bv A'_{[:,S]}$ formed by
Algorithm~\ref{alg:nys_squared}), and let
$\hat{\bv A} = \bv C [\Sbb^T \bv A' \Sbb]_{\tau}^{\dagger}\bv C^T$ with truncation threshold
$\tau = \frac{\epsilon\|\bv A\|_F}{2}$. By Lemma~\ref{lem:sqcur} (adjusting the constant $c$ in the sample complexity to achieve error $\frac{\epsilon}{2}$), for
$s \geq \frac{c \log^4 n}{\epsilon^2 \delta}\log \frac{1}{\epsilon\delta}$, we have
$\|\bv A - \hat{\bv A}\|_2 \le \frac{\epsilon\|\bv A\|_F}{2}$ with probability at least $1-\delta$; we condition on
this event, which fixes $\Sbb$ and hence $\hat{\bv A}$.

\emph{The output vectors are exact eigenvectors of $\hat{\bv A}$.} For any solution $(\bv x, \lambda)$ of the generalized eigenvalue problem
$\Sbb^T (\bv A')^2 \Sbb \, \bv x = \lambda \, [\Sbb^T \bv A' \Sbb]_{\tau} \, \bv x$ with
$|\lambda| \geq \frac{\epsilon\|\bv A\|_F}{2}$ (which captures all non-zero eigenvalues retained by the truncation), the output vector $\bv v = \frac{\bv A' \Sbb \bv x}{\|\bv A' \Sbb \bv x\|_2}$
is an exact eigenvector of $\hat{\bv A}$ with eigenvalue $\lambda$, exactly as in
Lemma~\ref{lem:nyseig} (using $\Sbb^T (\bv A')^2 \Sbb = (\bv A'\Sbb)^T(\bv A'\Sbb)$ since $\bv A'$ is
symmetric, and $[\Sbb^T \bv A' \Sbb]_{\tau}^{\dagger}[\Sbb^T \bv A' \Sbb]_{\tau}\bv x = \bv x$ on the large
subspace):
\[
 \hat{\bv A}\bv v = \frac{\bv C[\Sbb^T \bv A' \Sbb]_{\tau}^{\dagger}(\Sbb^T (\bv A')^2 \Sbb)\bv x}{\|\bv A'\Sbb\bv x\|_2}
 = \frac{\lambda\,\bv C[\Sbb^T \bv A' \Sbb]_{\tau}^{\dagger}[\Sbb^T \bv A' \Sbb]_{\tau}\bv x}{\|\bv A'\Sbb\bv x\|_2}
 = \frac{\lambda\,\bv A'\Sbb\bv x}{\|\bv A'\Sbb\bv x\|_2}
 = \lambda\bv v.
\]

\emph{The two bounds.} Since $\bv A$ and $\hat{\bv A}$ are symmetric, Weyl's
inequality~\cite{weyl1912} gives $|\lambda_i(\bv A) - \lambda_i(\hat{\bv A})| \leq \|\bv A - \hat{\bv A}\|_2 \leq \frac{\epsilon\|\bv A\|_F}{2}$
for all $i \in [n]$.

Let $\tilde \lambda_1 \ge \ldots \ge \tilde \lambda_n$ be the eigenvalues of $\hat{\bv A}$.
Algorithm~\ref{alg:nys_squared} outputs the pairs corresponding to the non-zero eigenvalues of $\hat{\bv A}$
clearing the magnitude threshold (which is $\ge \frac{\epsilon\|\bv A\|_F}{2}$ as governed by the truncation).
Define the bijection $\pi$ by the natural index matching from Weyl's inequality: the output pair
carrying the $j$-th largest positive eigenvalue of $\hat{\bv A}$ is mapped to $\lambda_j(\bv A)$ for $j \in [m_+]$,
and the pair carrying the $j$-th smallest negative eigenvalue is mapped to $\lambda_{n-j+1}(\bv A)$ for $j \in [m_-]$.
For each positive output pair, $\lambda_{\pi(j)}(\bv A) \ge \tilde \lambda_j - \frac{\epsilon\|\bv A\|_F}{2} \ge \frac{\epsilon\|\bv A\|_F}{2} - \frac{\epsilon\|\bv A\|_F}{2} = 0$,
and symmetrically each negative output pair is matched to an eigenvalue $\le 0$; hence the two matched sets
are disjoint and $\pi$ is a bijection onto the multiset $\{\lambda_1(\bv A), \ldots, \lambda_{m_+}(\bv A)\} \cup \{\lambda_{n-m_-+1}(\bv A), \ldots, \lambda_n(\bv A)\}$.

For every output pair $j \in [m]$, this matching directly gives
$|\lambda_{\pi(j)}(\bv A) - \tilde \lambda_j| \le \frac{\epsilon\|\bv A\|_F}{2} \le \epsilon\|\bv A\|_F$.
By the exact-eigenvector property, $\hat{\bv A}\bv v_j = \tilde \lambda_j \bv v_j$, so spectral subadditivity yields:
\[
 \|\bv A\bv v_j - \tilde \lambda_j \bv v_j\|_2 = \|(\bv A - \hat{\bv A})\bv v_j\|_2 \le \|\bv A - \hat{\bv A}\|_2 \|\bv v_j\|_2 \le \frac{\epsilon\|\bv A\|_F}{2}.
\]
Combining this with the eigenvalue bound via the triangle inequality gives:
\[
 \|\bv A\bv v_j - \lambda_{\pi(j)}(\bv A)\bv v_j\|_2 \le \|\bv A\bv v_j - \tilde \lambda_j \bv v_j\|_2 + |\tilde \lambda_j - \lambda_{\pi(j)}(\bv A)| \|\bv v_j\|_2 \le \frac{\epsilon\|\bv A\|_F}{2} + \frac{\epsilon\|\bv A\|_F}{2} = \epsilon\|\bv A\|_F.
\]

\emph{No large eigenvalue is missed.} Suppose $\lambda_i(\bv A) \ge \epsilon\|\bv A\|_F$ for some $i > m_+$.
Its matched value under Weyl's inequality satisfies
$\tilde \lambda \ge \lambda_i(\bv A) - \frac{\epsilon\|\bv A\|_F}{2} \ge \epsilon\|\bv A\|_F - \frac{\epsilon\|\bv A\|_F}{2} = \frac{\epsilon\|\bv A\|_F}{2} > 0$;
in particular $\tilde \lambda > 0$, so it is the $i$-th largest positive eigenvalue of $\hat{\bv A}$
and clears the magnitude threshold $\tau = \frac{\epsilon\|\bv A\|_F}{2}$, meaning its corresponding generalized eigenpair
is always output by the algorithm. Symmetrically, every eigenvalue $\lambda_{n-i+1}(\bv A) \le -\epsilon\|\bv A\|_F$
satisfies $i \le m_-$. Hence every eigenvalue of $\bv A$ outside the matched set has magnitude less than $\epsilon\|\bv A\|_F$.

Both bounds hold for every output pair simultaneously on the single conditioned event. Finally, Lemma~\ref{lem:sqcur} and the analysis above assume the sampling probabilities satisfy
$p_i = \frac{s\|\bv A_{i,:}\|_2^2}{\|\bv A\|_F^2} \le 1$ for all $i$. This assumption is removed by the
row-duplication (matrix inflation) argument of Theorem~\ref{thm:sqnorm1}: replacing each row and
column of $\bv A'$ by $s$ copies scaled by $\frac{1}{\sqrt{s}}$ produces an inflated matrix with the
same nonzero eigenvalues and per-row sampling probability at most $1$, to which the above applies.
(Equivalently, since Algorithm~\ref{alg:nys_squared} draws
$n_i \sim \mathrm{Binomial}(s, \|\bv A_{i,:}\|_2^2/\|\bv A\|_F^2)$ copies of each column, one may also
use the Poisson sampling scheme of Algorithm~\ref{alg:eigvec-rowsamp} in its place, as in
Section~\ref{sec:qram}.) The sample complexity is given by Lemma~\ref{lem:sqcur}.
\end{proof}

\begin{algorithm}[H]
\caption{Rank-Truncated Nystr\"om Using Squared Column-Norm Sampling}
\label{alg:nys_squared}
\begin{algorithmic}[1]
\State {\bfseries Input:} Symmetric $\bv A \in \mathbb{R}^{n\times n}$,
row norms $\{\|\bv A_{i,:}\|_2\}_{i=1}^n$,
accuracy $\epsilon \in (0,1)$, expected sample size $s$.
\State \textbf{(Sampling)} For each $i \in [n]$,
sample $n_i \sim \mathrm{Binomial}\!\left(s,\,
\frac{\|\bv A_{i,:}\|_2^2}{\|\bv A\|_F^2}\right)$
copies of column $i$ and rescale those columns by $\frac{1}{\sqrt{p_i}}$
where $p_i=\frac{s\| \bv A_{i,:}\|^2_2}{\| \bv A\|^2_F}$. Let $S $ be the ordered list of sampled indices.

\State \textbf{(Scaling and Zeroing out)} Form the matrices $\bv A'_{:,S} \in \R^{n \times |S|}$ and $\bv A'_{S,S} \in \R^{|S| \times |S|}$ exactly as in steps~\ref{step:zero} and~\ref{step:zero1} of Algorithm~\ref{alg:eigvec-rowsamp} by scaling and zeroing out the entries of $\bv A$ (for appropriate constant $c>0$).

\State Find all eigenvectors $\bv x \in \R^s$
corresponding to the eigenvalues $\lambda$
(with magnitude at least $ \epsilon \| \bv A\|_F$)
by solving the generalized eigenvalue problem:
\[
 (\bv A'_{:,S})^T\bv A'_{:,S} \, \bv x =
 \lambda \, [\bv A'_{S,S}]_{\tau} \, \bv x,
\]
\textit{(equivalent to solving $\Sbb^T (\bv A')^2 \Sbb \bv x = \lambda [\Sbb^T \bv A' \Sbb]_{\tau} \bv x$,
for any $\bv A$ with $\frac{s\|\bv A_{i,:} \|^2_2}{\| \bv A \|_F^2}$ and
$\bv A'$ is the zeroed out version of $\bv A$ as defined in Definition~\ref{def:zeroing})},
where $[\bv A'_{S,S}]_{\tau}$ is the truncated eigendecomposition
of the matrix $\bv A'_{S,S} $ restricted to the eigenvalues
with magnitude at least $\tau \geq \frac{\epsilon \|\bv A \|_F}{2}$.
\State Compute the approximate eigenvectors $\bv v_i$
corresponding to the eigenvectors $\bv x_i$ found in the previous step as:
\[
 \bv v_i = \frac{\left(\bv A'_{:,S} \right)\bv x_i}{\left\|\left(\bv A'_{:,S} \right) \bv x_i \right\|_2}.
\]
\State \Return all approximate eigenvector, eigenvalue pairs
$(\bv v_i, \lambda_i)$, found in the previous steps.
\end{algorithmic}
\end{algorithm}

\paragraph{Acknowledgement} CM and RB were partially supported by NSF grants AF-2427362 and AF-2046235. RB thanks Andrew Horning and David Woodruff for illuminating discussions on the techniques used in the paper.

\paragraph{Statement on AI Use.} LLMs were used to polish the writing of the paper and correct mistakes and typos. The Poissonization approach was suggested by an LLM.

\newpage
\bibliography{refs}

\end{document}